\documentclass[11pt,a4paper]{article}

\usepackage[a4paper, margin=0.85in]{geometry}
\usepackage[utf8]{inputenc}
\usepackage[OT1]{fontenc}
\usepackage{amsmath}
\usepackage{amsfonts}
\usepackage{mathabx}
\usepackage{amssymb}
\usepackage{graphicx}
\usepackage{bbm} 
\usepackage{amsthm}
\newtheorem{theorem}{Theorem}[section]
\newtheorem{proposition}[theorem]{Proposition}

\newtheorem{assumption}{Assumption}[section]

\usepackage{booktabs}
\usepackage{sectsty}

\usepackage{natbib}
\setcitestyle{sort&compress}
\bibpunct{(}{)}{,}{a}{,}{,}
\usepackage{subfigure}

\usepackage[nodisplayskipstretch]{setspace}
\usepackage{caption}
\usepackage[colorlinks,linkcolor=black,anchorcolor=black,citecolor=black]{hyperref}

\usepackage{titlesec}
\usepackage[title]{appendix}

\usepackage{mathtools}
\mathtoolsset{
	showonlyrefs,
	mathic 
}

\newcommand{\expect}{\mathbf{E}}

\newcommand{\vecomega}{\boldsymbol{\omega}}

\newcommand{\argmax}{\operatornamewithlimits{argmax}}

\newenvironment{pfof}[1]{\vspace{1ex}\noindent{\bf Proof of #1}\hspace{0.5em}}
{\qed\vspace{1ex}}
\usepackage[dvipsnames]{xcolor}

\date{}

\begin{document}
	\title{Optimal Design of Automated Market Makers on Decentralized Exchanges}
	
	\author{Xue Dong He\thanks{Department of Systems Engineering and Engineering Management, The Chinese University of Hong Kong, Shatin, N.T., Hong Kong SAR. Email: xdhe@se.cuhk.edu.hk} \and Chen Yang\thanks{Department of Systems Engineering and Engineering Management, The Chinese University of Hong Kong, Shatin, N.T., Hong Kong SAR. Email: cyang@se.cuhk.edu.hk. Corresponding Author.} \and Yutian Zhou\thanks{Department of Systems Engineering and Engineering Management, The Chinese University of Hong Kong, Shatin, N.T., Hong Kong SAR. Email: YutianZhouSEEM@link.cuhk.edu.hk}}
	
	\maketitle
	
	\begin{abstract}
		Automated market makers are a popular mechanism used on decentralized exchange, through which users trade assets with each other directly and automatically through a liquidity pool and a fixed pricing function. The liquidity provider contributes to the liquidity pool by supplying assets to the pool, and in return, they earn trading fees from investors who trade in the pool. We propose a model of optimal liquidity provision in which a risk-averse liquidity provider decides the amount of wealth she would invest in the decentralized market to provide liquidity in a two-asset pool, trade in a centralized market, and consume in multiple periods. We derive the liquidity provider's optimal strategy and the optimal design of the automated market maker that maximizes the liquidity provider's utility. We find that the optimal unit trading fee increases in the volatility of the fundamental exchange rate of the two assets. We also find that the optimal pricing function is chosen to make the asset allocation in the liquidity pool efficient for the liquidity provider.
	\end{abstract}
	
	\noindent{\textbf{Conflicts of Interest:} The authors declare that there are no conflicts of interest regarding this research.}\\[1em]
	\noindent{\textbf{Data Availability:} The data supporting the findings of this paper is available on Binance Market Data (\url{https://data.binance.vision/})}\\[1em]
	\noindent{\textbf{Funding Information:} He acknowledges the support from the Hong Kong Research Grant Council GRF (14218022). Yang acknowledges the support from the Hong Kong Research Grant Council ECS (24207621) and GRF (14207723).}
	
	\clearpage
	\section{Introduction}
	In recent years, Decentralized Finance (DeFi) has attracted much practical and academic attention due to its characteristic of not relying on central authority institutions. Playing a critical role in DeFi, decentralized exchanges (DEXs) allow for peer-to-peer transactions of crypto-assets without intermediaries, and its potential advantages are exemplified by the collapse of one of the biggest centralized exchanges (CEXs), FTX, in November 2022. DEXs have been an important part of the cryptocurrency exchange in practice. Indeed, in 2023, the proportion of DEX to CEX spot trade value mostly remains in the range of 10$\%$ to 20$\%$.\footnote{See \url{https://www.theblock.co/data/decentralized-finance/dex-non-custodial}.} \label{footnote:data}
	
	The most dominant mechanism adopted by DEXs is the Automated Market Maker (AMM). In contrast to the traditional trading mechanisms that are based on limit order books and widely used on CEXs, AMMs operate based on specific algorithms implemented via smart contracts, where the asset prices for trades on DEX are determined by the amount of deposits of these assets on the DEX through a pricing function. A group of participants, named liquidity providers (LPs), deposit their assets on the DEX to form a liquidity pool and earn trading fees from trades in the pool in return. LPs play a similar role to market makers on CEX in that both provide liquidity in the market. A higher level of liquidity provided by LPs benefits the investors by reducing the trading cost, and a higher trading volume benefits the LPs with higher trading fees. For details of the AMM mechanism, see \citet{FabiPrat2023:TheEconomicsOfCFMM}, \citet{AngerisEtal2019:AnAnalysisOfUniswap}, \citet{AngerisChitraEvans2020:WhenDoesTheTailWag}, as well as an introduction to AMM in Section \ref{sect:AMMintro}.
	
	There has been a growing literature on the study of AMMs in recent years; see Section \ref{subse:Literature} for a review. However, the following important yet challenging questions have not been fully answered. First, how do we design an AMM that encourages liquidity provision? Second, how should such a design adapt to the return distributions of the assets in the pool? In the context of DEXs with two assets in the liquidity pool, we tackle these questions using a framework with the following three building blocks. As the base block, we model investors' optimal trading decisions based on the liquidity pool on DEX and asset prices on CEX. The investors include arbitrageurs who exploit the difference between the exchange rates on DEX and CEX and liquidity traders who trade to maximize their profits based on their heterogeneous beliefs on the asset values. As the second building block, we model a representative LP's optimal dynamic liquidity provision, where the return from the liquidity provision depends on investors' decisions on DEX and the change of asset values on CEX. The LP can also invest in a risk-free asset and in the same assets as in the liquidity pool, but on CEX. Thus, the LP's decision is the allocation of her wealth to the liquidity provision on DEX, investment on CEX, and the risk-free asset, and her objective is to maximize the discounted sum of the expected utility of consumption in an infinite horizon. As the final building block, we focus on the optimal design of the AMM by choosing the unit trading fee and pricing function to maximize the LP's value.
	
	The main contributions of the paper are threefold. First, we are the first to consider LPs' optimal liquidity provision with a combination of the following features: (i) risk-averse LPs, (ii) outside investment opportunities, and (iii) dynamic settings. Most studies in the literature on AMM assume risk-neutral LPs. However, assets in a liquidity pool are typically cryptocurrencies, whose prices are usually highly volatile, and liquidity provision and withdrawal in the pool need to be validated on blockchain and thus cannot be conducted in a high frequency. Therefore, liquidity provision can entail a large amount of risk, and it is reasonable to take LPs' risk aversion into account.
	
	Most studies of AMM on DEX in the literature either isolate LPs' liquidity provision on DEX from other investment opportunities or implicitly assume that LPs' investment on CEX has to follow the same allocation as in the liquidity pool on DEX. For example, impermanent loss and loss-versus-rebalancing, which capture the LPs' losses to arbitrageurs in the liquidity pool and are commonly discussed in the literature, are defined to be the difference between liquidity pool value and the value of the static portfolio with the same initial holding of assets as in the liquidity pool or the replicating portfolio with the same holding of assets as in the liquidity pool at all time; see, for instance, \citet{EngelHerlihy2022:PresentationAndPublication}, \citet{AngerisChitra2020:ImprovedPriceOracles}, \citet{MilionisMoallemiRoughgardenZhang2024:AutomatedMarketMaking}, and \citet{CarteaDrissiMonga2023:PredictableLosses}. Thus, the implicit assumption behind impermanent loss and loss-versus-rebalancing is that the LP compares the liquidity pool on DEX with the same portfolio of assets on CEX. In practice, LPs can hold any portfolio of assets on CEX, while the allocation of assets in the liquidity pool is determined by the pricing function and the exchange rate of assets in the pool, the latter closely following the exchange rate on CEX. Thus, the allocation of assets in the liquidity pool can be inefficient because it can be different from the optimal allocation of assets if the LPs invest on CEX only, e.g., the Merton's allocation if the LPs maximize the expected utility of their returns. This inefficiency leads to an {\em opportunity cost} when LPs provide liquidity on DEX. We are the first in the literature to consider this opportunity cost, and we will show that it plays a critical role in the optimal design of the AMM.
	
	While most studies of liquidity provision consider a single-period setting, we consider a multi-period setting and thus study the LP's dynamic liquidity provision. For a detailed comparison of our model with the existing literature, see Section \ref{subse:Literature}.
	
	Our second contribution is to provide the optimal design of the unit trading fee and pricing function of AMM by maximizing the LP's value, and show how the optimal design depends on the return distributions of the assets in the pool. Using numerical studies, we find that the optimal unit trading fee is increasing in the volatility of the fundamental exchange rate of the two assets but is not sensitive to the mean and fundamental volatility of each individual asset in the pool, holding the fundamental exchange rate volatility constant. This result can be explained as follows: By providing liquidity in the pool, the LP (i) earns trading fees from the liquidity traders, (ii) suffers a loss to the arbitrageurs, which is the impermanent loss or loss-versus-rebalancing net the trading fee, and (iii) suffers an opportunity cost. The loss to the arbitrageurs occurs when the asset prices on CEX change, making the prices on DEX stale. In this case, the arbitrageurs take advantage of the mispricing and trade against the liquidity pool, thereby leading to a loss borne by the LP. While the trading fee has little impact on the opportunity cost, it affects both the LP's fee revenue and her loss to the arbitrageurs. With a higher unit trading fee, the LP's loss to the arbitrageurs becomes smaller. The trading volume of the liquidity traders decreases in the unit trading fee, so the LP's fee revenue from the liquidity traders can first increase and then decrease in the unit trading fee. Thus, the optimal unit trading fee exists. When the fundamental exchange rate volatility becomes higher, the mispricing due to the asset price change becomes larger, so the LP can potentially suffer a larger loss to arbitrageurs. Thus, in designing the unit trading fee, the LP would focus more on discouraging the arbitrageurs than earning fee revenues from the liquidity traders. This implies that the optimal unit trading fee becomes larger. Because the mean and volatility of each individual asset, when holding the fundamental exchange rate volatility constant, have little impact on the size of mispricing, they do not have a significant impact on the optimal unit trading fee. Using price data from Binance and fee selection data from Uniswap v3, we find that our model implication is consistent with the fee choice in the market.

	We design the optimal pricing function among the class of constant geometric mean market makers (CGMMM), a single-parameter pricing function commonly used in the market. On the one hand, the pricing function determines the allocation of the assets in the liquidity pool. This allocation may be inefficient for the LP, creating an opportunity cost. On the other hand, the pricing function determines price slippage, which stands for the cost of trading due to the price impact resulting from the pricing function. Price slippage discourages trading and thus is detrimental to the LP's fee revenues from the liquidity traders, but it reduces the LP's loss to the arbitrageurs. The single parameter of the CGMMM determines the amount of the LP's fee revenue from the liquidity traders net her loss to the arbitrageurs. Thus, the LP faces a tradeoff between choosing a parameter value to minimize the opportunity cost and choosing a parameter value to maximize her fee revenue net the loss to the arbitrageurs. We find that the opportunity cost is the dominant factor: The LP would choose the pricing function to make the asset allocation in the DEX liquidity pool efficient. This result still holds even if we allow the LP to short-sell assets on CEX. The LP would still choose the pricing function to have an efficient asset allocation in the liquidity pool, and she would short-sell the assets on CEX to fund her investment on DEX.

	Our third contribution is to prove the existence and uniqueness of the solution of the dynamic programming equation for the LP's decision problem. Our problem features a special structure where the LP's investment decision occurs at a frequency higher than the utility she derives from her consumption.  With the infinite horizon setting, the solution of the dynamic programming equation becomes the fixed point of a certain operator. Similar structures also appear in the recent studies \citet{KeppoReppenSoner2021:DiscreteDividendPayments} and \citet{DaiKouSonerYang2023:LeveragedETF}. We prove the existence and uniqueness of the fixed point using the concave operator theory because the dynamic programming operator in our problem is not a contraction mapping.
	
	The remainder of the paper is organized as follows: In Section \ref{subse:Literature}, we review the literature on AMM. In Section \ref{sect:AMMintro}, we give an introduction to AMM. In Section \ref{sect:mainresult}, we propose our model and derive the main theoretical results. In Section \ref{sect:numerics}, we find the optimal design of the liquidity pool by numerical studies. Section \ref{sect:conclusion} concludes. All proofs are in the Appendix.

	\subsection{Related Literature}\label{subse:Literature}
	
	Our modeling of AMM participants in every single period is similar to the settings in \citet{CapponiJia2021:TheAdoptionOfBlockchainBased}, \citet{hasbrouck2026need},  \citet{LeharParlour2023:DecentralizedExchange}, \citet{Aoyagi2022:LiquidityProvision}, and \citet{AoyagiIto2024:CoexistingExchangePlatforms}. They, however, consider single-period settings and have a different focus from us. \citet{CapponiJia2021:TheAdoptionOfBlockchainBased} consider a group of risk-neutral LPs who decide the amount of liquidity provision to maximize the expected profit in one period. They find that the LPs provide liquidity if and only if the fundamental exchange rate volatility is sufficiently high. The authors then find the best mixture of the linear pricing function and the constant product market maker (CPMM) that maximizes the social welfare, defined as the ex-ante expected payoff of LPs, arbitrageurs, and investors, plus the gas fees earned by validators. \citet{hasbrouck2026need} consider a group of risk-neutral investors who can choose between liquidly provision on DEX with AMM and an alternative outside investment opportunity and a group of investors who can choose to trade on DEX and CEX. The authors show that a higher unit trading fee on DEX can increase liquidity provision and thus increase trading volume on DEX due to decreasing price impact. \citet{LeharParlour2023:DecentralizedExchange} consider risk-neutral LPs who maximize their expected net profit by providing liquidity on DEX with AMM. The authors show that the amount of liquidity provision is decreasing in the volatility of the exchange rate of the assets in the pool and increasing in the amount of uninformed trading, and they conduct empirical studies supporting their model implication. \citet{Aoyagi2022:LiquidityProvision} considers the competition of risk-neutral LPs and study both competitive and strategic liquidity provision strategies. \citet{AoyagiIto2024:CoexistingExchangePlatforms} study the coexistence of two exchange platforms: a CEX with a limit order book and a DEX with AMM. Both liquidity traders and arbitrageurs can choose the trading venue. The LP's objective is to maximize a weighted average of the impermanent loss and the fee revenue minus a constant exogenous opportunity cost. The authors show that the equilibrium amount of liquidity provision, which makes the LP break even, is hump-shaped in the exchange rate volatility.

	Almost all studies of liquidity provision do not take into account LPs' opportunity costs. In a simulation study in \citet{AngerisEtal2019:AnAnalysisOfUniswap}, the authors assume that LPs can choose any portfolio in a centralized market. Compared to their work, we formally propose the notion of opportunity cost and solve the LP's liquidity provision strategy exactly using dynamic programming.  \citet{EvansAngerisChitra2021:OptimalFees} consider an LP who wants to track a given target portfolio weight. If we set this target to be the optimal portfolio that the LP would have held if she could invest on CEX only, then the objective of tracking this target portfolio can be understood as the LP's incentive to minimize the opportunity cost.

	There are a few studies of the impact of trading fees on liquidity provision. For example, \citet{hasbrouck2026need} show that a higher unit trading fee increases liquidity provision and that the LP's value is first increasing and then decreasing in the fee level.  There are, however, few studies of optimal trading fees. \citet{BergaultEtal2024:AutomatedMarketMakers} consider an LP who maximizes the expected value of fee revenue plus a mean-variance value of the pool value at the end of a continuous-time horizon, but they do not directly consider the LP's loss to arbitrageurs. They assume that the LP can choose different fee levels for different times, trading sizes, and asset deposits in the pool and derive the optimal fee process using a stochastic control approach. Compared to their work, we consider not only the LP's fee revenue but also her loss to arbitrageurs and opportunity cost, and we consider a constant unit trading fee as in the current market practice. \citet{MilionisEtal2023:AMMFees} study the impact of trading fees on the LP's loss to arbitrageurs on AMM. They find that the fees scale down the loss by a fraction of the times that arbitrageurs find a profitable trade.
	
	The literature has observed that the curvature of the pricing function determines the price slippage and that a higher level of price slippage reduces both LPs' fee revenue from liquidity traders and LPs' losses to arbitrageurs. See, for instance, \citet{AngerisChitraEvans2020:WhenDoesTheTailWag}, \citet{CapponiJia2021:TheAdoptionOfBlockchainBased}, and \citet{FabiPrat2023:TheEconomicsOfCFMM}. Thus, LPs face a trade-off between high price slippage, which implies a smaller amount of loss to arbitrageurs, and lower price slippage, which implies more fee revenues from liquidity traders. There are a few studies of optimally designing the pricing function in the literature.  \citet{CapponiJia2021:TheAdoptionOfBlockchainBased} derive the optimal mixture of the constant sum market maker and the CPMM to maximize the aggregate welfare of LPs, arbitrageurs, investors, and blockchain validators.  \citet{LiptonSepp2021:AutomatedMarketMaking} use simulation to compute the price slippage and the LP's profit and loss defined as the fee revenue net the impermanent loss on AMM with a weighted average of the CPMM and the constant sum market maker. They find that both the price slippage and the LP's profit and loss become higher when the pricing function moves closer to the CPMM. \citet{GoyalEtal2023:FindingTheRightCurve} derive the optimal pricing function that minimizes the expected inefficiency of AMM, defined as the probability of a trade with a fixed size and tolerance of price slippage not being executed on AMM. \citet{CapponiChen2024:OptimalPricingFunction} find the optimal pricing function that maximizes the expected utility of the LP's payoff under the same single-period setting as in \citet{CapponiJia2021:TheAdoptionOfBlockchainBased}. \citet{MilionisMoallemiRoughgarden2023:ComplexityApproximation} and  \citet{MilionisMoallemiRoughgarden2024:AMyersonianFramework} describe exchange mechanisms, including DEX with AMM, by demand curves submitted by LPs to exchanges. The authors then study the optimal design of demand curves. \citet{AdamsEtal2024:amAMM} propose a new AMM design, the auction-managed AMM, that can benefit liquidity providers by reducing losses to arbitrageurs and increasing revenue from liquidity traders. The authors show that under certain assumptions, this AMM has higher liquidity in equilibrium than the standard AMM. There are also some studies of providing axiomatic foundations for various pricing functions; see, for instance, \citet{SchlegelEtal2023:AxiomsForConstantFunction} and \citet{FrongilloEtal2024:AnAxiomatic}.

	There are some existing studies on liquidity provisions on DEX with AMM. \citet{ChenDengFuZou2023:LiquidityProvision} consider static liquidity provision and show that the LPs provide liquidity if and only if the exchange rate volatility falls into an interval. \citet{FanEtal2022:DifferentialLiquidity} and \citet{HeimbachEtal2022:RisksAndReturns} consider static liquidity provision in Uniswap v3 \citep{AdamsEtal2021:UniswapV3}. By contrast, we consider dynamic liquidity provision and the optimal design of the liquidity pool. \citet{FanEtal2023:StrategicLiquidityProvision},  \citet{BarOnMansour2023:UniswapLiquidityProvision}, and \citet{CarteaDrissiMonga2023:DecentralisedFinance} consider dynamic liquidity provision on Uniswap v3, but they do not consider optimal design of the liquidity pool. \citet{BayraktarCohenNellis2024:DEX} study the competition between multiple LPs through mean-field games.

	\citet{AngerisEvansChitra2023:Replicating} study the portfolio value of the liquidity pool on AMM. The authors show that any concave, nonnegative, nondecreasing, and 1-homogeneous portfolio value function can replicated by some pricing function on AMM. \citet{FabiPrat2023:TheEconomicsOfCFMM} use results in the consumer theory to study the properties of CFMMs. In particular, they show that the cost of trading, captured by the slippage and proportional to the curvature of the bonding curve, is inversely proportional to the divergent loss of the LP. \citet{SabateVidalesSiska2022:TheCaseForVariableFees} study the optimal trading strategy of an investor on DEX using reinforcement learning algorithms. There are also studies of trading in multiple pools and pool competition; see, for instance, \citet{Fritsch2021:ANoteOnOptimalFees}, \citet{FritschKaserWattenhofer2022:TheEconomicsOfAMM}, and \citet{ZhangFanFangCai2022:EconomicAnalysis}. \citet{KrishnamachariEtal2021:DynamicCurves} and \citet{PortTiruviluamala2022:MixingConstantSum} discuss dynamic pricing functions. \citet{CarteaDrissiMonga2024:DecentralizedFinanceAMMExecutionSpeculation} study an optimal liquidation problem for an investor who wants to exchange a large position of an asset for another asset on DEX with AMM in a fixed period of time. The authors consider a continuous-time setting and formulate the problem as a stochastic control problem similar to the optimal liquidation problem on a centralized market \citep{AlmgrenChriss2001:OptimalExecution}. \citet{CarteaDrissiMonga2023:ExecutionAndStatisticalArbitrage} extend \citet{CarteaDrissiMonga2024:DecentralizedFinanceAMMExecutionSpeculation} to the case of liquidating multiple assets. \citet{CohenEtal2023:InefficiencyOfCFMs} consider the trading fees gained by the LPs on DEX in a continuous-time model and derive an upper bound for the fees. This upper bound resembles the loss-versus-rebalance proposed by \citet{MilionisMoallemiRoughgardenZhang2024:AutomatedMarketMaking}. \citet{AngerisEtal2022:ConstantFunctionMarketMakers} discuss trading and liquidity provision in multiple-asset pools in a single-period setting. \citet{CarteaEtal2024:AutomatedMarketMakers} propose the so-called arithmetic liquidity pool and study impermanent loss and optimal action of the liquidity provider in this pool. On decentralized lending, leasing, and renting protocol Olympus, by controlling the allocation to bonds and the inflation rate, \citet{ChitraEtal2022:DeFiLiquidityManagment} design a liquidity provision framework that improves protocol efficiency.

	\section{An Introduction to the Automated Market Maker}\label{sect:AMMintro}
	The AMM has become the dominant trading mechanism on DEXs in recent years. Implementable via smart contracts, AMMs provide continuous liquidity over the exchange and conduct trading directly and automatically between participants. One key functionality of AMMs is to determine the marginal exchange rate over the DEX. To achieve this, AMMs adopt the pricing function, typically a function $F$ of the amounts of asset deposits in the DEX liquidity pool.
	
	To illustrate, in the following, we consider a DEX with two assets, $A$ and $B$, and denote the current deposits of these two assets in the liquidity pool as $y^A$ and $y^B$, respectively. Furthermore, denote the amount of asset $A$ and $B$ \emph{withdrawn} from the liquidity pool as $D^A$ and $D^B$, respectively. Consequently, $-D^i$ means the amount of asset $i$ deposited to the liquidity pool.
	
	\subsection{Trading on DEX with AMM}
	We first consider trading on the DEX with AMM. Suppose an investor wants to exchange asset $A$ for an amount $D^B>0$ of asset $B$. In order to acquire this amount of asset $B$, the investor needs to deposit an amount $-D^A>0$ of asset $A$ (exclusive of fee) to the liquidity pool in exchange, where $D^A$ is determined by the pricing function $F$ through
	\begin{align}\label{AMMrule}
		F(y^A-D^A,y^B-D^B)= F(y^A,y^B).
	\end{align}
	
	There are many types of AMMs with different pricing functions in practice. The most popular one is the CPMM, whose pricing function is
	\begin{align*}
		F(y^A,y^B) = y^Ay^B.
	\end{align*}
	CPMM is used in Uniswap (\url{uniswap.org}), which accounts for over half of the total volume in the DEX market\footnote{See \url{https://www.theblock.co/data/decentralized-finance/dex-non-custodial}.}. Another type of AMM is the Constant Geometric Mean Market Maker (CGMMM) in Balancer (\url{balancer.fi}), with the pricing function
	\begin{align}\label{eq:CGMMM}
		F(y^A,y^B) = (y^A)^{\eta}(y^B)^{1-\eta}
	\end{align}
	for some $\eta\in (0,1)$. Here, $\eta$ and $1-\eta$ represent the weights of assets $A$ and $B$, respectively, and the special case of $\eta=1/2$ makes the CGMMM equivalent to the CPMM. A comprehensive overviews of AMMs can be found in, e.g., \citet{XuEtal2023:SoKDeX}, \citet{Mohan2022:AMM}, and \citet{AngerisChitra2020:ImprovedPriceOracles}.
	
	Trading on the DEX incurs trading fees. As a result, following the previous example of acquiring $D^B$ amount of asset $B$, the actual amount of asset $A$ that the investor needs to deposit is $-(1+f)D^A$, which is higher than $-D^A$ as calculated from \eqref{AMMrule}. Here, $f>0$ represents the unit trading fee. As a result, the amounts of the two assets in the pool after this trade become $y^A-(1+f)D^A$ and $y^B - D^B$, respectively. In general, with $D^i$ representing the amount of asset $i$ withdrawn from the liquidity pool and $-D^i$ representing the amount deposited into the pool, the post-trade deposits of the two assets in the liquidity pool are given by
	\begin{align}\label{change_deposit}
		y^A-(1+f\mathbf{1}_{D^A<0})D^A,~~y^B-(1+f\mathbf{1}_{D^B<0})D^B.
	\end{align}
	Note that the value of the pricing function $F$ will increase after the trade due to the trading fee; that is,
		$
		F(y^A-(1+f\mathbf{1}_{D^A<0})D^A,y^B - (1+f\mathbf{1}_{D^B<0})D^B)> F(y^A-D^A,y^B-D^B)=F(y^A,y^B).
		$
	It is worth mentioning that we follow the protocol in Uniswap V2 and assume that the trading fees are deposited in the pool. By contrast, in some existing studies of AMM in the literature, e.g., in \citet{LeharParlour2023:DecentralizedExchange} and \citet{MilionisMoallemiRoughgardenZhang2024:AutomatedMarketMaking}, consistent with the protocol in Uniswap v3, the trading fees are assumed to be paid to the LPs directly and thus do not change the value of the pricing function.

	\subsection{Liquidity Provision on DEX with AMM}
	Besides investors exchanging one asset for the other, another type of participants on AMM is the LP. The LPs provide liquidity to the DEX by depositing both assets into the liquidity pool and remove liquidity by withdrawing both assets from the pool. The LPs' investment return on DEX is determined by the change of deposits in the liquidity pool due to the trading activities conducted by investors and change of fundamental values of the assets.
	
	Denote the LP withdrawal (deposit if negative) amounts of asset $A$ and $B$ as $D^A$ and $D^B$, respectively, such that the subsequent amounts of asset deposits in the pool become $y^A-D^A$ and $ y^B-D^B$. Consequently, the value of pricing function $F$ also changes along with deposit or withdrawal:
		$
		F(y^A-D^A,y^B-D^B)\neq F(y^A,y^B).
		$
	In practice, the amount $D^A$ and $D^B$ may be subject to constraints in various exchanges. A common requirement used in the most popular DEX Uniswap is that the deposit or withdrawal has to keep the ratio between the asset amounts in the pool unchanged, i.e.,
		$
		D^A/D^B = y^A/y^B.
		$
	
	Both LPs and investors are essential to the health of the DEX with AMM. LPs inject and maintain the liquidity on DEX, and a high level of liquidity benefits and attracts investors. On the other hand, a high level of activity from the investors increases the trading fees, which in turn contributes to the pool and encourages LPs' participation. Therefore, how to design a good DEX with AMM, including the choice of unit trading fee and pricing function, is an important yet challenging question, which we will subsequently explore.

	\section{The Model and Main Results}\label{sect:mainresult}
	In this section, we describe LP's optimal investment and consumption model and our main theoretical results.
	
	\subsection{The Market}\label{sect:sub:market}
	We consider risky assets $A$ and $B$, and a risk-free asset. The risky assets are traded both on a DEX where trading follows a rule dictated by an AMM, and on CEX. Following \citet{hasbrouck2026need}, we assume that CEX prices the risky assets at their fundamental values.
	
	The market participants include a representative liquidity provider and two types of investors: liquidity traders and arbitrageurs. The LP allocates her wealth in asset $A$ and asset $B$ on both the CEX and the DEX, thereby providing liquidity to the DEX, and gains intertemporal utility over consumption. Both the liquidity traders and arbitrageurs maximize their profits on the DEX by exchanging one risky asset for the other, subject to the trading fees. The arbitrageurs and the LP trade based on perfect information about the true exchange rate between the risky assets on CEX, while the liquidity traders trade based on heterogeneous beliefs on the exchange rate, which is a noisy version of the true exchange rate. The detailed objectives of the LP and investors will be described subsequently. Note that liquidity traders are commonly assumed in recent studies of AMM; see, for instance, \citet{MilionisMoallemiRoughgardenZhang2024:AutomatedMarketMaking}, \citet{CapponiJia2021:TheAdoptionOfBlockchainBased},  \citet{hasbrouck2026need}, and \citet{LeharParlour2023:DecentralizedExchange}.

	\begin{figure}[htbp!]
		\centering
		\includegraphics[width = 1\textwidth]{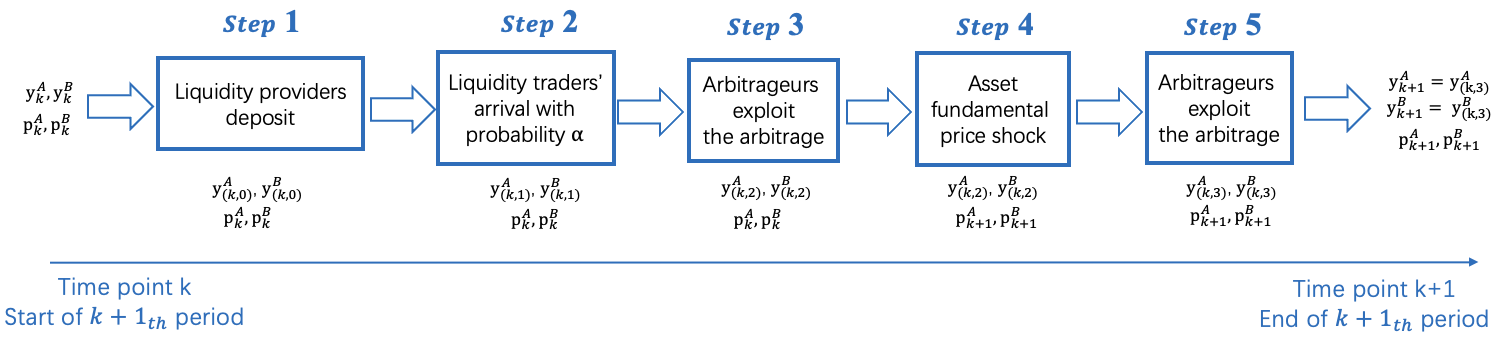}
		\caption {Timeline of trading activities in the $(k+1)$-th period. Here, $y^i$ and $p^i$ denote the deposit amount on DEX and fundamental value on CEX, respectively, of asset $i\in \{A,B\}$. The subscript represents the timeline: time $k$ is the start of the $(k+1)$-th period, and $k+1$ is the end of the $(k+1)$-th period and also the start of the $(k+2)$-th period; $(k,j)$ for $j=0,1,2,3$ represent the time points right after deposit or withdrawal by the LP (Step 1), liquidity traders' trade (Step 2), arbitrageurs' pre-shock trade (Step 3), and arbitrageurs' post-shock trade (Step 5), respectively.}\label{fi:timeline}
	\end{figure}
	
	The trading occurs in discrete time over an infinite horizon. The timeline in each period contains five steps as summarized in Figure \ref{fi:timeline}. Upon observing the deposit amounts $y_k^i$ and fundamental asset values $p_k^i$, $i\in\{A,B\}$ at the start of the $(k+1)$-th period, the LP adjusts her deposit on DEX in Step 1, resulting in an updated deposit $y^i_{(k,0)}$. In Step 2, liquidity traders may arrive and trade on the DEX, and their arrival is captured via the Bernoulli variable $\xi_{k+1}$ with arrival probability $\alpha$. The liquidity traders' trading then leads to an updated deposit amount $y^i_{(k,1)}$. In Step 3, arbitrageurs trade if there is any arbitrage opportunity between the DEX and CEX, resulting in an updated deposit amount $y^i_{(k,2)}$. In Step 4, shocks to the fundamental values of both assets arrive on CEX, changing the values to $p^i_{k+1}$. We denote by 	
	\begin{align}\label{fundaShock}
		R^i_{k+1}:=p^i_{k+1}/p^i_k,\quad i=A,B,
	\end{align}
	the random gross returns of assets $A$ and $B$ due to the shock. The fundamental values change triggers another round of potential trades from the arbitrageurs in Step 5, updating the deposit to $y^i_{(k,3)}$. The deposit amounts $y^i_{(k,3)}$ and the post-shock fundamental values $p^i_{k+1}$ are carried forward as $y^i_{k+1}$ and $p^i_{k+1}$ at the end of the $(k+1)$-th period, which is also the beginning of the $(k+2)$-th period. Steps 3 and 5 in our model are, respectively, the reversion arbitrage and the rebalancing arbitrage in \citet{MilionisMoallemiRoughgardenZhang2024:AutomatedMarketMaking}.
	
	For model tractability, we impose the following assumptions on AMM, which are similar to those imposed in \citet{CapponiJia2021:TheAdoptionOfBlockchainBased} and \citet{FabiPrat2023:TheEconomicsOfCFMM}.
	
	\begin{assumption}[Pricing Formula]\label{as:PricingFunction} The pricing formula $F:(0,\infty)\times (0,\infty)\to \mathbb{R}$ satisfies
		\leavevmode
		\begin{enumerate}
			\item[(i)] $F$ is continuously differentiable with $F_{x}(x,y)>0, F_{y}(x,y)>0$ for all $x,y>0$.
			\item[(ii)] $F_{x}(x,y)/F_{y}(x,y) =G(x/y)$ for some continuously differentiable univariate function $G$ on $(0,\infty)$ satisfying
				$
				G'(z)<0,\; \forall z>0,~~ \lim_{z\downarrow 0}G(z) = \infty,~~ \lim_{z\uparrow \infty}G(z) = 0.
				$
			\item[(iii)] If $F(x,y)=F(x',y')$, then $F(cx,cy)=F(cx',cy')$ for all $c>0$.
		\end{enumerate}
	\end{assumption}
	
	Assumption \ref{as:PricingFunction}-(i) ensures that starting from any deposit amount $y^A>0$ and $y^B>0$, for any sufficiently small, nonnegative $D^A$, there exists a unique $D^B\le 0$, denoted as $h(D^A)$, such that \eqref{AMMrule} holds, and $h(D^A)$ is continuous and strictly decreasing in $D^A$. Therefore, we can consider the maximal interval $[0,K)$ on which $h(D^A)$ is defined. We must have $K\in (0,y^A]$ and that $\lim_{D^A\uparrow K} h(D^A)=-\infty$ if $K<y^A$. Therefore, $\lim_{D^A\uparrow K} (y^A-D^A)/(y^B-h(D^A))=0$, showing that an investor can acquire any amount for asset $A$ from the pool as long as the acquisition does not deplete asset $A$ in the pool. Similarly, an investor can acquire any amount for asset $B$ from the pool as long as the acquisition does not deplete asset $B$. The more an investor wants to exchange for one asset in the pool, the larger the amount of the other asset she needs to deposit in the pool.
	
	With Assumption \ref{as:PricingFunction}-(i), for any fixed deposits in the pool $y^A>0$ and $y^B>0$, the {\em average exchange rate} of acquiring asset $A$ for a small amount $D_A>0$ is well defined, which is $-D_B/D_A$ with $D_A$ and $D_B$ satisfying \eqref{AMMrule}. We define the {\em marginal exchange rate} of acquiring asset $A$ to be the average exchange rate when the amount $D_A$ of asset $A$ to be acquired is infinitesimally small. Applying the implicit function theorem, we can derive from \eqref{AMMrule} that the marginal exchange rate of acquiring asset $A$ is
		$
		F_x(y^A,y^B)/F_y(y^A,y^B).
		$
	Similarly, we can define the average exchange rate of acquiring asset $B$ for a small amount and the marginal exchange rate of acquiring asset $B$, and the latter is $F_y(y^A,y^B)/F_x(y^A,y^B)$, which is the reciprocal of the marginal exchange rate of acquiring asset $A$.
	
	Assumption \ref{as:PricingFunction}-(ii) imposes that the marginal exchange rate of acquiring asset $A$ depends on $y^A$ and $y^B$ only through their ratio $y^A/y^B$. Moreover, the marginal exchange rate is strictly decreasing in $y^A/y^B$ and converges to 0 and $\infty$ when $y^A/y^B$ goes to $\infty$ and 0, respectively. In other words, when the amount of asset $A$ in the liquidity pool becomes larger relative to the amount of asset $B$, it becomes cheaper to acquire asset $A$ from the pool, and the exchange rate ranges from $\infty$ to 0 as the percentage of the amount of asset $A$ in the pool varies from 0 to 100\%. Consequently, it is prohibitively costly for an investor to acquire the whole amount of an asset in the pool.
	
	Assumption \ref{as:PricingFunction}-(iii) can be reformulated as follows: For any $y^A>0$, $y^B>0$, $D^A$ and $D^B$ such that $F(y^A-D^A,y^B-D^B) = F(y^A,y^B)$, we have $F(cy^A-cD^A,cy^B-cD^B) = F(cy^A,cy^B)$. Therefore, Assumption \ref{as:PricingFunction}-(iii) stipulates that if we split a liquidity pool into several identical mini-pools, then any trade in the original pool can be split into trades in the mini-pools, and the same holds if we aggregate several identical pools. This assumption ensures that the optimal amount of assets to trade in the pool is proportional to the asset deposits in the pool; see Proposition \ref{le:TradersProblem} below.
	
	It is straightforward to verify that the CGMMM pricing function \eqref{eq:CGMMM} satisfies Assumption \ref{as:PricingFunction} with
	\begin{align}
		G(z) = \frac{\eta}{(1-\eta)z},\quad z>0.\label{eq:CGMMMMarginalExchangeRateFun}
	\end{align}

	Note that a monotone transformation of the pricing function does not change the exchange rate on AMM. Indeed, consider any strictly increasing, continuously differentiable function $g$. Define the new pricing function $\tilde F(x,y):=g(F(x,y))$. Then, the pricing equation \eqref{AMMrule} holds if and only if
		$
		\tilde F(y^A-D^A,y^B-D^B)= \tilde F(y^A,y^B).
		$
	It is straightforward to verify that Assumption \ref{as:PricingFunction} holds for $F$ if and only if it holds of $\tilde F$.
	
	\subsubsection{Discussions on the Market Model Setup}\label{subsubse:DiscussionModelSetup}
	
	As previously mentioned, we follow \citet{hasbrouck2026need} and \citet{CapponiJia2021:TheAdoptionOfBlockchainBased} to assume that the asset prices on the CEX are the same as their fundamental values. In other words, these fundamental values are fully discovered on CEX. In consequence, the LP has no incentive to do speculative trading based on private information or the discrepancy between the CEX prices and fundamental values of the assets. We are thus able to focus on the trade-off of trading fees and impermanent loss from the LP's perspective and on the AMM design. It would be an interesting future research direction to consider the case in which the LP can also do speculative trading. Nevertheless, we expect that the main insights of our model on the trade-off of trading fees and impermanent loss and the optimal AMM design would still hold.

	The distinction between arbitrageurs and liquidity traders is meant to capture different trading motives, not a literal difference in access to CEX price data. Arbitrageurs trade on DEX--CEX price discrepancies and use the CEX price as the reference price. Liquidity traders represent non-arbitrage order flow on the DEX. Such order flow may come from liquidity needs, portfolio rebalancing, wallet or settlement constraints, heterogeneous valuations, or a preference for on-chain execution. This separation follows the standard market-microstructure distinction between informed or arbitrage trading and liquidity or noise trading, as in \citet{kyle1985continuous}, and it is also consistent with related AMM models such as \citet{CapponiJia2021:TheAdoptionOfBlockchainBased}.
		
	The role of liquidity traders is also economically important for our AMM-design problem. Price shocks alone would generate arbitrage opportunities and hence arbitrage losses for LPs, but they would not generate the non-arbitrage trading volume from which LPs earn fee income. Removing liquidity traders from our model would therefore remove one side of the key tradeoff studied in the paper: the tradeoff between fee revenue from liquidity demand and losses to arbitrageurs. In  consequence, providing liquidity on DEXs would only incur losses compared to investing in DEXs, so LPs would not be willing to provide liquidity on DEXs.

	In practice, the LP needs to pay transaction fees when trading on the CEX and pay gas fees when providing or withdrawing liquidity on the DEX. Our model abstracts from all these frictions, as in many classical portfolio selection problems. Transaction fees on CEXs are often small. For instance, Binance's regular spot trading fee is around 10 basis points and can be lower for high volume users, and gas fees have been decreasing in recent years and do not scale with dollar transaction size.\footnote{For example, see \url{https://www.binance.com/en/fee/trading} for spot maker and taker fees; see also \cite{barbon2026quality} for a discussion on the gas fees.} Therefore, it is not restrictive to ignore CEX trading fees and gas fees in our model if the LP does not trade frequently, i.e., if the length of each period in our model is not too short.

	On the other hand, the trading fee on the DEX is not a friction to the LP, because the LP does not trade against the pool. Instead, it is a fee paid by investors who trade against the pool to the LP. It is necessary to model the trading fee in our problem, because as previously mentioned, without the fee revenue from the liquidity traders, the LP would not provide liquidity on the DEX.

	Finally, Assumption \ref{as:PricingFunction} is imposed mainly for the purpose of tractability, as will be seen in Sections \ref{sect:sub:tradersProblem} and \ref{subse:LPProblem}. As aforementioned, this assumption is satisfied by the most commonly used class of pricing functions in the market, namely, CGMMM. It is also satisfied by some other function classes, for instance, the following class of homothetic pricing functions:
	\begin{align*}
			F(y^A,y^B)
			=
			\left(\theta (y^A)^\rho+(1-\theta)(y^B)^\rho\right)^{1/\rho},
			\qquad 
			\theta\in(0,1),\quad \rho<1,\quad \rho\neq 0.
		\end{align*}
		On the other hand, it is worth mentioning that pricing functions that do not satisfy Assumption \ref{as:PricingFunction} are also studied in some recent AMM literature; see for instance \citet{CarteaEtal2024:AutomatedMarketMakers}, \citet{SchlegelEtal2023:AxiomsForConstantFunction}, and \citet{FrongilloEtal2024:AnAxiomatic}.

	\subsection{Investors' Problem}\label{sect:sub:tradersProblem}
	We first study the optimal investment strategy of the investors (the liquidity traders and arbitrageurs) corresponding to Steps 2, 3, and 5 of Figure \ref{fi:timeline}. We assume that the investors aim to maximize their profits by trading in the DEX. Following \citet{EvansAngerisChitra2021:OptimalFees} and \citet{CapponiJia2021:TheAdoptionOfBlockchainBased}, we assume that the investors have the trading objectives based on their beliefs of the values of assets $A$ and $B$, denoted as $a$ and $b$, respectively.
	
	Denote by $y^i$ the deposit of asset $i\in \{A,B\}$ on DEX immediately before the trade (i.e. pre-trade deposit). Then, the investors' investment problem is formulated as
	\begin{align}\label{eq:TraderProblem}
		\begin{array}{rl}
			\underset{{D^{A}, D^{B}}}{\max} & a (1+f\mathbf 1_{D^{A}<0})D^{A}+ b (1+f\mathbf 1_{D^{B}<0}) D^{B} \\ 			\text{subject to} & F\left(y^A, y^B\right)=F\left(y^A-D^{A}, y^B-D^{B}\right),\\
			& y^A-D^{A} > 0,~~y^B-D^{B} > 0.
		\end{array}
	\end{align}
	Here, the investors' strategy is represented by the trading amount of both risky assets $D^{A}$ and $D^{B}$, subject to the constraints of the AMM trading rule described in Section \ref{sect:AMMintro}. The objective function $a(1+f\mathbf 1_{D^{A}<0})D^{A}+ b (1+f\mathbf 1_{D^{B}<0}) D^{B}$, based on \eqref{change_deposit}, describes the post-fee profit of the strategy measured by the multiplications of asset values and traded asset amounts.
	
	We provide the solution to the problem \eqref{eq:TraderProblem} in the following proposition, and the proof is relegated to Appendix \ref{proof}. Note that this problem has been solved in the literature, e.g., in \citet{EvansAngerisChitra2021:OptimalFees} and \citet{CapponiJia2021:TheAdoptionOfBlockchainBased}, albeit under slightly different assumptions. Here, we not only derive the solution but also prove some properties of the solution that will be used in solving the LP's decision problem.
	
	\begin{proposition}\label{le:TradersProblem}
		Let Assumption \ref{as:PricingFunction} hold.
		\begin{enumerate}
			\item[(i)] Denote by $G^{-1}$ the inverse function of $G$ and fix any $\alpha>0$ and $\beta>0$. If $G(\beta)/\alpha<1/(1+f)$, then the following system of equations
			\begin{align}
				\frac{1- d^A}{1-d^B} = \frac{1}{\beta} G^{-1}\big(\alpha/(1+f)\big),\quad  F\big(\beta(1-d^A),1-d^B\big)=F(\beta,1)\label{eq:EquationsOptimalTradingForAssetA}
			\end{align}
			admits a unique solution $(d^A,d^B)$ with $d^A\in (0,1)$ and $d^B<0$. If $G(\beta)/\alpha>(1+f)$, then the following system of equations
			\begin{align}
				\frac{1- d^A}{1-d^B} = \frac{1}{\beta} G^{-1}\big(\alpha(1+f)\big),\quad  F\big(\beta(1-d^A),1-d^B\big)=F(\beta,1)\label{eq:EquationsOptimalTradingForAssetB}
			\end{align}
			admits a unique solution $(d^A,d^B)$ with $d^A<0$ and $d^B\in (0,1)$. Denote by $(\varphi^A(\alpha,\beta),\varphi^B(\alpha,\beta))$ the solution in the above two cases and set $(\varphi^A(\alpha,\beta),\varphi^B(\alpha,\beta))=(0,0)$ when $G(\beta)/\alpha \in [1/(1+f),(1+f)]$. Then, $\varphi^i(\alpha,\beta)$ is continuous in $(\alpha,\beta)\in (0,\infty)\times (0,\infty)$ and takes value 0 in the region with $G(\beta)/\alpha\in  [1/(1+f),(1+f)]$, $i\in\{A,B\}$, and $\varphi^A(\alpha,\beta)$ is strictly increasing and $\varphi^B(\alpha,\beta)$ is strictly decreasing in both $\alpha$ and $\beta$ in the region with $G(\beta)/\alpha\notin [1/(1+f),(1+f)]$.
			\item[(ii)] For any $y^A>0$, $y^B>0$, $a>0$, and $b>0$, \eqref{eq:TraderProblem} admits a unique solution $(\hat D^A,\hat D^B)$ where $\hat D^i = \varphi^i(a/b,y^A/y^B) y^i$, $i\in \{A,B\}$.
		\end{enumerate}
	\end{proposition}
	
	Proposition \ref{le:TradersProblem} shows that the optimal trading amount is proportional to the pre-trade deposit. Moreover, this proportion depends on the investor's belief of the asset values $a$ and $b$ and the amount of deposit $y^A$ and $y^B$ only through their respective ratios $a/b$ and $y^A/y^B$. Furthermore, there is no trading when $G(y^A/y^B)/(a/b)$ is in the region $[1/(1+f),(1+f)]$, which means that the difference between the marginal exchange rate on AMM and the exchange rate believed by the investor is small compared with the unit trading fee. When $G(y^A/y^B)/(a/b)<1/(1+f)<1$ (resp. $>1+f>1$), which means that the marginal exchange rate for asset $A$ on AMM is lower (resp. higher) than the exchange rate believed by the investor by a sufficiently large margin, the investor chooses to exchange for a positive amount of asset $A$ (resp. asset B) on AMM.
	
	By Proposition \ref{le:TradersProblem}, the post-trade deposit of asset $i$, $i\in\{A,B\}$ in the liquidity pool is
	\begin{align}\label{eq:AfterTradeDeposit}
		y^i-(1+f\mathbf 1_{\hat D^i<0})\hat D^i &= \phi^i(a/b,y^A/y^B) y^i,\quad i\in\{A,B\},\\ \label{eq:PostTradeDeposit}
		\text{where }\phi^i(\alpha,\beta)&:=1-(1+f\mathbf 1_{\varphi^i(\alpha,\beta)<0}) \varphi^i(\alpha,\beta).
	\end{align}
	Thus, the post-trade deposit ratio is determined by the pre-trade deposit ratio $y^A/y^B$ and the investor's belief of the exchange rate $a/b$. By Proposition \ref{le:TradersProblem}, for $i\in\{A,B\}$, $\varphi^i(\alpha,\beta)$ is continuous in $(\alpha,\beta)\in (0,\infty)\times (0,\infty)$ and takes value 0 in the region with $G(\beta)/\alpha\in  [1/(1+f),(1+f)]$, and $\varphi^A(\alpha,\beta)$ is strictly increasing and $\varphi^B(\alpha,\beta)$ is strictly decreasing in both $\alpha$ and $\beta$ in the region with $G(\beta)/\alpha\notin [1/(1+f),(1+f)]$. It is then straightforward to see that $\phi^i(\alpha,\beta)$ is continuous in $(\alpha,\beta)\in (0,\infty)\times (0,\infty)$ and takes value 1 in the region with $G(\beta)/\alpha\in  [1/(1+f),(1+f)]$, and $\phi^A(\alpha,\beta)$ is strictly decreasing and $\phi^B(\alpha,\beta)$ is strictly increasing in both $\alpha$ and $\beta$ in the region with $G(\beta)/\alpha\notin [1/(1+f),(1+f)]$. Moreover, because $\varphi(\alpha,\beta)<1$, $\phi^i$ takes values in $(0,+\infty)$.
	
	The ratio of the marginal exchange rate on AMM and the investor's belief of the exchange rate after the trade (post-trade ratio, left hand side of \eqref{eq:ExchangeRateRatioTransition}) depends on the ratio before the trade (pre-trade ratio, first argument of $H$ on right hand side of \eqref{eq:ExchangeRateRatioTransition}) as follows
	\begin{align}
		G\left(\frac{y^A-(1+f\mathbf 1_{\hat D^A<0})\hat D^A}{y^B-(1+f\mathbf 1_{\hat D^B<0})\hat D^B}\right)/(a/b)= H\left(G\left(\frac{y^A}{y^B}\right)/(a/b),a/b\right),\label{eq:ExchangeRateRatioTransition}
	\end{align}
	where
	\begin{align}
		H(s,z):=\begin{cases}
			G\left(\frac{1-\varphi^B\left(z,G^{-1}\big(sz\big)\right)}{1-(1+f)\varphi^B\left(z,G^{-1}\big(sz\big)\right)} G^{-1}\left(z/(1+f)\right)\right)/z, &\text{if } s<1/(1+f)\\
			s,& \text{if }s\in[1/(1+f),1+f],\\
			G\left(\frac{1-(1+f)\varphi^A\left(z,G^{-1}\big(sz\big)\right)}{1-\varphi^A\left(z,G^{-1}\big(sz\big)\right)} G^{-1}\left(z(1+f)\right)\right)/z, &\text{if }s>(1+f).
		\end{cases}\label{eq:ExchangeRateRatioTransitionFunction}
	\end{align}
	Note that $H(s,z)\in[1/(1+f),1+f]$ for all $s>0$ and $z>0$, i.e., the post-trade ratio lies in $[1/(1+f),1+f]$; otherwise the trading will continue. If the pre-trade ratio is in the region $[1/(1+f),1+f]$, there is no trade, so the post-trade ratio remains the same.
	For $s<1/(1+f)$, by Proposition \ref{le:TradersProblem}-(i), $\varphi^B\left(z,G^{-1}\big(sz\big)\right)<0$, so $H(s,z)>1/(1+f)$. Thus, when the pre-trade ratio is strictly less than $1/(1+f)$, the investor exchanges for asset $A$ and the trading brings the ratio not onto the boundary of the no-trading region $[1/(1+f),1+f]$ but into the interior of the no-trading region. This is because the trading cost is deposited into the liquidity pool and thus changes the exchange rate on AMM. Moreover, by Proposition \ref{le:TradersProblem}-(i), $\varphi^B\left(z,G^{-1}\big(sz\big)\right)$ is strictly increasing in $s$ for $s<1/(1+f)$ and thus $H(s,z)$ is strictly decreasing in $s$ for $s<1/(1+f)$. Consequently, the post-trade ratio is strictly decreasing in the pre-trade ratio on the left side of the no-trading region. This is because the further the pre-trade ratio is away from the no-trading region from the left side, the more valuable asset $A$ is from the investor's perspective and thus a larger amount of asset $A$ is exchanged by the investor, which incurs a larger trading cost deposited in the liquidity pool and pushes the post-trade ratio further into the no-trading region. Similarly, when the pre-trade ratio is strictly larger than $1+f$, trading occurs and brings the post-trade ratio into the interior of the no-trading region, and the latter is strictly decreasing in the former.	
	
	Recall the timeline of trading in the $(k+1)$-th period as described in Figure \ref{fi:timeline}. In Step 2, the liquidity trader arrives based on the Bernoulli random variable $\xi_{k+1}$. We assume that her belief of the exchange rate of the two assets is
	\begin{align}
		a/b = (p^{A}_k/p^B_k) I_{k+1},\label{eq:LiquidityTradersBelief}
	\end{align}
	where $I_{k+1}$ is a random variable representing heterogeneity of investors' beliefs. As a result, by \eqref{eq:AfterTradeDeposit},
	\begin{align}
		y^i_{(k,1)}=\left[(1-\xi_{k+1}) + \xi_{k+1}\phi^i\left((p^{A}_k/p^B_k) I_{k+1},y^A_{(k,0)}/y^B_{(k,0)}\right)\right] y^i_{(k,0)},\quad i\in \{A,B\}.\label{eq:PostTradeDepositStep2}
	\end{align}
	Consequently,
	\begin{align}
		\frac{ y^A_{(k,1)}}{ y^B_{(k,1)}} = \left[(1-\xi_{k+1}) + \xi_{k+1}\frac{\phi^A\left((p^{A}_k/p^B_k) I_{k+1},y^A_{(k,0)}/y^B_{(k,0)}\right)}{\phi^B\left((p^{A}_k/p^B_k) I_{k+1},y^A_{(k,0)}/y^B_{(k,0)}\right)}\right]\frac{ y^A_{(k,0)}}{ y^B_{(k,0)}}.\label{eq:PostTradeDepositRatioStep2}
	\end{align}

	We assume that the arbitrageurs have correct beliefs of the exchange rate. Therefore, in Step 3, we set $a/b=p_k^A/p_k^B$ and in Step 5 (after the price shock), we set $a/b=p_{k+1}^A/p_{k+1}^B$. Consequently, by \eqref{eq:AfterTradeDeposit}, we have
	\begin{align}
		&y^i_{(k,2)} = \phi^i\left(p^{A}_k/p^B_k,y^A_{(k,1)}/y^B_{(k,1)}\right)y^i_{(k,1)},\quad i\in\{A,B\},\label{eq:PostTradeDepositStep3}\\
		&\frac{ y^A_{(k,2)}}{ y^B_{(k,2)}} = \frac{\phi^A\left(p^{A}_k/p^B_k ,y^A_{(k,1)}/y^B_{(k,1)}\right)}{\phi^B\left(p^{A}_k/p^B_k ,y^A_{(k,1)}/y^B_{(k,1)}\right)}\times \frac{ y^A_{(k,1)}}{ y^B_{(k,1)}},\label{eq:PostTradeDepositRatioStep3}\\
		&y^i_{(k,3)} = \phi^i\left(p^{A}_{k+1}/p^B_{k+1},y^A_{(k,2)}/y^B_{(k,2)}\right)y^i_{(k,2)},\quad i\in\{A,B\},\label{eq:PostTradeDepositStep5}\\
		&\frac{ y^A_{(k,3)}}{ y^B_{(k,3)}} = \frac{\phi^A\left(p^{A}_{k+1}/p^B_{k+1} ,y^A_{(k,2)}/y^B_{(k,2)}\right)}{\phi^B\left(p^{A}_{k+1}/p^B_{k+1} ,y^A_{(k,2)}/y^B_{(k,2)}\right)}\times \frac{ y^A_{(k,2)}}{ y^B_{(k,2)}}.\label{eq:PostTradeDepositRatioStep5}
	\end{align}

	\subsection{LP's Problem}\label{subse:LPProblem}
	
	\subsubsection{LP's Investment Return on DEX}
	Recall the timeline of trading in each period as described in Figure \ref{fi:timeline}.
	The gross return of the LP's investment on DEX in the $(k+1)$-th period, denoted as $R^M_{k+1}$, is then
	\begin{align}
		&R^M_{k+1} 
		= \frac{p_{k+1}^Ay_{(k,3)}^A+p_{k+1}^By_{(k,3)}^B}{p_{k}^Ay_{(k,0)}^A+p_{k}^By_{(k,0)}^B}\notag\\
		=& \frac{p_{k}^Ay_{(k,1)}^A+p_{k}^By_{(k,1)}^B}{p_{k}^Ay_{(k,0)}^A+p_{k}^By_{(k,0)}^B}\cdot  \frac{p_{k}^Ay_{(k,2)}^A+p_{k}^By_{(k,2)}^B}{p_{k}^Ay_{(k,1)}^A+p_{k}^By_{(k,1)}^B}\cdot \frac{p_{k+1}^Ay_{(k,2)}^A+p_{k+1}^By_{(k,2)}^B}{p_{k}^Ay_{(k,2)}^A+p_{k}^By_{(k,2)}^B}\cdot
		\frac{p_{k+1}^Ay_{(k,3)}^A+p_{k+1}^By_{(k,3)}^B}{p_{k+1}^Ay_{(k,2)}^A+p_{k+1}^By_{(k,2)}^B}, 
		\label{eq:AMMGrossReturn}
	\end{align}
	where the four terms on the right hand side are denoted as $R^M_{(k+1,i)},~i=1,2,3,4$, respectively. Here, $R^M_{(k+1,1)}$, $R^M_{(k+1,2)}$, and $R^M_{(k+1,4)}$ stand for the change of the liquidity pool value due to trading by investors and arbitrageurs in Steps 2, 3, and 5, respectively, and  $R^M_{(k+1,3)}$ stands for the change of the liquidity pool value due to the price shock in Step 4. By \eqref{fundaShock}, \eqref{eq:PostTradeDepositStep2}, \eqref{eq:PostTradeDepositStep3}, and \eqref{eq:PostTradeDepositStep5}, we derive
	\begin{align}
		R^M_{(k+1,1)} 
		&= 1-\xi_{k+1} + \xi_{k+1}\frac{\frac{p_{k}^A}{p_{k}^B}\frac{y_{(k,0)}^A}{y_{(k,0)}^B}\phi^A\left(\frac{p_{k}^A}{p_{k}^B} I_{k+1},\frac{y_{(k,0)}^A}{y_{(k,0)}^B}\right)+\phi^B\left(\frac{p_{k}^A}{p_{k}^B} I_{k+1},\frac{y_{(k,0)}^A}{y_{(k,0)}^B}\right)}{\frac{p_{k}^A}{p_{k}^B}\frac{y_{(k,0)}^A}{y_{(k,0)}^B}+1},\label{eq:AMMReturnStep2}\\
		R^M_{(k+1,2)}
		&= \frac{\frac{p_{k}^A}{p_{k}^B}\frac{y_{(k,1)}^A}{y_{(k,1)}^B}\phi^A\left(\frac{p_{k}^A}{p_{k}^B},\frac{y_{(k,1)}^A}{y_{(k,1)}^B}\right)+ \phi^B\left(\frac{p_{k}^A}{p_{k}^B},\frac{y_{(k,1)}^A}{y_{(k,1)}^B}\right)}{\frac{p_{k}^A}{p_{k}^B}\frac{y_{(k,1)}^A}{y_{(k,1)}^B}+1},\label{eq:AMMReturnStep3}\\
		R^M_{(k+1,3)}&=R^{B}_{k+1}\frac{\frac{R^{A}_{k+1}}{R^{B}_{k+1}}\frac{p_{k}^A}{p_{k}^B}\frac{y_{(k,2)}^A}{y_{(k,2)}^B}+1}{\frac{p_{k}^A}{p_{k}^B}\frac{y_{(k,2)}^A}{y_{(k,2)}^B}+1}\label{eq:AMMReturnStep4}\\
		R^M_{(k+1,4)} 
		&= \frac{\frac{R^{A}_{k+1}}{R^{B}_{k+1}}\frac{p_{k}^A}{p_{k}^B}\frac{y_{(k,2)}^A}{y_{(k,2)}^B}\phi^A\left(\frac{R^{A}_{k+1}}{R^{B}_{k+1}}\frac{p_{k}^A}{p_{k}^B},\frac{y_{(k,2)}^A}{y_{(k,2)}^B}\right)+ \phi^B\left(\frac{R^{A}_{k+1}}{R^{B}_{k+1}}\frac{p_{k}^A}{p_{k}^B},\frac{y_{(k,2)}^A}{y_{(k,2)}^B}\right)}{\frac{R^{A}_{k+1}}{R^{B}_{k+1}}\frac{p_{k}^A}{p_{k}^B}\frac{y_{(k,2)}^A}{y_{(k,2)}^B}+1}.\label{eq:AMMReturnStep5}
	\end{align}

	\subsubsection{LP's Actions}
	We follow the classical framework of \citet{MertonR:69ps,MertonR:71oc} to model LP's decision problem.
	The LP's actions consist of four components: consumption, investment in assets $A$ and $B$ on DEX (i.e., liquidity provision), investment in assets $A$ and $B$ on CEX, and investment in the risk-free asset. The LP rebalances the investment portfolio every period but consumes every $N\ge 1$ periods: At the beginning of $(jN+1)$-th period, i.e., at time $jN$, the LP chooses a percentage $c_{jN}$ of her wealth to consume, for every $j=0,1,\dots$. To ease exposition, for $k\notin \{jN:j=0,1,\dots\}$, we assume that the percentage of wealth consumed to be $c_k=0$. 
	
	At time $k$, $k=0,1,\cdots$, denote by $X_k$ the LP's {\em pre-consumption} wealth, and therefore the LP's post-consumption wealth is $(1-c_k)X_k$. The LP then decides the amount to be invested on DEX and CEX. For the investment on DEX, we impose the following protocol of LP's deposit into and withdrawal from DEX, which is used in Uniswap.
	\begin{assumption}[Deposit and Withdrawal Protocol]\label{as:LPDeposit}
		The LP deposits and withdraws assets from the liquidity pool by keeping the marginal exchange rate unchanged.
	\end{assumption}

	Assumption \ref{as:LPDeposit} is a protocol assumption on liquidity provision and withdrawal, not a restriction to a particular pricing function.
	With this assumption and Assumption \ref{as:PricingFunction}-(ii) in place, we have
		\begin{align}
			y_{k}^A/y_{k}^B=y_{(k,0)}^A/y_{(k,0)}^B.\label{eq:PostTradeDepositRatioStep1}
		\end{align}
	As a result, for the investment on DEX, the LP only needs to decide the total value invested on DEX or, equivalently, the percentage $\omega^M_k$ of the post-consumption wealth to be invested on DEX. Subsequently, the deposits in the liquidity pool after the LP's investment become
	\begin{align}\label{initialDepo}
		y^i_{(k,0)}=\omega^M_k (1-c_k) X_k \times \frac{y^i_k}{y^A_kp^A_k+y^B_k p^B_k},\quad i\in\{A,B\}.
	\end{align}
	Therefore, the LP deposits into (resp. withdraws from) the liquidity pool if and only if $\omega^M_k (1-c_k) X_k - (y^A_k p^A_k+y^B_k p^B_k)>0$ (resp. $<0$).

	Besides DEX, the LP can also invest in assets $A$ and $B$ on CEX, represented by the percentages $\omega^A_k$ and $\omega^B_k$ of the post-consumption wealth allocated to these two assets, respectively. The remaining wealth is invested in the risk-free asset, which generates a constant gross return $R_f$. Recalling the investment return on DEX as in \eqref{eq:AMMGrossReturn} and the returns of investment in assets $A$ and $B$ on CEX as in \eqref{fundaShock}, we derive the LP's portfolio gross return as
	\begin{align}\label{eq:PortfolioReturn}
		R_{k+1}^p:=\omega^M_k R^M_{k+1} + \omega^A_k R^A_{k+1} + \omega^B_k R^B_{k+1}+ \big(1-\omega^M_k-\omega^A_k-\omega^B_k)R_f,
	\end{align}
	and the dynamics for the LP's wealth as
	\begin{align}
		X_{k+1} = (1-c_k) X_k R^p_{k+1}.\label{eq:WealthEquation}
	\end{align}
	The LP's actions are summarized in the 4-tuple $(c_k,\omega_{k}^{M},\omega_k^A,\omega_k^B)$. The feasible set ${\cal C}_k$ for the consumption action $c_k$ is $(0,1)$ for $k\in \{jN:j=1,2,\dots\}$ and $\{0\}$ otherwise. We assume no short sale and no borrowing, so the feasible set $\Omega$ for the investment action $\vecomega_k=(\omega_{k}^{M},\omega_k^A,\omega_k^B)$ is \
	\begin{align}
		\Omega=\{\vecomega=(\omega^M,\omega^A,\omega^B)\mid \omega^i\ge 0,i\in\{M,A,B\},\omega^{M}+\omega^A+\omega^B\le 1\}.\label{eq:PortfolioConstraints}
	\end{align}

	\subsubsection{LP's objective}
	We assume that the LP is {\em risk-averse} with respect to the consumption flow and her preferences for the consumption flow are modeled by an additive expected utility with utility function
	\begin{align}\label{utility}
		U(C) := \begin{cases}
			\frac{C^{1-\gamma}-1}{1-\gamma}, &0<\gamma\neq 1,\\
			\log (C) ,& \gamma =1,
		\end{cases}
	\end{align}
	where the parameter $\gamma>0$ measures the relative risk aversion degree of the LP.

	Denote by $\delta\in(0,1)$ the one-period discounting factor. The LP's objective is to maximize her expected utility of intertemporal consumption over an infinite horizon:
	\begin{align}\label{controlObjective}
		\sup_{\{(c_k,\vecomega_k)\}_{k=0}^\infty} \mathbb{E}\left[\sum_{j=0}^{\infty} \delta^{jN} U(c_{jN}X_{jN})\right],
	\end{align}
	subject to the wealth equation \eqref{eq:WealthEquation} and constraints on $(c_k,\vecomega_k)$.

	\subsection{Main Results}
	\subsubsection{State Variables}\label{subsubse:StateVariable}
	In the following, we solve the LP's problem \eqref{controlObjective}. To this end, we need to identify the state variables for this problem.
	
	In view of \eqref{eq:AMMGrossReturn}--\eqref{eq:AMMReturnStep5}, \eqref{eq:PostTradeDepositRatioStep1}, \eqref{eq:PostTradeDepositRatioStep2}, and \eqref{eq:PostTradeDepositRatioStep3}, we immediately conclude that the return on AMM $R^M_{k+1}$ is independent of the LP's action $(c_k,\vecomega_k)$ and determined by (i) the deposit ratio $y^A_{k}/y^B_{k}$ at the beginning of the $(k+1)$-th period, (ii) the fundamental exchange rate $p^A_k/p^B_k$ at the beginning of this period, and (iii) random disturbances in this period, which include $\xi_{k+1}$ representing the arrival of investors in Step 2, $I_{k+1}$ representing the heterogeneity of the investors' belief of the fundamental exchange rate, and $R^i_{k+1}$ representing the growth of the fundamental value of asset $i$ and hence the gross return of investing in asset $i$ on CEX, $i\in \{A,B\}$.

	We make the following assumption on the random disturbances.
	\begin{assumption}[Random Disturbances]\label{as:RandomShock}
		The random vector $(\xi_{k+1},I_{k+1},R^A_{k+1},R^B_{k+1})$ is independent and identically distributed across $k$.
	\end{assumption}
	
	With Assumption \ref{as:RandomShock} in place, given the information at the beginning of $(k+1)$-th period, the distributions of the investment returns $R^M_{k+1}$, $R^A_{k+1}$, and $R^B_{k+1}$ are completely determined by $y^A_{k}/y^B_{k}$ and $p^A_k/p^B_k$. Moreover, $p^A_{k+1}/p^B_{k+1} = (p^A_k/p^B_k)(R^A_{k+1}/R^B_{k+1})$ and by \eqref{eq:PostTradeDepositRatioStep1}, \eqref{eq:PostTradeDepositRatioStep2}, \eqref{eq:PostTradeDepositRatioStep3}, and \eqref{eq:PostTradeDepositRatioStep5}, $y^A_{k+1}/y^B_{k+1}$ is a function of $y^A_{k}/y^B_k$, $p^A_k/p^B_k$, and $(\xi_{k+1},I_{k+1},R^A_{k+1}/R^B_{k+1})$. Therefore, the LP makes decision based on her wealth $X_k$ and the information summarized by $y^A_{k}/y^B_{k}$ and $p^A_k/p^B_k$.
	
	Note that knowing $y^A_{k}/y^B_k$ and $p^A_k/p^B_k$ is equivalent to knowing $G(y^A_{k}/y^B_k)/(p^A_k/p^B_k)$, the ratio of the marginal exchange rate on AMM and the fundamental exchange rate, and $p^A_k/p^B_k$. Thus, we can identify the state variable as $(X_k,S_k,Z_k)$, where
	\begin{align*}
		S_k:=G(y^A_{k}/y^B_k)/(p^A_k/p^B_k) \text{ and } Z_k:=p^A_k/p^B_k.
	\end{align*}
	The reason of using $S_k$ instead of $y^A_{k}/y^B_k$ is that the former always takes value in a bounded interval $[1/(1+f),1+f]$ as discussed in Section \ref{sect:sub:tradersProblem}. Moreover, for the CGMMM pricing function, $S_k$ is a natural choice of the state variable, as we will show later.
	
	The state $Z_k$ evolves according to
	\begin{align}
		Z_{k+1}=  (R^A_{k+1}/R^B_{k+1})Z_k.\label{eq:StateZDynamics}
	\end{align}
	Recalling \eqref{eq:AMMGrossReturn}--\eqref{eq:AMMReturnStep5}, \eqref{eq:PostTradeDepositRatioStep1}, \eqref{eq:PostTradeDepositRatioStep2}, and \eqref{eq:PostTradeDepositRatioStep3}, we can represent
	\begin{align}
		R^M_{k+1} = R^B_{k+1}L^M(S_k,Z_k,\xi_{k+1},I_{k+1},R^A_{k+1}/R^B_{k+1})\label{eq:AMMGrossReturnFunction}
	\end{align}
	for some function $L^M$.
	Thus, the dynamics of the state $X_k$ is given by \eqref{eq:WealthEquation} with $R^M_{k+1}$ given by \eqref{eq:AMMGrossReturnFunction}. Because $y^A_{k+1}/y^B_{k+1}$ is a function of $y^A_{k}/y^B_k$, $p^A_k/p^B_k$, and $(\xi_{k+1},I_{k+1},R^A_{k+1}/R^B_{k+1})$, we can represent
	\begin{align}
		S_{k+1} = L^S(S_k,Z_k,\xi_{k+1},I_{k+1},R^A_{k+1}/R^B_{k+1})\label{eq:StateSDynamics}
	\end{align}
	for some function $L^S$ that can be determined by \eqref{eq:PostTradeDepositRatioStep1}, \eqref{eq:PostTradeDepositRatioStep2}, \eqref{eq:PostTradeDepositRatioStep3}, and \eqref{eq:PostTradeDepositRatioStep5}. Note that only the state $X_{k+1}$ depends on the LP's action $(c_k,\vecomega_k)$, but not the states $S_{k+1}$ and $Z_{k+1}$.
	
	Because $\varphi^i$ is continuous and takes values in $(0,+\infty)$, we conclude that both $L^M$ and $L^S$ are continuous in $(S_k,Z_k,\xi_{k+1},I_{k+1},R^A_{k+1}/R^B_{k+1})\in (0,\infty)\times (0,\infty)\times \{0,1\}\times (0,\infty)\times (0,\infty)$. Moreover, $L^M$ takes values in $(0,\infty)$ and $L^S$ takes values in $[1/(1+f),1+f]$.

	\subsubsection{Dynamic Programming}
	We solve the LP's problem using dynamic programming. Denote by $V_k$ the value function for the LP's problem starting from time $k$, $k=0,1,\dots, N-1$. In view of the discussion in Section \ref{subsubse:StateVariable}, the LP's problem is a controlled Markov processes with states $(X_k,S_k,Z_k)$. Thus, we only need to focus on Markov policies, i.e., to consider $(c_k,\vecomega_k)$ to be functions of $(X_k,S_k,Z_k)$ for each $k$; see Chapters 8--9 in \citet{BertsekasDShreveS:78soc}. Moreover, $V_k$ is a function of $(X_k,S_k,Z_k)$, i.e.,
	\begin{align}
		V_k(x,s,z):= \sup_{\{(c_i,\vecomega_i)\}_{i=k}^\infty}\mathbb{E}\left[U(c_0X_0)\mathbf 1_{k=0}+\sum_{j=1}^\infty \delta^{jN-k} U(c_{jN}X_{jN})\Bigg| X_k=x,S_k=s,Z_k=z\right],\label{eq:LPProblem}
	\end{align}
	subject to the wealth equation \eqref{eq:WealthEquation} with $R^M_{k+1}$ given by \eqref{eq:AMMGrossReturnFunction} and state dynamics \eqref{eq:StateZDynamics} and \eqref{eq:StateSDynamics}. Here and hereafter, with slight abuse of notations, $(c_k,\vecomega_k)$ denotes both the Markov policy at time $k$, which is a function of the state random vector $(X_k,S_k,Z_k)$, and the actions taken by the LP under this policy, i.e., the function evaluated at the observation of $(X_k,S_k,Z_k)$.

	We expect the following dynamic programming equation (DPE) to hold:
	\begin{align*}
		V_0(x,s,z) &= \sup_{c_0\in (0,1),\vecomega_0\in \Omega}\mathbb{E}\left[U(c_0X_0)+\delta V_1\big((1-c_0)X_0R^{p}_1,S_1,Z_1\big) \Bigg| X_0=x,S_0=s,Z_0=z\right],\\
		V_k(x,s,z) &= \sup_{\vecomega_k\in \Omega }\mathbb{E}\left[\delta V_{k+1}\big(X_kR^{p}_{k+1},S_{k+1},Z_{k+1}\big) \Bigg| X_k=x,S_k=s,Z_k=z\right],\; k=1,\dots, N-1,\\
		V_N(x,s,z)& = V_0(x,s,z),
	\end{align*}
	where for every portfolio weight vector $\vecomega_k=(\omega^M_k,\omega^A_k,\omega^B_k)\in \Omega$, the gross return of the portfolio $R^p_{k+1}$ is as defined in \eqref{eq:PortfolioReturn}.
	It is straightforward to see that
	\begin{align}
		V_k(x,s,z) = x^{1-\gamma} V_k(1,s,z) + \left(\mathbf 1_{k=0}+\frac{\delta^{N-k}}{1-\delta^N}\right) U(x),\; k=0,1,\dots, N.
	\end{align}
	Thus, recalling that $U(1)=0$ and denoting
	\begin{align*}
		v_k(s,z):=V_k(1,s,z),\quad k=0,1,\dots, N,
	\end{align*}
	we can derive the following DPE for $v_k$ from the DPE for $V_{k}$:
	\begin{align}
		&v_0(s,z) = (\hat{\mathbb{T}}_0 v_1)(s,z),\quad v_k(s,z) = (\mathbb{T}_k v_{k+1})(s,z),\; k=1,\dots, N-1,\label{eq:DPE}\\
		&v_N(s,z) = v_0(s,z),\label{eq:DPETerminal}
	\end{align}
	where the dynamic programming operators $\mathbb{T}_k$ and $\hat{\mathbb{T}}_0$ are defined to be
	\begin{align}
		&(\mathbb{T}_k J)(s,z):=\delta \sup_{\vecomega_k\in \Omega} \mathbb{E}\Big[ \big(R^p_{k+1}\big)^{1-\gamma}J(S_{k+1},Z_{k+1})+\frac{\delta^{N-k-1}}{1-\delta^N} U\big(R^p_{k+1}\big) \mid S_k=s,Z_k=z\Big],\label{eq:PortfolioOperator}\\
		&(\hat{\mathbb{T}}_0 J)(s,z):= \sup_{c_0\in (0,1)}\left\{U(c_0) + (1-c_0)^{1-\gamma} (\mathbb{T}_0 J)(s,z) + \frac{\delta^N}{1-\delta^N} U(1-c_0)\right\}.\label{eq:ConsumptionOperator}
	\end{align}
	In the DPE \eqref{eq:DPE}--\eqref{eq:DPETerminal}, the terminal condition at time $N$ is the solution to the equation at time 0. Thus, the solution to the DPE is a fixed point of the operator $\hat{\mathbb{T}}_0 \mathbb{T}_1\cdots \mathbb{T}_{N-1}$.
	
	Recall that starting from any exchange rate ratio $S_k$ at time $k$, $S_{k+1}$ always falls in $[1/(1+f),1+f]$. Thus, we assume without loss of generality that the state variable $S_k$ takes values in $[1/(1+f),1+f]$. Denote by
	$
	{\cal S}:=[1/(1+f),1+f]\times (0,\infty)
	$
	the state space. Denote by ${\cal X}$ the set of all bounded, measurable functions on ${\cal S}$, and denote by ${\cal X}_+$ the subset of all nonnegative functions in ${\cal X}$. Denote by $\|\cdot\|$ the maximum norm on ${\cal X}$.
	
	For the case with $\gamma=1$, we will work on $\hat{\mathbb{T}}_0 \mathbb{T}_1\cdots \mathbb{T}_{N-1}$ to show that this operator has a fixed point. For other cases, we need to do some transform. In the case when $\gamma\in (0,1)$, we have $U(C)>-\frac{1}{1-\gamma}$ for every $C>0$, so we have
	\begin{align}\label{eq:DomainValueFunction}
		(1-\gamma)v_0(s,z)+\frac{1}{1-\delta^N}>0,\quad (1-\gamma) v_k(s,z) + \frac{\delta^{N-k}}{1-\delta^N}>0,\; k=1,\dots, N.
	\end{align}
	In the case with $\gamma>1$, we have $U(C)<-\frac{1}{1-\gamma}$ for every $C>0$, so \eqref{eq:DomainValueFunction} also holds. Thus, we denote
	\begin{align*}
		\tilde v_0:=(1-\gamma)v_0+\frac{1}{1-\delta^N},\quad \tilde v_k:=(1-\gamma) v_k + \frac{\delta^{N-k}}{1-\delta^N},\; k=1,\dots, N.
	\end{align*}
	Then, the DPE \eqref{eq:PortfolioOperator}--\eqref{eq:ConsumptionOperator} can be written as
	\begin{align}
		&\tilde v_0(s,z) = (\hat{\mathbb{S}}_0 \tilde v_1)(s,z),\quad \tilde v_k(s,z) = (\mathbb{S}_k \tilde v_{k+1})(s,z),\; k=1,\dots, N-1,\label{eq:DPEPlus}\\
		&\tilde v_N(s,z) = \tilde v_0(s,z),\label{eq:DPEPlusTerminal}
	\end{align}
	where
	\begin{align}
		&(\mathbb{S}_k J)(s,z):=\delta \sup_{\vecomega_k\in \Omega} \mathbb{E}\Big[ \big(R^p_{k+1}\big)^{1-\gamma}J(S_{k+1},Z_{k+1}) \mid S_k=s,Z_k=z\Big],\label{eq:PortfolioOperatorPlus}\\
		&(\hat{\mathbb{S}}_0 J)(s,z):= \sup_{c_0\in (0,1)}\left\{c_0^{1-\gamma} + (1-c_0)^{1-\gamma} (\mathbb{S}_0 J)(s,z)\right\}\label{eq:ConsumptionOperatorPlus}
	\end{align}
	for the case with $\gamma\in(0,1)$ and
	\begin{align}
		&(\mathbb{S}_k J)(s,z):=\delta \inf_{\vecomega_k\in \Omega} \mathbb{E}\Big[ \big(R^p_{k+1}\big)^{1-\gamma}J(S_{k+1},Z_{k+1}) \mid S_k=s,Z_k=z\Big],\label{eq:PortfolioOperatorNegative}\\
		&(\hat{\mathbb{S}}_0 J)(s,z):= \inf_{c_0\in (0,1)}\left\{c_0^{1-\gamma} + (1-c_0)^{1-\gamma} (\mathbb{S}_0 J)(s,z)\right\}\label{eq:ConsumptionOperatorNegative}
	\end{align}
	for the case with $\gamma>1$.
	
	The following is a standard assumption imposing growth conditions.
	\begin{assumption}[Growth Condition]\label{as:GrowthCondition}
		\begin{enumerate}
			\item[(i)] The portfolio gross return has finite expected utility:
			\begin{align}\label{eq:AssBoundedPortfolioReturnUtility}
				\sup_{(s,z)\in {\cal S}}\sup_{\vecomega_k\in \Omega}\mathbb{E}\left[\big|U(R^p_{k+1})\big| \mid S_k=s,Z_k=z\right]<\infty.
			\end{align}
			\item[(ii)] For $\gamma\neq 1$, the following growth condition holds:
			\begin{align}
				\delta \bar R^{1-\gamma}<1,
				~~\text{where }  \bar R:= \begin{cases}
					\sup_{(s,z)\in {\cal S}}\sup_{\vecomega_k\in \Omega}U^{-1}\left(\mathbb{E}\Big[ U\big(R^p_{k+1}\big)\big]\mid S_k=s,Z_k=z\Big]\right), ~\gamma\in (0,1),\\
					\inf_{(s,z)\in {\cal S}}\sup_{\vecomega_k\in \Omega}U^{-1}\left(\mathbb{E}\Big[ U\big(R^p_{k+1}\big)\big]\mid S_k=s,Z_k=z\Big]\right), ~\gamma>1.
				\end{cases}
			\end{align}
		\end{enumerate}
	\end{assumption}
	
	\begin{theorem}\label{th:OptimalSolution}
		Suppose Assumptions \ref{as:PricingFunction}-\ref{as:GrowthCondition} hold. Then, the following are true:
		\begin{enumerate}
			\item[(i)] When $\gamma=1$, $\hat{\mathbb{T}}_0 \mathbb{T}_1\cdots \mathbb{T}_{N-1}$ is a contraction mapping on ${\cal X}$ and thus admits a unique fixed point $v_0^*$. Moreover, the unique optimal percentage of wealth consumed and invested is
			\begin{align}
				c_0^* = 1-\delta^{N},\quad \vecomega_k^* \in \argmax_{\vecomega_k \in \Omega}\mathbb{E}\Big[ U\big(R^p_{k+1}\big) \mid S_k=s,Z_k=z\Big],\label{eq:OptimalActionGammaEq1}
			\end{align}
			and for every $k=0,1,\dots, N-1$, $v^*_k$ is the LP's optimal utility starting from unit wealth at time $k$, where $v^*_k:=\mathbb{T}_k\cdots \mathbb{T}_{N-1}v^*_0$, $k=1,\dots, N-1$.
			\item[(ii)] When $\gamma\neq 1$, $\hat{\mathbb{S}}_0 \mathbb{S}_1\cdots \mathbb{S}_{N-1}$ admits a unique fixed point $\tilde v_0^*$ on ${\cal X}$ and for any $\tilde v_0\in {\cal X}$, the iterative sequence $\tilde v_0^n:= \left(\hat{\mathbb{S}}_0 \mathbb{S}_1\cdots \mathbb{S}_{N-1}\right)^n \tilde v_0$, $n\ge 1$ converges to $\tilde v_0^*$ exponentially, i.e., there exists $M>0$ and $r\in (0,1)$ such that $\|\tilde v_0^n-\tilde v^*_0\|\le Mr^n$ for all $n\ge 1$. Moreover, $\tilde v_0^*\ge 1$ and the unique optimal percentage of wealth consumed and invested is
			\begin{align}
				c_0^* = (1+\left(\mathbb{S}_0 \tilde v_1^*(s,z)\right)^{1/\gamma})^{-1},\quad \vecomega_k^* \in  \argmax_{\vecomega_k \in \Omega}\mathbb{E}\Big[ U\big(R^p_{k+1}\big)\tilde v^*_{k+1}(S_{k+1},Z_{k+1}) \mid S_k=s,Z_k=z\Big],\label{eq:OptimalActionGammaNeq1}
			\end{align}
			where $\tilde v_k^*:=\mathbb{S}_k\cdots \mathbb{S}_{N-1} \tilde v_0^*$, $k=1,\dots, N-1$, and
			\begin{align}
				v_k^*:=\frac{1}{1-\gamma}\left( \tilde v_k^*-\left(\mathbf 1_{k=0}+\frac{\delta^{N-k}}{1-\delta^N}\right)\right)\label{eq:OptimalValueGammaNeq1}
			\end{align}
			is the LP's optimal utility starting from unit wealth at time $k$, $k=0,1,\dots, N-1$.
		\end{enumerate}
	\end{theorem}
	
	Theorem \ref{th:OptimalSolution} shows the existence and uniqueness of the fixed point of the dynamic programming operators $\hat{\mathbb{T}}_0 \mathbb{T}_1\cdots \mathbb{T}_{N-1}$ for the case with $\gamma=1$ and $\hat{\mathbb{S}}_0 \mathbb{S}_1\cdots \mathbb{S}_{N-1}$ for the case with $\gamma\neq 1$. Moreover, the fixed point can be computed by the standard value iteration algorithm with exponential convergence rate. The optimal percentage of wealth consumed and invested by the LP is the optimizer of the optimization problem in the dynamic programming operators and thus can be solved in the course of the fixed point computation. This optimization problem is easy to solve because it is either a concave maximization problem or a convex minimization problem.\footnote{\label{ft:NoClosedForm}We, however, cannot derive closed-form formulae for the optimal consumption and investment. Indeed, even in the Merton's problem in the discrete-time setting, which can be regarded as a simplified version of the LP's consumption and investment problem in the present paper, the optimal solution cannot be derived in closed form and needs to be solved from an optimization problem; see for instance Section 3.3 in \citet{bertsekas1976dynamic}. This is in contrast to the continuous-time setting, in which closed-form solutions are available for some consumption and investment problems.}
	
	Note that when $\gamma=1$, the LP is myopic in that her optimal investment strategy is to maximize the one-period expected utility of return, as if her investment horizon were one period; see \eqref{eq:OptimalActionGammaEq1}. When $\gamma\neq 1$, the LP is not myopic and her optimal investment strategy depends on the reward-to-go from the next period; see \eqref{eq:OptimalActionGammaNeq1}.

	\subsubsection{The Case of CGMMM Pricing Function}\label{subsect:modelCGMMM}
	
	For the CGMMM pricing function \eqref{eq:CGMMM}, we recall the marginal exchange rate function \eqref{eq:CGMMMMarginalExchangeRateFun}. Thus, with $\alpha$ and $\beta$ standing for the deposit ratio on AMM and the investor's belief of the exchange rate, the ratio of the marginal exchange rate on AMM and the investor's belief of the exchange rate is given by $\eta/\big((1-\eta)\alpha \beta\big)$. Straightforward calculation yields the following form of $\varphi^A$ and $\varphi^B$ defined in Proposition \ref{le:TradersProblem}:
	\begin{align}
		(\varphi^A(\alpha,\beta),\varphi^B(\alpha,\beta)) = \left(\overline{\varphi}^A\left(\frac{\eta}{(1-\eta)\alpha \beta}\right),\overline{\varphi}^B\left(\frac{\eta}{(1-\eta)\alpha \beta}\right)\right),
		\label{eq:CGMMMOptimalTrading}
	\end{align}
	where
	\begin{align}
		\left(\overline{\varphi}^A(s),\overline{\varphi}^B(s)\right) := \begin{cases}
			\left(1-\left(s(1+f)\right)^{1-\eta},1-\left(s(1+f)\right)^{-\eta}\right), & \text{if }s<1/(1+f),\\
			(0,0), & \text{if }s \in \left[\frac{1}{1+f},(1+f)\right],\\
			\left(1-\left(s/(1+f)\right)^{1-\eta},1-\left(s/(1+f)\right)^{-\eta}\right), & \text{if }s>1+f.
		\end{cases}\label{eq:CGMMMOptimalTradingFunction}
	\end{align}
	
	Note that for the CGMMM pricing function, the optimal percentage trading amount $(\varphi^A(\alpha,\beta),\varphi^B(\alpha,\beta))$ depends on $\alpha$ and $\beta$ only through the pre-trade exchange rate ratio $\eta/\big((1-\eta)\alpha \beta\big)$. Straightforward calculation also yields the following form of $\phi^A$ and $\phi^B$ defined in \eqref{eq:PostTradeDeposit}:
	\begin{align}
		&\big(\phi^A(\alpha,\beta),\phi^B(\alpha,\beta)\big)=\left(\overline{\phi}^A\left(\frac{\eta}{(1-\eta)\alpha \beta}\right),\overline{\phi}^B\left(\frac{\eta}{(1-\eta)\alpha \beta}\right)\right),\label{eq:CGMMMDeposit}
	\end{align}
	where
	\begin{align}
		&\left(\overline{\phi}^A\left(s\right),\overline{\phi}^B\left(s\right)\right):= \begin{cases}
			\left(\left(s(1+f)\right)^{1-\eta}, (1+f)\left(s(1+f)\right)^{-\eta}-f\right), & \text{if }s<1/(1+f),\\
			(1,1), & \text{if }s \in \left[\frac{1}{1+f},(1+f)\right],\\
			\left((1+f)\left(s/(1+f)\right)^{1-\eta}-f,\left(s/(1+f)\right)^{-\eta}\right), & \text{if }s>1+f.
		\end{cases}\label{eq:CGMMMDepositFunction}
	\end{align}
	
	Note that with CGMMM, the transition function $H$ of the exchange rate ratio as defined in \eqref{eq:ExchangeRateRatioTransition} does not depend on $z$, and we denote it as $\bar{H}$, which takes the following form:
	\begin{align}
		\bar{H}(s) = \begin{cases}
			\Big(1+f\big(1-(s(1+f))^\eta\big)\Big)\times \frac{1}{1+f}, &\text{if } s<1/(1+f)\\
			s,& \text{if }s\in[1/(1+f),1+f],\\
			\frac{1}{1+f\big(1-(s/(1+f))^{\eta-1}\big)}\times (1+f), &\text{if }s>(1+f).
		\end{cases}\label{eq:CGMMMExchangeRateRatioTransitionFunction}
	\end{align}
	
	Recall the timeline of the $(k+1)$-th period in Figure \ref{fi:timeline}. Denote by $S_{(k,0)}$, $S_{(k,1)}$, $S_{(k,2)}$, and $S_{(k,3)}$ the ratios of the marginal exchange rate on AMM and the fundamental exchange rate at the end of Step 1, Step 2, Step 3, and Step 5, respectively; i.e.,
	\begin{align*}
		&S_{(k,i)}:=G(y^A_{(k,i)}/y^B_{(k,i)})/(p^A_k/p^B_k) = \frac{\eta }{(1-\eta)(y^A_{(k,i)}/y^B_{(k,i)})(p^A_k/p^B_k)},\quad i=0,1,2,\\
		&S_{(k,3)}:=G(y^A_{(k,3)}/y^B_{(k,3)})/(p^A_{k+1}/p^B_{k+1}) = \frac{\eta }{(1-\eta)(y^A_{(k,3)}/y^B_{(k,3)})(p^A_{k+1}/p^B_{k+1})}.
	\end{align*}
	Then, by \eqref{eq:PostTradeDepositRatioStep2}, \eqref{eq:PostTradeDepositRatioStep3}, \eqref{eq:PostTradeDepositRatioStep5}, and \eqref{eq:CGMMMDeposit}, we derive
	\begin{align}
		& S_{(k,1)} = \left[ (1-\xi_{k+1})+\xi_{k+1}\frac{\overline{\phi}^B\left(S_{(k,0)}/I_{k+1}\right)}{\overline{\phi}^A\left(S_{(k,0)}/I_{k+1}\right)}\right]S_{(k,0)},\label{eq:CGMMMExchangeRatioTransitStep2}\\
		&S_{(k,2)} = \frac{\overline{\phi}^B\left(S_{(k,1)}\right)}{\overline{\phi}^A\left(S_{(k,1)}\right)} S_{(k,1)} ,\label{eq:CGMMMExchangeRatioTransitStep3}\\
		& S_{(k,3)} = \frac{\overline{\phi}^B\left(S_{(k,2)}(R^B_{k+1}/R^A_{k+1})\right)}{\overline{\phi}^A\left(S_{(k,2)}(R^B_{k+1}/R^A_{k+1})\right)}(R^B_{k+1}/R^A_{k+1}) S_{(k,2)},\label{eq:CGMMMExchangeRatioTransitStep5}\\
		& S_{(k,0)}=S_{k},\quad S_{k+1} = S_{(k,3)}, \label{eq:CGMMMExchangeRatioTransitStep1}
	\end{align}
	where \eqref{eq:CGMMMExchangeRatioTransitStep1} is due to \eqref{eq:PostTradeDepositRatioStep1} implied by Assumption \ref{as:LPDeposit} and because the end of Step 5 of the $(k+1)$-th period is the beginning of the $(k+2)$-th period. Thus, \eqref{eq:CGMMMExchangeRatioTransitStep2}--\eqref{eq:CGMMMExchangeRatioTransitStep1} establish the transition of the exchange rate ratio: $S_{k+1}$ is a function of $S_k$ and random disturbances $(\xi_{t+1},I_{t+1},R^A_{k+1}/R^B_{k+1})$, but it is independent of the fundamental exchange rate $Z_k$ at the beginning of the $(k+1)$-th period.
	
	Similarly, by \eqref{eq:CGMMMDeposit} and \eqref{eq:AMMReturnStep2}--\eqref{eq:AMMReturnStep5}, we derive
	\begin{align}
		R^M_{(k+1,1)} 
		&= 1-\xi_{k+1} + \xi_{k+1}\frac{\frac{\eta}{1-\eta}\overline{\phi}^A\left(S_{(k,0)}/I_{k+1}\right)+S_{(k,0)}\overline{\phi}^B\left(S_{(k,0)}/I_{k+1}\right)}{ \frac{\eta}{1-\eta} + S_{(k,0)}},\label{eq:CGMMMAMMReturnStep2}\\
		R^M_{(k+1,2)} 
		&= \frac{\frac{\eta}{1-\eta}\overline{\phi}^A\left(S_{(k,1)}\right)+S_{(k,1)}\overline{\phi}^B\left(S_{(k,1)}\right)}{ \frac{\eta}{1-\eta} + S_{(k,1)}},\label{eq:CGMMMAMMReturnStep3}\\
		R^M_{(k+1,3)}&=R^{B}_{k+1}\frac{\frac{\eta}{1-\eta}  +S_{(k,2)}(R^{B}_{k+1}/R^{A}_{k+1})}{\frac{\eta}{1-\eta}(R^{B}_{k+1}/R^{A}_{k+1}) + S_{(k,2)}(R^{B}_{k+1}/R^{A}_{k+1})},\label{eq:CGMMMAMMReturnStep4}\\
		R^M_{(k+1,4)} 
		&= \frac{\frac{\eta}{1-\eta}\overline{\phi}^A\left(S_{(k,2)}(R^B_{k+1}/R^A_{k+1})\right)+ S_{(k,2)}(R^B_{k+1}/R^A_{k+1})\overline{\phi}^B\left(S_{(k,2)}(R^B_{k+1}/R^A_{k+1})\right)}{ \frac{\eta}{1-\eta}+ S_{(k,2)}(R^B_{k+1}/R^A_{k+1})}.\label{eq:CGMMMAMMReturnStep5}
	\end{align}
	Combining \eqref{eq:CGMMMAMMReturnStep2}-\eqref{eq:CGMMMAMMReturnStep5} with \eqref{eq:AMMGrossReturn} and \eqref{eq:CGMMMExchangeRatioTransitStep2}--\eqref{eq:CGMMMExchangeRatioTransitStep1}, we derive the return on AMM: $R^M_{k+1}/R^B_{k+1}$ is a function of $S_k$ and random disturbances $(\xi_{k+1},I_{k+1},R^A_{k+1}/R^B_{k+1})$, but it is independent of  $Z_k$.
	Therefore, for the CGMMM pricing function, given the exchange ratio $S_k$, the fundamental exchange rate $Z_k$ is irrelevant for the LP's decision. Thus, we can reduce the dimension of the state space.
	
	It is also worth mentioning that for the CGMMM pricing function, the exchange rate ratio takes the form
		$
		S_k = \frac{\eta}{(1-\eta)}\times \frac{y_k^B p_k^B}{y_k^A p_k^A},
		$
	where $(y_k^B p_k^B)/(y_k^A p_k^A)$ is the {\em relative portfolio weight of asset $B$ to asset $A$} in the liquidity pool. In view of Proposition \ref{le:TradersProblem}, the no-trading region for arbitrageurs can then be written as an interval for this relative portfolio weight. This explains the no-fee interval assumed in \citet{EvansAngerisChitra2021:OptimalFees}.
	
	\section{Optimal Design of AMM for CGMMM}\label{sect:numerics}
	
	In this section, we conduct numerical experiments to study the impact of asset return parameters on the LP's optimal action and the optimal design of AMM in the class of CGMMM pricing function \eqref{eq:CGMMM}.
	
	\subsection{Data Description and Parameter Estimation}\label{sec:para}
	We assume a bivariate log-normal distribution for $(R^A_{k+1},R^B_{k+1})$, the gross returns of assets $A$ and $B$, with the following parametric structure:
	\begin{align}\label{log_return}
		\left(\log R^{A}_{k+1},\log R^{B}_{k+1}\right)\sim \mathrm{\texttt{Normal}}\left(\begin{pmatrix}
			\mu_A\\ \mu_B
		\end{pmatrix},\begin{pmatrix}
			\sigma_A^2, & \rho\sigma_A\sigma_B\\
			\rho\sigma_A\sigma_B, & \sigma_B^2
		\end{pmatrix}\right).
	\end{align}
	Therefore, $\sigma_i$ and $\mu_i$ can be understood as the {\em volatility} and mean log return rate of asset $i$, respectively. The bivariate lognormal specification is a standard benchmark in portfolio choice models and allows the two crypto asset returns to be correlated. The logarithmic transformation of the fundamental exchange rate change in one period, i.e., $\log\big( (p^A_{k+1}/p^B_{k+1})/(p^A_{k}/p^B_{k})\big) = \log (R^{A}_{k+1}/R^{B}_{k+1})$, follows the normal distribution:
	\begin{align}
		\log (R^{A}_{k+1}/R^{B}_{k+1}) \sim \mathrm{\texttt{Normal}}\big(\mu_A-\mu_B, \sigma_A^2 + \sigma_B^2-2\rho \sigma_A\sigma_B \big). \label{eq:logExRateGrowth}
	\end{align}
	Therefore, $\sigma:=\sqrt{\sigma_A^2 + \sigma_B^2-2\rho \sigma_A\sigma_B}$ is the {\em volatility of the exchange rate}.

	Recall that the arrival of investors, $\xi_{k+1}$, follows the Bernoulli distribution with arrival probability $\alpha$. We assume that $I_{k+1}$, which represents the heterogeneity of investors' beliefs of the exchange rate, follows a log-normal distribution; i.e., $\log I_{k+1}\sim \mathrm{\texttt{Normal}}\left(-\frac{1}{2}\sigma_I^2,\sigma_I^2\right)$.
	Assume that $I_{k+1}$, $(R^A_{k+1},R^B_{k+1})$, and $\xi_{k+1}$ are independent.
	
	We set the length of one period to be 8 hours and assume the LP consumes every 24 hours, so that $N=3$ in our setting. We set the LP's risk aversion $\gamma=2$.  We take ETH and BTC as assets $A$ and $B$, respectively, and estimate the market parameters $\mu_A$, $\mu_B$, $\sigma_A$, $\sigma_B$, and $\rho$ using the 8-hour trading data of ETH/USDT and BTC/USDT from July 1st, 2022 to June 30, 2023 obtained from Binance Market Data (\url{https://data.binance.vision}). We set $\eta=0.5$ corresponding to the CPMM used in Uniswap and $f=0.5\%$ in the middle of the range of the fee levels on Uniswap v3 (0.05\%, 0.3\%, and 1\%). The parameters $\sigma_I$ and $\alpha$ govern the non-arbitrage order flow: $\sigma_I$ measures the dispersion of liquidity traders' beliefs around the CEX-implied exchange rate, and $\alpha$ is their arrival probability in one period. We use $\sigma_I=0.02$ and $\alpha=0.5$ as moderate benchmark values that generate liquidity-trader demand without making this channel dominate the arbitrage channel. In Section \ref{subsubse:sigmaIAlpha}, we will show that our results are robust with respect to reasonable values of $\sigma_I$ and $\alpha$. We summarize the parameter values in Table \ref{tbl:parameters}.
	\begin{table}[htp!]
		\caption{Default parameter values.}\label{tbl:parameters}
		\centering
		\footnotesize
		\begin{tabular}{@{}ccccccccccccc@{}}
			\toprule
			\multicolumn{9}{c}{Market Parameters} &\multicolumn{2}{c}{LP's Attributes} & \multicolumn{2}{c}{AMM Parameters} \\
			\cmidrule(lr){1-9} \cmidrule(lr){10-11} \cmidrule(lr){12-13}
			$\delta$  & $R_f$ & $\mu_A$ & $\mu_B$ & $\sigma_A$ & $\sigma_B$ & $\rho$ & $\alpha$ & $\sigma_I$ & $N$ & $\gamma$ & $\eta$ & $f$ \\
			\midrule
			0.998   & 1+0.002\% &   0.05\%   & 0.038\%&   1.99\% & 1.52\%  &  0.8642 & 0.5  & 2\% &   3 &  2 & 0.5 & 0.5\%  \\
			\bottomrule
		\end{tabular}%
	\end{table}

	\subsection{Opportunity Cost and Price Slippage}\label{subse:OpportunityCostSlippage}
	To understand the subsequent numerical results, we first quantify the opportunity cost and price slippage in our model with the CGMMM.
	
	\subsubsection{Opportunity Cost}
	Recall that the opportunity cost arises from the inefficient allocation of the two assets by the LP in the liquidity pool on DEX compared to the efficient allocation that the LP could have achieved by investing on CEX instead. According to \eqref{eq:CGMMMMarginalExchangeRateFun}, the marginal exchange rate on DEX is $\eta y^B/((1-\eta) y^A)$, which has to stay close to the fundamental exchange rate $p^A/p^B$ at the beginning of each period due to the presence of arbitrageurs. As a result, the dollar-amount ratio of asset $A$ over asset $B$ held by the LP in the liquidity pool satisfies
	\begin{align}
		\frac{y^Ap^A}{y^Bp^B}\approx \frac{\eta}{1-\eta}.\label{eq:RatioOnDEX}
	\end{align}
	This ratio can be different from the optimal ratio that the LP could have achieved by investing on CEX, where this optimal ratio can be computed by maximizing LP's expected utility assuming that she invests on CEX only. Given the assumption of i.i.d. asset returns and CRRA utility function, it is well known that the optimal portfolio on CEX is the one that maximizes the LP's expected utility of one-period portfolio return:
	\begin{align}
		\mathbb{E}[U( \omega^A_k R^A_{k+1} + \omega^B_k R^B_{k+1}+ \big(1-\omega^A_k-\omega^B_k)R_f)].\label{eq:EUPortfolioReturnCEX}
	\end{align}
	Denote by $(\hat \omega^A,\hat\omega^B)$ the optimal portfolio maximizing \eqref{eq:EUPortfolioReturnCEX}. Then, the optimal ratio, which is referred to be the {\em efficient allocation ratio} of asset $A$ over asset $B$, is $\hat \omega^A/\hat\omega^B$. The opportunity cost arises when the allocation ratio in the liquidity pool, $\frac{\eta}{1-\eta}$, differs from the efficient allocation ratio, $\hat \omega^A/\hat\omega^B$.
	
	\subsubsection{Price Slippage}
	
	The price slippage is the difference between the marginal exchange rate on AMM and the average exchange rate of trading a positive amount of an asset on AMM. Under CGMMM, it is straightforward to derive that the {\em relative price slippage} for acquiring $D_A>0$ amount of asset $A$ at marginal exchange rate $E_0$ is
	$
	\frac{-D^B/D^A - E_0}{E_0} = \frac{1-\eta}{\eta}\frac{\big(1-D^A/y^A\big)^{\eta/(\eta-1)}-1}{D^A/y^A}-1,
	$
	where $D^B$ is determined by the pricing equation \eqref{AMMrule}, and therefore $-D^B/D^A$ stands for the average exchange rate. It is straightforward to verify that the above relative price slippage is strictly increasing in the trading size $D^A/y^A$ and in $\eta$. Similarly, the relative price slippage for acquiring $D_B>0$ amount of asset $B$ at marginal exchange rate $E_0'$ is
	$
	\frac{-D^A/D^B - E_0'}{E_0'} = \frac{\eta}{1-\eta}\frac{\big(1-D^B/y^B\big)^{(1-\eta)/\eta}-1}{D^B/y^B}-1,
	$
	which is strictly increasing in the trading size $D^B/y^B$ and strictly decreasing in $\eta$.
	
	When the price slippage becomes higher, it becomes more costly to trade for both the liquidity traders and arbitrageurs. Thus, fixing the asset deposit amount in the liquidity pool, higher price slippage reduces both the LP's fee revenue from the liquidity traders and her loss to the arbitrageurs. Note that, in the CGMMM, the parameter $\eta$ determines the price slippage. On the other hand, $\eta$ determines the asset deposit amount in the liquidity pool. Indeed, as discussed above, the LP takes the allocation as specified in \eqref{eq:RatioOnDEX}; i.e., $\eta$ of the liquidity pool value is in asset $A$, and the remaining is in asset $B$. The deposit amount of an asset in the liquidity pool also affects the LP's fee revenue from the liquidity traders and her loss to the arbitrageurs: with a larger amount of asset $A$ in the pool, the price slippage for acquiring asset $A$ has a larger impact than that for acquiring asset $B$ does. In the following, we compute the aggregate effect of $\eta$ on the LP's fee revenue from the liquidity traders and her loss to the arbitrageurs.
	
	\subsubsection{LP's Fee Revenue from Liquidity Traders and Loss to Arbitrageurs}
	
	The LP's loss to the arbitrageurs is the arbitrageurs's profit, i.e., the optimal value of problem \eqref{eq:TraderProblem} with $(a,b)$ to be the updated fundamental asset prices $(\tilde p^A,\tilde p^B)$ after the price shock.\footnote{\label{ft:LossZeroSumGame}Here, we are referring to the LP's loss in Step 3 or Step 5 of the timeline of trading activities as illustrated in Figure \ref{fi:timeline}, with the value of the LP's assets evaluated at the CEX price. This is exactly the definition of impermanent loss \citep{EngelHerlihy2022:PresentationAndPublication,AngerisChitra2020:ImprovedPriceOracles}. As a result, the game between the LP and arbitrageurs is zero sum. However, the LP's total gain or loss in the whole period, namely, from Step 1 throughout Step 5 of the trading timeline, which is evaluated at the CEX prices at the beginning and end of the period, is not the opposite of the arbitrageur's profit, because the former also includes the fee revenues from the liquidity traders as well as the asset appreciation or depreciation (i.e., the change of CEX prices) in this period. Note that the asset appreciation or depreciation can be hedged by trading on the CEX, so the LP's total gain or loss in the whole period, after hedging, is the LP's fee revenues from the liquidity traders net her loss to the arbitrageurs.} 
	Recall Proposition \ref{le:TradersProblem} and the discussion in Section \ref{subsect:modelCGMMM}, and recall that before the arbitrageurs trade, the marginal exchange rate on DEX, $\frac{\eta}{(1-\eta)(y^A/y^B)}$, is approximately equal to the fundamental exchange rate before the trade, namely $p^A/p^B$, where $p^i$ and $y^i$ stand for the fundamental price and deposit amount of asset $i$ in the pool before the trade. Then, the LP's loss to the arbitrageurs, denoted as IL (aka impermanent loss, net the trading fee), is approximately equal to
	\begin{align*}
		&\mathrm{IL}\approx \tilde p^A(1+f\mathbf 1_{\tilde R^{-1} >(1+f)})\bar \varphi^{A}(\tilde R^{-1}) y^A+ \tilde p^B (1+f\mathbf 1_{\tilde R^{-1} <1/(1+f)})\bar \varphi^{B}(\tilde R^{-1})y^B\\
		& \approx  \left[ \eta R^A (1+f\mathbf 1_{\tilde R^{-1} >(1+f)})\bar \varphi^{A}(\tilde R^{-1}) + (1-\eta)R^B(1+f\mathbf 1_{\tilde R^{-1} <1/(1+f)})\bar \varphi^{B}(\tilde R^{-1})\right]\cdot (y^Ap^A+y^Bp^B),
	\end{align*}
	where $R^i:=\tilde p^i/p^i$ denotes the gross return of asset $i$ due to the price shock, $\tilde R:=(\tilde p^A/\tilde p^B)/(p^A/p^B)=R^A/R^B$ denotes the fundamental exchange rate change due to the price shock, and the second approximate equality is because before the trade, the marginal exchange rate on AMM, $\frac{\eta}{(1-\eta)(y^A/y^B)}$, is approximately equal to the fundamental exchange rate $p^A/p^B$ and thus $p^Ay^A/(p^Ay^A+p^By^B)\approx \eta$. With a sufficiently small $f$ and recalling \eqref{eq:CGMMMOptimalTradingFunction} and $R^A/R^B = \tilde R$, we derive
	\begin{align*}
		\frac{\mathrm{IL}}{y^Ap^A+y^Bp^B}&\approx  R^B\left[\eta \tilde R +(1-\eta)-\tilde R^\eta\right]\left(\mathbf 1_{\tilde R>(1+f)}+\mathbf 1_{\tilde R<1/(1+f)}\right).
	\end{align*}
	Note that $h(\eta,\tilde R):=\eta \tilde R +(1-\eta)-\tilde R^\eta$ is convex in $\tilde R$ with $h(\eta,1)=0$, showing that ex post the price shock, the LP's loss to arbitrageurs becomes larger when $\tilde R$ is further away from 1, i.e., when the fundamental exchange rate change is larger. On the other hand, $h(\eta,\tilde R)$ is concave in $\eta$ with $h(0,\tilde R)=h(1,\tilde R)=0$. For any fixed $\tilde R\neq 1$, the maximizer of $h(\eta,\tilde R)$ in $\eta$ is attained at $\ln [(\tilde R-1)/(\ln \tilde R)]/\ln \tilde R$, which converges to $1/2$ when $\tilde R$ goes to $1$. Thus, ex ante the price shock, the LP's expected loss to arbitrageurs vanishes when $\eta=0$ or 1, and when the price shock has a small magnitude, i.e., when $\tilde R$ takes values around 1, the expected loss is maximized when $\eta$ is around 1/2.
	
	With $\tilde R$ being close to 1, which is the case when the length of one period is short (see the parameter values in Table \ref{tbl:parameters}), we have the following approximation based on Taylor expansion:
	\begin{align}
		\frac{\mathrm{IL}}{y^Ap^A+y^Bp^B}&\approx  \frac{1}{2}R^B\eta(1-\eta)(\tilde R-1)^2\left(\mathbf 1_{\tilde R>(1+f)}+\mathbf 1_{\tilde R<1/(1+f)}\right).\label{eq:ILApprox}
	\end{align}
	\citet{MilionisMoallemiRoughgardenZhang2024:AutomatedMarketMaking} consider a continuous-time setting with asset $B$ as the num\'eraire and thus $R^B\equiv 1$. In this setting, $\tilde R$ represents the fundamental exchange rate change in an infinitesimally small period of time and thus has a mean of 1, and the expectation of $(\tilde R-1)^2$ is proportional the fundamental exchange rate volatility. Thus, \citet{MilionisMoallemiRoughgardenZhang2024:AutomatedMarketMaking} derive that the LP's expected loss to arbitrageurs in their setting is proportional to $\eta(1-\eta)$ and the volatility of the fundamental exchange rate; see Example 2 therein.

	On the other hand, the LP's fee revenue from the liquidity traders amounts to
	\begin{align*}
		\mathrm{FEE}:=\left[-p^AD^A\mathbf 1_{D^{A}<0}-p^BD^{B}\mathbf 1_{D^{B}<0}\right]f \xi ,
	\end{align*}
	where $\xi$ is the Bernoulli random variable denoting the arrival of the liquidity traders, $p^A$ and $p^B$ stand for the fundamental asset prices, and
	\begin{align*}
		(D^A,D^B) =  \left(\overline{\varphi}^A\left(\frac{\eta}{(1-\eta)(a/b) (y^A/y^B)}\right)y^A,\overline{\varphi}^B\left(\frac{\eta}{(1-\eta)(a/b) (y^A/y^B)}\right)y^B\right)
	\end{align*}
	is the optimal trading amount of the liquidity traders with $a$ and $b$ representing their valuation of assets $A$ and $B$, respectively, $y^A$ and $y^B$ standing for the pre-trade deposit amount of assets $A$ and $B$ in the pool, and $\overline{\varphi}^A$ and $\overline{\varphi}^B$ as given by \eqref{eq:CGMMMOptimalTradingFunction}. Before the trade, the marginal exchange rate on DEX, $\frac{\eta}{(1-\eta)(y^A/y^B)}$, is aligned with the fundamental exchange rate $p^A/p^B$. Thus, by \eqref{eq:LiquidityTradersBelief}, the LP's fee revenue from the liquidity traders is approximately equal to
	\begin{align*}
		&\mathrm{FEE}:\approx\left[-p^Ay^A\overline{\varphi}^A\left(I^{-1}\right)\mathbf 1_{D^{A}<0}-p^By^B\overline{\varphi}^B\left(I^{-1}\right)\mathbf 1_{D^{B}<0}\right]f\xi\\
		&\approx
		\left[\eta\big(\left(I^{-1}/(1+f)\right)^{1-\eta}-1\big)\mathbf 1_{I^{-1}>1+f}+(1-\eta)\big(\left(I^{-1}(1+f)\right)^{-\eta}-1\big)\mathbf 1_{I^{-1}<1/(1+f)}\right]\cdot (p^Ay^A+p^By^B) f\xi,
	\end{align*}
	where $I$ denotes the ratio of the liquidity traders' belief of the exchange rate and the fundamental exchange rate and the approximate equality holds due to \eqref{eq:CGMMMOptimalTradingFunction} and because the marginal exchange rate on DEX, $\frac{\eta}{(1-\eta)(y^A/y^B)}$, is approximately equal to the fundamental exchange rate $p^A/p^B$. We observe that the fee revenue vanishes when $\eta$ approaches to 0 or 1. The value of $\eta$ that maximizes the LP's expected fee revenue ex ante the arrival of the liquidity traders depends on the distribution of $I$. With further approximation, we have
	\begin{align}
		\frac{\mathrm{FEE}}{p^Ay^A+p^By^B}:&\approx  \left[\eta\big(I^{\eta-1}-1\big)\mathbf 1_{I<(1+f)^{-1}}+(1-\eta)\big(I^{\eta}-1\big)\mathbf 1_{I>(1+f)}\right] f\xi\notag\\
		&\approx \eta(1-\eta)|I-1|\big(1_{I<(1+f)^{-1}}+\mathbf 1_{I>(1+f)}\big) f\xi,\label{eq:FeeApprox}
	\end{align}
	where the first approximation is because $f$ is small and the second approximation is due to Taylor expansion. Therefore, the approximated fee is proportional to $\eta(1-\eta)$.

	To summarize, both the LP's loss to the arbitrageurs and her fee revenue from the liquidity traders vanish when $\eta$ approaches to 0 or 1 and become larger when $\eta$ take intermediate values. Thus, there is a tradeoff between a larger fee revenue and a lower loss. The fee revenue depends on the arrival rate of the liquidity traders and the heterogeneity of their belief but does not depend on the size of fundamental exchange rate shock. The loss to the arbitrageurs, on the other hand, is determined by the size of fundamental exchange rate shock but is independent of the arrival rate and belief of the liquidity traders. Combining \eqref{eq:ILApprox} and \eqref{eq:FeeApprox}, the LP's expected fee revenue from the liquidity traders net her expected loss to the arbitrageurs, per unit liquidity pool value, is approximately equal to
	\begin{align}
	\eta(1-\eta)\Big[ & \alpha f\expect\big[ |I-1|\big(\mathbf 1_{I<(1+f)^{-1}}+\mathbf 1_{I>(1+f)}\big) \big]\notag \\
	&	-\frac{1}{2}\expect\left[R^B(\tilde R-1)^2\left(\mathbf 1_{\tilde R>(1+f)}+\mathbf 1_{\tilde R<1/(1+f)}\right)\right] \Big],\label{eq:ILNetFeeApprox}
	\end{align}
	where $\alpha$ is the probability of the arrival of the liquidity traders in one period.
	From this approximation, we conclude that the LP's net profit of liquidity provision on DEX is positive and, consequently, she is willing to invest on DEX when
	\begin{align}
		\alpha f\expect\big[ |I-1|\big(\mathbf 1_{I<(1+f)^{-1}}+\mathbf 1_{I>(1+f)}\big) \big]>\frac{1}{2}\expect\left[R^B(\tilde R-1)^2\left(\mathbf 1_{\tilde R>(1+f)}+\mathbf 1_{\tilde R<1/(1+f)}\right)\right].\label{eq:FeeDominateIL}
	\end{align}

	\subsection{LP's Optimal Investment Strategy}\label{sect:sub:LPallocation}
	We use value iteration to solve the DPE for the LP's optimal value and strategy. Recall that the LP cannot borrow or short-sell; see the portfolio constraint $\Omega$ in \eqref{eq:PortfolioConstraints}. In the following, we report the LP's optimal value and strategy at time 0, when she needs to choose both the investment and consumption strategies. Recall that with the CGMMM pricing function, the state variable of the LP's problem reduces to the ratio of the marginal exchange rate on DEX and the fundamental exchange rate, which always lies in the interval $[1/(1+f),1+f]$.
	
	Figure \ref{state} plots LP's optimal percentage investment on DEX and CEX with respect to the state variable. We observe that the optimal percentage investment on DEX is higher when the state variable, namely, the exchange rate ratio, is closer to 1. Indeed, in this case, there are fewer arbitrage opportunities taken by the arbitrageurs after a price shock. Therefore, the LP suffers a smaller amount of loss from the arbitrageurs and is more willing to invest on DEX. We also observe that the CEX investment is concentrated in asset $A$. This is a parameter-dependent corner solution, not a general prediction of the model. The LP jointly chooses her DEX liquidity position, CEX investments, and risk-free investment. In the baseline calibration as in Table \ref{tbl:parameters}, asset $A$ has a higher estimated mean return than asset $B$, and the two risky assets are highly correlated. Since the DEX liquidity position already gives exposure to both assets through the pool composition, an additional CEX position in asset $B$ provides little diversification benefit relative to the available alternatives. Therefore, the LP prefers asset $A$ to asset $B$ on the CEX, and under the no-short sale constraint, the LP does not investment in asset $B$. Under different return parameters, the LP may invest in asset $B$ on the CEX, as shown in the comparative statics in Figure \ref{fi:InvestmentVSReturnParameters}.
	
	\begin{figure}
		\centering
		\includegraphics[width = 0.55\textwidth]{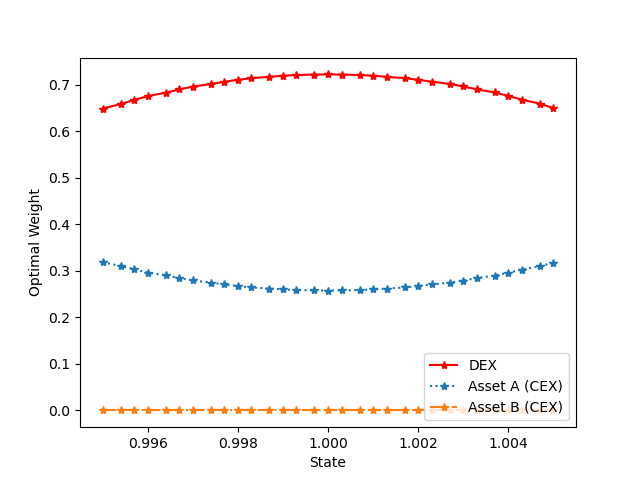}
		\caption {LP's optimal percentage investment on DEX (solid line) and in asset $A$ on CEX (dotted line) and in asset $B$ on CEX (dash-dotted line) with respect to the ratio of the marginal exchange rate on DEX and the fundamental exchange rate. The pricing function is CGMMM with $\eta=1/2$. The portfolio constraint is given by \eqref{eq:PortfolioConstraints}. Other model parameter values are given in Table \ref{tbl:parameters}.}\label{state}
	\end{figure}  
	
	Next, we vary the values of the asset return parameters $\mu_i$ and $\sigma_i$, $i\in \{A,B\}$ and study their impact on the LP's optimal strategy. We can show that there exists a unique stationary distribution for the state variable, namely, the ratio of the marginal exchange rate on DEX and the fundamental exchange rate; see, for instance,  \citet{Harris1960:TheExistenceOfStationary}. We then compute the LP's expected percentage investment on DEX and CEX, where the expectation is taken under the stationary distribution of the state variable, and plot it with respect to $\mu_A$, $\mu_B$, $\sigma_A$, and $\sigma_B$ in the top-left, top-right, bottom-left, and bottom-right panels, respectively, of Figure \ref{fi:InvestmentVSReturnParameters}. We consider 110 AMM pools on Uniswap v3 and estimate the return parameters of the assets in each pool using price date from Binance. We then consider the value of each parameter in Figure \ref{fi:InvestmentVSReturnParameters} ranging from the minimum value to the maximum value of the estimated parameters over these 110 AMM pools. Thus, the plot in Figure \ref{fi:InvestmentVSReturnParameters} covers a wide range of reasonable return parameter values.

	The LP's investment on DEX generates a fee revenue from the liquidity traders, incurs a loss to the arbitrageurs, and incurs an opportunity cost of holding assets in the liquidity pool. As discussed in Section \ref{subse:OpportunityCostSlippage}, with $\eta$ fixed at $1/2$, both a sufficiently small value and a sufficiently large value of $\mu_A$ lead to a large opportunity cost because the efficient allocation ratio becomes very different from the allocation ratio $\eta/(1-\eta)$ in the liquidity pool. On the other hand, a sufficiently small or large value of $\mu_A$ also leads to a large loss to the arbitrageurs because $R^B(\tilde R-1)^2$ becomes large; see \eqref{eq:ILApprox}. Therefore, it becomes unfavourable for the LP to invest on DEX when $\mu_A$ is very large or very small. When $\mu_A$ takes intermediate values and the other parameter values are given as in Table \ref{tbl:parameters}, the LP's fee revenue from the liquidity traders net the loss to the arbitrageurs is positive and her opportunity cost is low. Therefore, the LP would like to invest on DEX. This explains the top left panel of Figure \ref{fi:InvestmentVSReturnParameters}. The top right panel can be explained similarly.
	
	From \eqref{eq:ILNetFeeApprox}, we can see that when the fundamental exchange rate volatility, which is approximately equal to $\sqrt{\expect[(\tilde R-1)^2]}$ given that $\expect[\tilde R]\approx 1$, becomes larger, the LP's loss to the arbitrageurs also becomes larger. According to \eqref{eq:logExRateGrowth}, the fundamental exchange rate volatility is decreasing in $\sigma_A$ for $\sigma_A\le \rho\sigma_B=1.314\%$ and increasing in $\sigma_A$ for $\sigma_A\ge \rho\sigma_B=1.314\%$ with the values of $\sigma_B$ and $\rho$ in Table \ref{tbl:parameters}. Therefore, the LP's loss to the arbitrageurs is first decreasing and then increasing in $\sigma_A$. On the other hand, the LP's investment in asset $A$ in the efficient allocation is decreasing in $\sigma_A$. Thus, with $\eta=1/2$, the efficient allocation ratio deviates more from the allocation ratio $\eta/(1-\eta)$ in the liquidity pool when $\sigma_A$ becomes very small or very large. Consequently, the opportunity cost is first decreasing and then increasing in $\sigma_A$. Combining the effect of the LP's loss to the arbitrageurs and her opportunity cost, we conclude that she is willing to invest on DEX if and only if $\sigma_A$ takes intermediate values. This conclusion is confirmed in the bottom left panel of Figure \ref{fi:InvestmentVSReturnParameters}. The bottom right panel can be explained similarly.

	\begin{figure}
		\centering
		\begin{minipage}{0.35\textwidth}
			\includegraphics[width=1.15\textwidth]{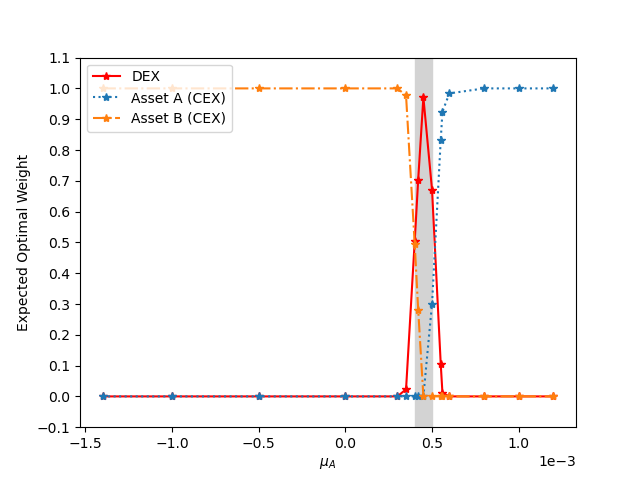}
		\end{minipage}
		\hspace{1em}
		\begin{minipage}{0.35\textwidth}
			\includegraphics[width=1.15\textwidth]{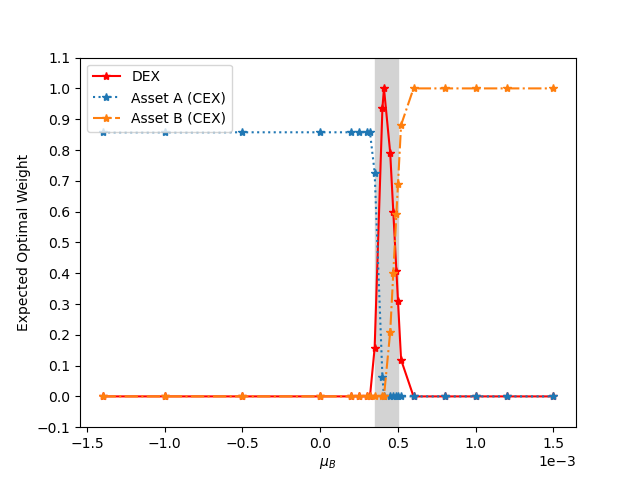}
		\end{minipage}
		\\
		\begin{minipage}{0.35\textwidth}
			\includegraphics[width=1.15\textwidth]{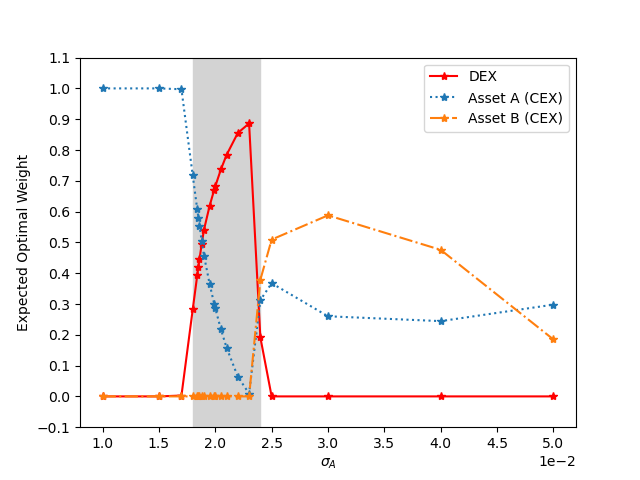}
		\end{minipage}
		\hspace{1em}
		\begin{minipage}{0.35\textwidth}
			\includegraphics[width=1.15\textwidth]{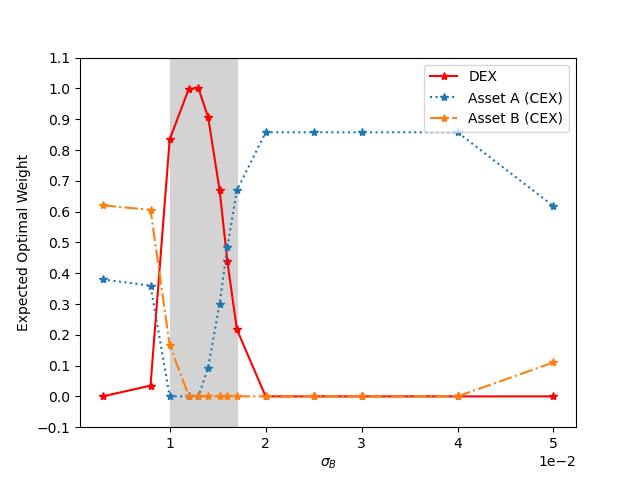}
		\end{minipage}
		\caption{LP's expected optimal percentage investment on DEX (solid line), in asset $A$ on CEX (dotted line), and in asset $B$ on CEX (dash-dotted line) with respect to $\mu_A$ (top left panel), $\mu_B$ (top right panel), $\sigma_A$ (bottom left panel), and $\sigma_B$ (bottom right panel). The pricing function is CGMMM, the portfolio constraint is given by \eqref{eq:PortfolioConstraints}, and the model parameter values are given in Table \ref{tbl:parameters}.}\label{fi:InvestmentVSReturnParameters}
	\end{figure}

	\subsection{Optimal Design of Trading Fee and Pricing Function}\label{subse:OptimalTradingFeePricingFunction}
	In the following, we study the optimal choice of the unit trading fee $f$ and the pricing function parameter $\eta$ and how the choice depends on the asset return parameters. We first compute the LP's optimal value with respect to $f$ and $\eta$, averaged under the stationary distribution of the state variable. We use the parameter values in Table \ref{tbl:parameters} and plot in Figure \ref{fi:LPValueVSfeta} the LP's optimal value with respect to $f$ (left panel) and to $\eta$ (right panel). With these parameter values, the LP invests on DEX. As discussed in Section \ref{subse:OpportunityCostSlippage}, a higher unit trading fee $f$ leads to a higher fee revenue of the LP from the liquidity traders per trade and a lower net loss to the arbitrageurs, and at the same time it leads to a smaller number of trades from the liquidity traders. Thus, the total fee revenue from the liquidity traders net the loss to the arbitrageurs is first increasing and then decreasing in $f$, explaining the bell shape of the LP's optimal value in the left panel of Figure \ref{fi:LPValueVSfeta}. On the other hand, as long as \eqref{eq:FeeDominateIL} holds, the LP's fee revenue from the liquidity traders net the loss to the arbitrageurs is positive and becomes larger when $\eta$ is closer to $1/2$; see \eqref{eq:ILNetFeeApprox}. The opportunity cost becomes higher when $\eta$ is very close 0 or 1 because, in this case, the allocation ratio in the liquidity pool, $\eta/(1-\eta)$, is further away from the efficient allocation ratio. Therefore, the LP's total value should peak when $\eta$ is in the middle range of $(0,1)$. This is confirmed by the plot in the right panel of Figure \ref{fi:LPValueVSfeta}.
	
	\begin{figure}
		\centering
		\includegraphics[width=0.35\textwidth]{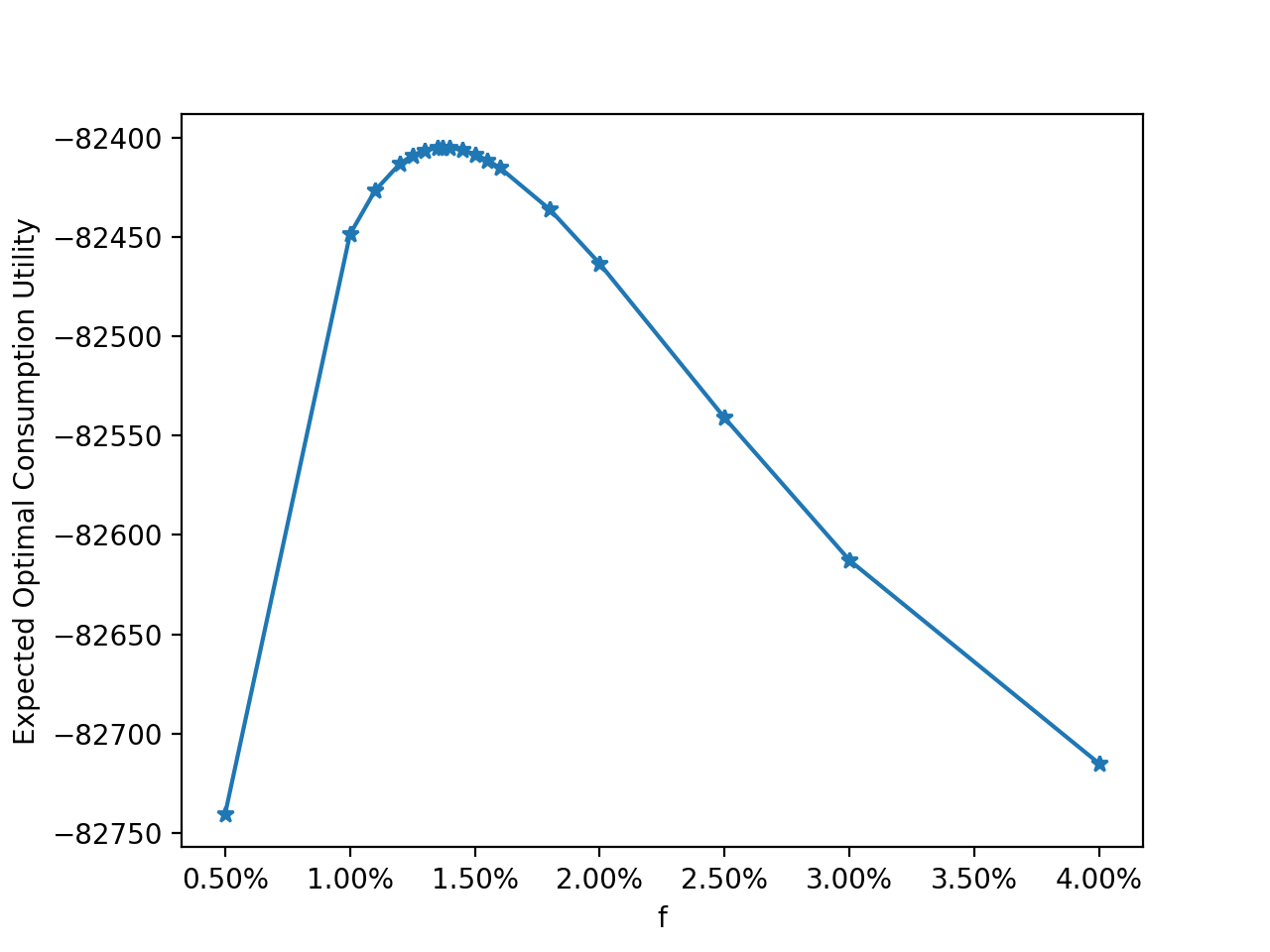}
		\includegraphics[width=0.35\textwidth]{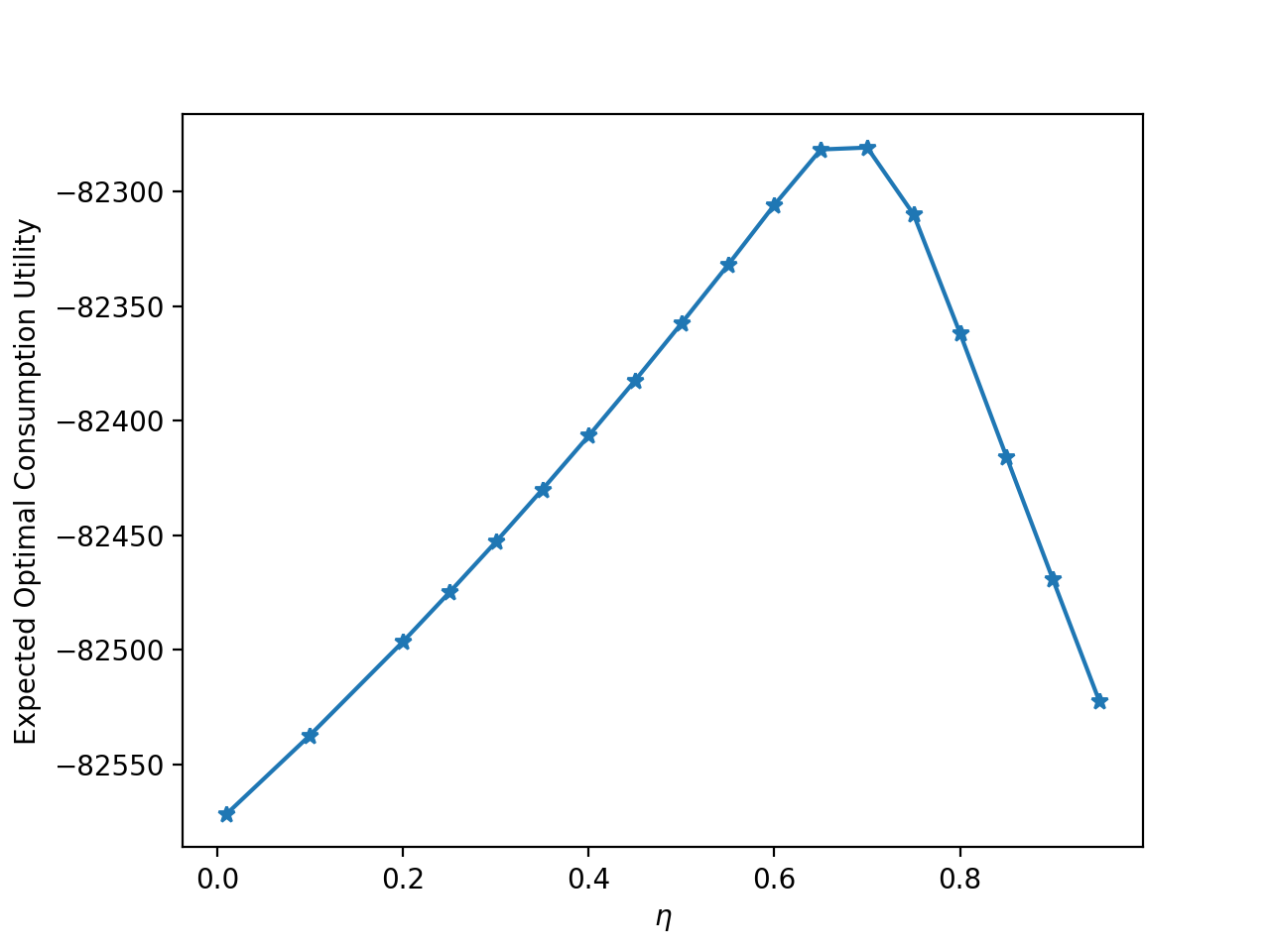}
		\caption{LP's expected optimal utility with respect to the unit trading fee $f$ (left panel) and pricing function parameter $\eta$ (right panel). The pricing function is CGMMM. The portfolio constraint is given by \eqref{eq:PortfolioConstraints} so borrowing or short-selling is not allowed. Other model parameter values are given in Table \ref{tbl:parameters}.}
		\label{fi:LPValueVSfeta}
	\end{figure}

	\subsubsection{Optimal Unit Trading Fee}\label{subsect:fee}
	
	We first plot the optimal $f$ with respect to $\mu_A$ (top left panel), $\mu_B$ (top right panel), $\sigma_A$ (middle left panel), and $\sigma_B$ (middle right panel) in Figure \ref{fi:FeeVSReturn}. Because the value of $f$ is relevant only when the LP invests on DEX, we only plot $f$ for the return parameter values with which the LP invests on DEX.
	
	We can observe from the top panels of Figure \ref{fi:FeeVSReturn} that the optimal fee is insensitive with respect to $\mu_A$ and $\mu_B$. This is because the fee revenue from the liquidity traders does not depend on the asset returns and the mean asset returns have little impact on the LP's losses to the arbitrageurs; see \eqref{eq:ILNetFeeApprox}. On the other hand, the optimal fee $f$ is increasing in $\sigma_A$. This is because the fundamental exchange rate volatility $\sqrt{\sigma_A^2+\sigma_B^2-2\rho \sigma_A\sigma_B}$ is increasing in $\sigma_A$ for $\sigma_A\ge \rho\sigma_B=1.314\%$ with the values of $\sigma_B$ and $\rho$ in Table \ref{tbl:parameters}. A higher volatility of the exchange rate leads to a larger loss to the arbitrageurs; see \eqref{eq:ILNetFeeApprox}. Thus, the LP would choose a higher unit trading fee to mitigate the loss to the arbitrageurs. Similarly, the fundamental exchange rate volatility is decreasing in $\sigma_B$ for $\sigma_B\le \rho\sigma_A= 1.712\%$ with the values of $\sigma_A$ and $\rho$ in Table \ref{tbl:parameters}. Therefore, a smaller value of $\sigma_B$ leads to a higher fundamental exchange rate volatility and the LP prefers a higher unit trading fee. In the bottom left panel, we plot the optimal $f$ with respect to $\sigma_A$ with the fundamental exchange rate volatility $\sigma$ fixed; i.e., we vary $\sigma_A$ and $\sigma_B$ together by keeping the fundamental exchange rate volatility unchanged. We find that the optimal $f$ is insensitive with respect to $\sigma_A$. Similarly, the bottom right panel shows that the optimal $f$ is insensitive with respect to $\sigma_B$ if we fix $\sigma$. Therefore, the optimal unit trading fee depends on the asset return parameters only through the fundamental exchange rate volatility $\sigma$.
	
	\begin{figure}
		\centering
		\begin{minipage}{0.35\textwidth}
			\includegraphics[width=1.15\textwidth]{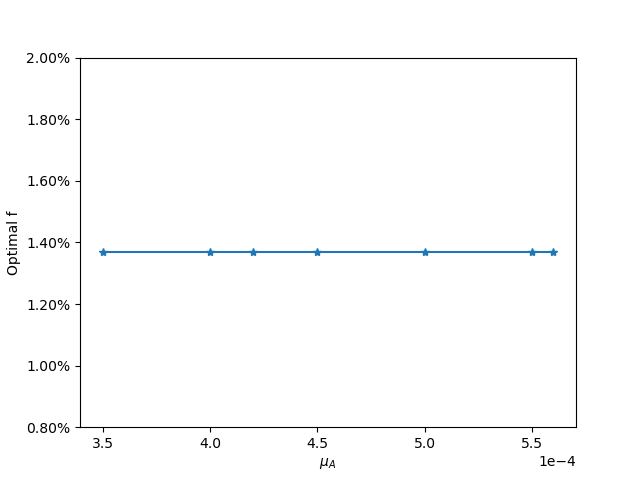}
		\end{minipage}
		\hspace{1em}
		\begin{minipage}{0.35\textwidth}
			\includegraphics[width=1.15\textwidth]{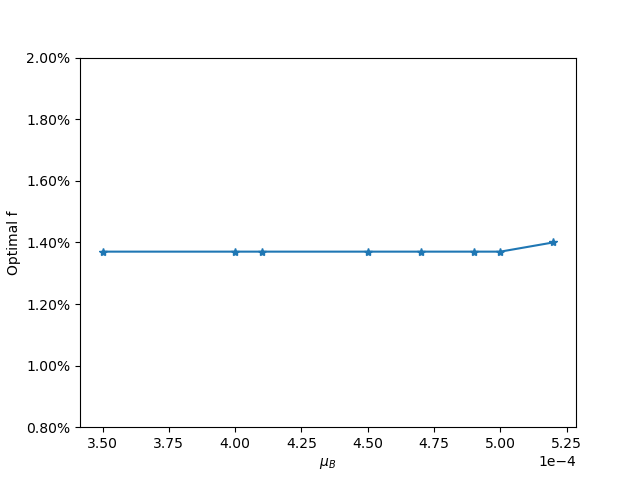}
		\end{minipage}
		\\
		\begin{minipage}{0.35\textwidth}
			\includegraphics[width=1.15\textwidth]{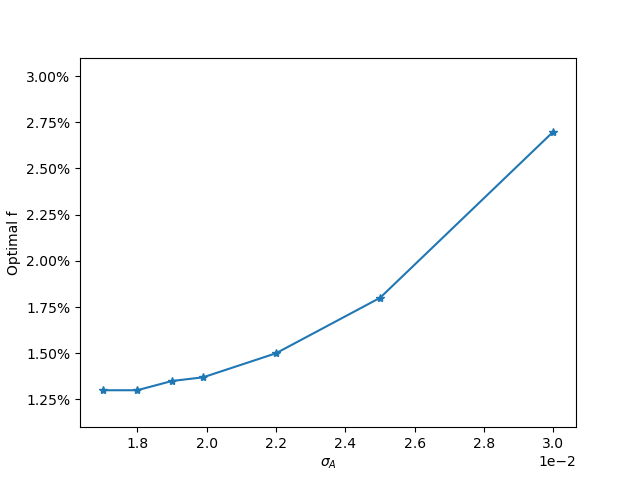}
		\end{minipage}
		\hspace{1em}
		\begin{minipage}{0.35\textwidth}
			\includegraphics[width=1.15\textwidth]{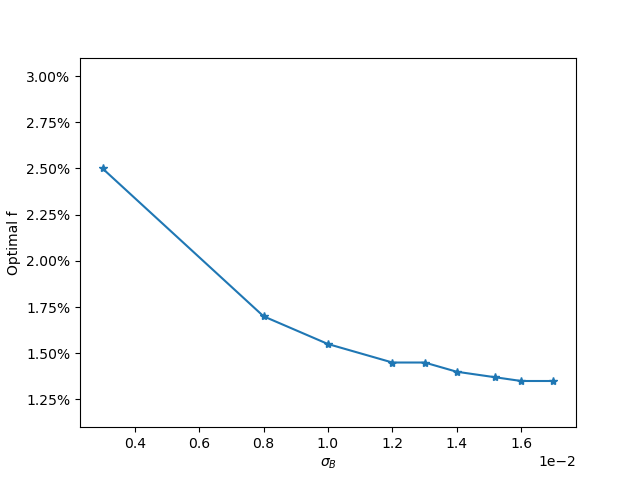}
		\end{minipage}
		\\
		\begin{minipage}{0.35\textwidth}
			\includegraphics[width=1.15\textwidth]{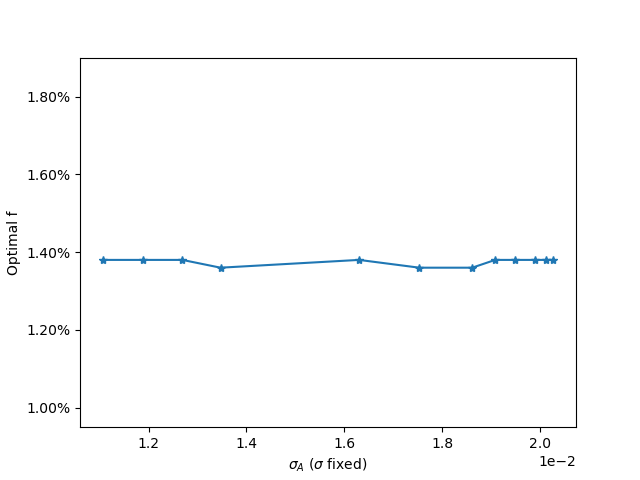}
		\end{minipage}
		\hspace{1em}
		\begin{minipage}{0.35\textwidth}
			\includegraphics[width=1.15\textwidth]{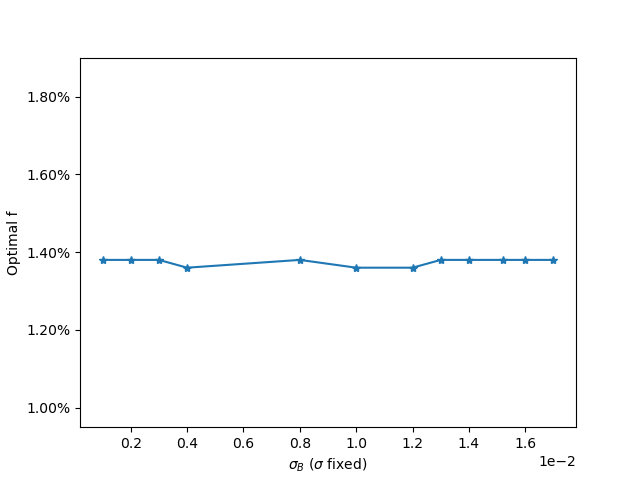}
		\end{minipage}
		\caption{Optimal unit trading fee $f$ maximizing the LP's utility with respect to $\mu_A$ (top left panel), $\mu_B$ (top right panel), $\sigma_A$ (middle left panel), $\sigma_B$ (middle right panel), $\sigma_A$ with $\sigma:=\sqrt{\sigma_A^2+\sigma_B^2-2\rho \sigma_A\sigma_B}$ fixed (bottom left panel), and $\sigma_B$ with $\sigma$ fixed (bottom right panel). The pricing function is CGMMM with $\eta=1/2$. The portfolio constraint is given by \eqref{eq:PortfolioConstraints}, so borrowing or short-selling is not allowed. Other model parameter values are given in Table \ref{tbl:parameters}.} 
		\label{fi:FeeVSReturn}
	\end{figure}

	\subsubsection{Optimal Pricing Function}
	Next, we plot in Figure \ref{fi:OptimalEta} the optimal $\eta$ with respect to $\mu_A$ (top left panel), $\mu_B$ (top bottom panel), $\sigma_A$ (bottom left panel), and $\sigma_B$ (bottom right panel). We plot $\eta$ only for the return parameter values such that the LP invests on DEX, because otherwise the value of $\eta$ is irrelevant.
	
	Recall that as long as \eqref{eq:FeeDominateIL} holds, the LP's fee revenue from the liquidity traders net the loss to the arbitrageurs is positive and is maximized when $\eta=1/2$. Thus, if we do not consider the opportunity cost, the LP would choose $\eta=1/2$. However, with the opportunity cost, the LP has to achieve a trade-off between setting $\eta=1/2$ so as to have a large fee revenue net the loss to the arbitrageurs and setting $\eta/(1-\eta)$ as close to the efficient allocation ratio as possible to minimize the opportunity cost. Because the efficient allocation ratio is increasing in $\mu_A$, the LP's optimal choice of $\eta$ is also increasing in $\mu_A$. This explains the top left panel of Figure \ref{fi:OptimalEta}. The other three panels of the figure can explained similarly.

	\begin{figure}
		\centering
		\begin{minipage}{0.35\textwidth}
			\includegraphics[width=1.15\textwidth]{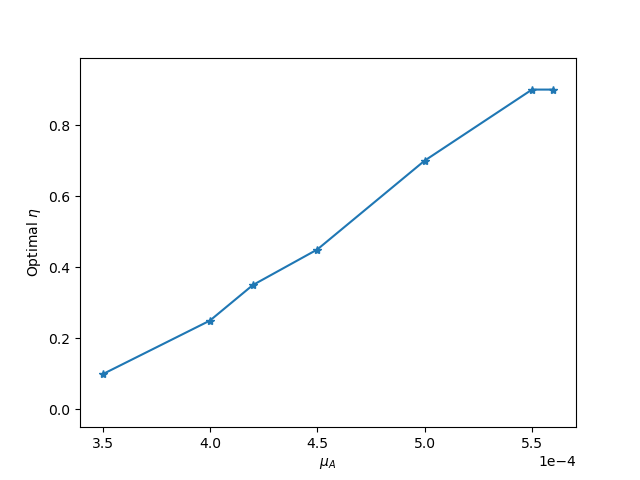}
		\end{minipage}
		\hspace{1em}
		\begin{minipage}{0.35\textwidth}
			\includegraphics[width=1.15\textwidth]{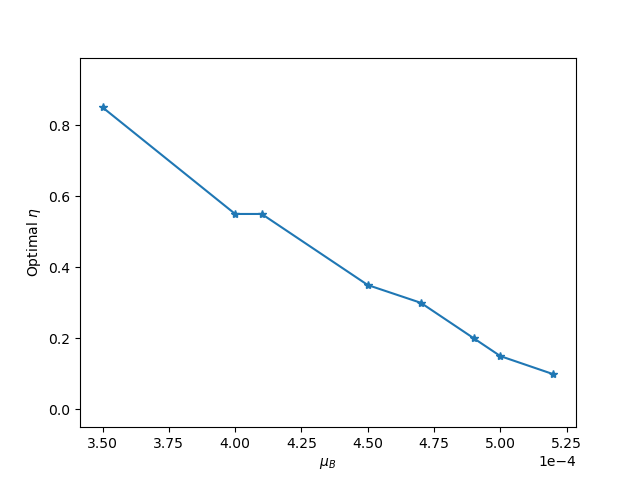}
		\end{minipage}
		\\
		\begin{minipage}{0.35\textwidth}
			\includegraphics[width=1.15\textwidth]{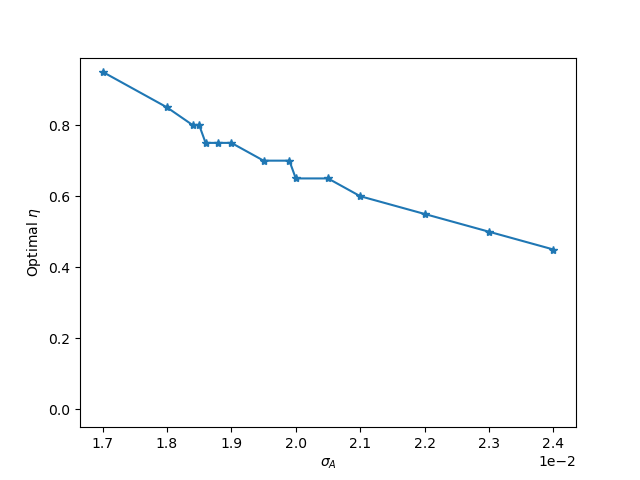}
		\end{minipage}
		\hspace{1em}
		\begin{minipage}{0.35\textwidth}
			\includegraphics[width=1.15\textwidth]{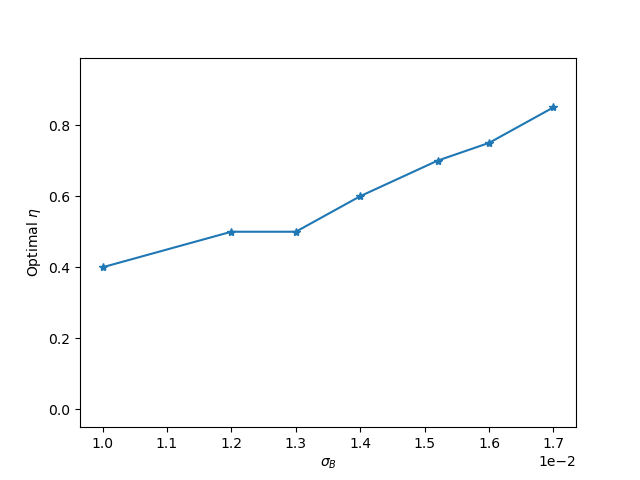}
		\end{minipage}
		\caption{Optimal pricing function parameter $\eta$ maximizing the LP's utility with respect to $\mu_A$ (top left panel), $\mu_B$ (top right panel), $\sigma_A$ (bottom left panel), and $\sigma_B$ (bottom right panel). The pricing function is CGMMM. The portfolio constraint is given by \eqref{eq:PortfolioConstraints}, so borrowing or short-selling is not allowed. Other model parameter values are given in Table \ref{tbl:parameters}.}\label{fi:OptimalEta}
		\begin{minipage}{0.35\textwidth}
			\includegraphics[width=1.15\textwidth]{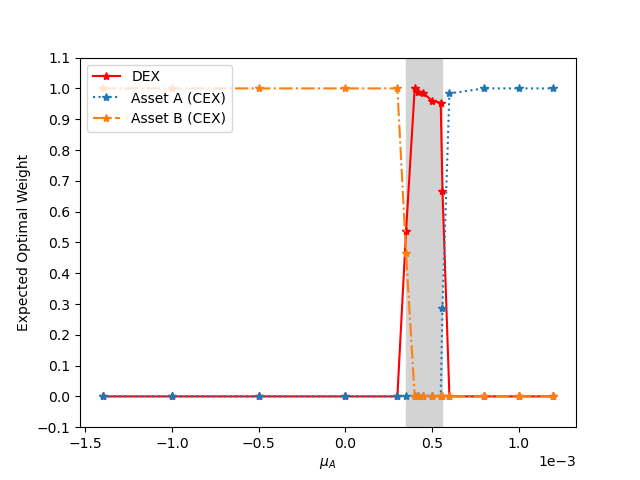}
		\end{minipage}
		\hspace{1em}
		\begin{minipage}{0.35\textwidth}
			\includegraphics[width=1.15\textwidth]{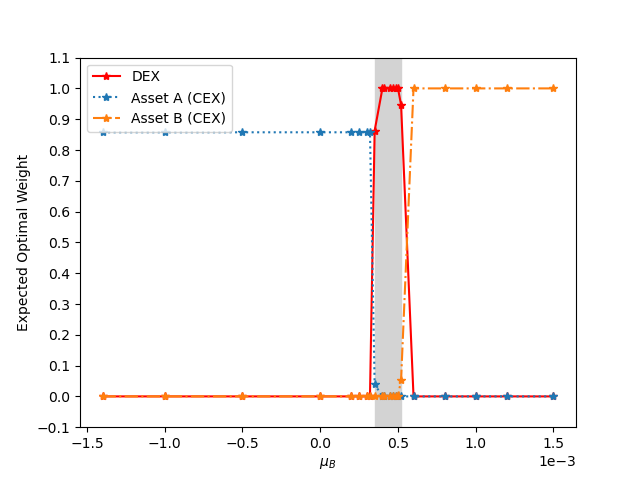}
		\end{minipage}
		\\
		\begin{minipage}{0.35\textwidth}
			\includegraphics[width=1.15\textwidth]{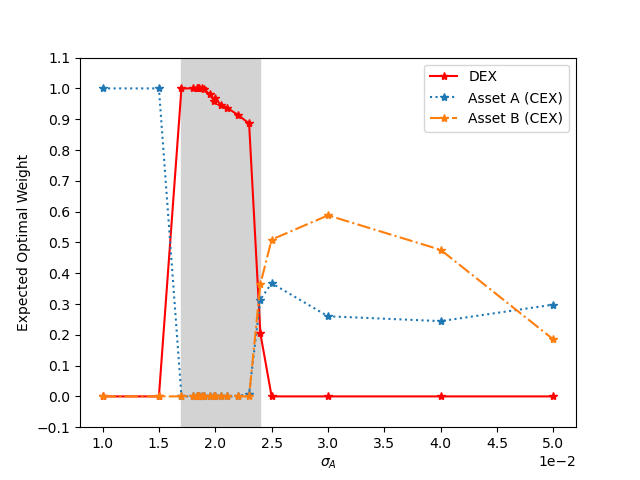}
		\end{minipage}
		\hspace{1em}
		\begin{minipage}{0.35\textwidth}
			\includegraphics[width=1.15\textwidth]{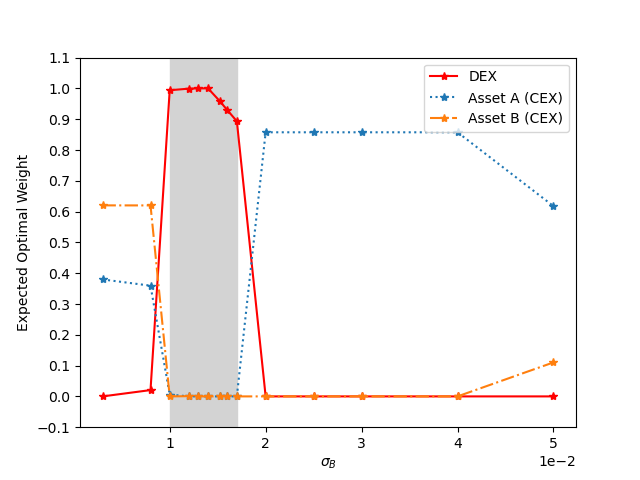}
		\end{minipage}
		\caption{LP's expected optimal percentage investment on DEX (solid line), in asset $A$ on CEX (dotted line), and in asset $B$ on CEX (dash-dotted line) with respect to $\mu_A$ (top left panel), $\mu_B$ (top right panel), $\sigma_A$ (bottom left panel), and $\sigma_B$ (bottom right panel) under the optimal pricing function. The pricing function is CGMMM with $\eta$ chosen optimally to maximize the LP's utility for each return parameter value. The portfolio constraint is given by \eqref{eq:PortfolioConstraints}, so borrowing or short-selling is not allowed. Other model parameter values are given in Table \ref{tbl:parameters}.}\label{fi:InvestmentVSReturnParametersOptimalEta}
	\end{figure}

	In Figure \ref{fi:InvestmentVSReturnParametersOptimalEta}, we plot the LP's optimal investment with respect to the return parameters $\mu_A$ (top left panel), $\mu_B$ (top bottom panel), $\sigma_A$ (bottom left panel), and $\sigma_B$ (bottom right panel), with $\eta$ chosen optimally for each parameter value. We can observe that compared to Figure \ref{fi:InvestmentVSReturnParameters} where $\eta$ is fixed at $1/2$, allowing the LP to choose $\eta$ increases her investment on DEX.

	\subsubsection{The Effect of $\sigma_I$ and $\alpha$}\label{subsubse:sigmaIAlpha}
	We also conducted robustness check by varying $\alpha$ and $\sigma_I$ while holding the other market parameters at their default values in Table \ref{tbl:parameters}. Because of space limit, we present the main conclusions below without including detailed numerical results.
			
	First, the previous conclusions regarding the optimal trading fee and pricing function still hold: (i) The optimal unit trading fee remains relatively insensitive to $\mu_A$, $\mu_B$, and individual asset volatilities while fixing the exchange-rate volatility $\sigma$; (ii) The opportunity cost is the dominant force of determining the optimal pricing function.

	Second, the optimal unit trading fee is decreasing in $\alpha$ and increasing in $\sigma_I$. This can be explained as follows:  Recall \eqref{eq:ILNetFeeApprox} and denote by $H(f):=f\expect\big[ |I-1|\big(\mathbf 1_{I<(1+f)^{-1}}+\mathbf 1_{I>(1+f)}\big) \big]$ and $L(f):=\frac{1}{2}\expect\big[R^B(\tilde R-1)^2\big(\mathbf 1_{\tilde R>(1+f)}+\mathbf 1_{\tilde R<1/(1+f)}\big)\big]$, which stand for the LP's expected fee revenue per unit arrival of liquidity traders and the LP's expected loss to the arbitrageurs, respectively. Then, the LP's net profit from the AMM as in \eqref{eq:ILNetFeeApprox} can be written as $\eta(1-\eta)\big[\alpha H(f) - L(f)\big]$, and the optimal unit trading fee $f$ is the one that maximizes this value because the opportunity loss does not depend on $f$. In our numerical setting, $H(f)$ is first increasing and then decreasing in $f$, correctly reflecting the trade-off that a higher $f$ leads to a higher fee revenue of the LP per liquidity trade and at the same time to a smaller number of trades. On the other hand, $L(f)$ is decreasing in $f$, because a larger $f$ discourages arbitrageurs and thus leads to a lower net loss of the LP to the arbitrageurs. In our numerical setting, $\alpha H(f) - L(f)$ is quasi-concave in $f$ and thus the optimal $f$ is solved by the first-order condition $\alpha H'(f)=L'(f)$, which stipulates that the marginal change in the LP's fee revenue from the liquidity traders is equal to the marginal change in his loss to the arbitrageurs. Note that both $\alpha H'(f)$ and $L'(f)$ are negative at the optimal $f$. As a result, when $\alpha$ becomes larger, $\alpha H'(f)$ becomes more negative, so the optimal $f$ must become smaller so as to equate $\alpha H'(f)$ and $L'(f)$. When $\sigma_I$ becomes larger, $H'(f)$ becomes larger in a reasonable range of $f$ in our numerical setting, reflecting the idea that a more heterogeneous belief among the liquidity traders implies more trades and thus higher LP's incremental fee revenue from increasing the unit trading fee. In consequence, the optimal $f$ must be larger in order to equate $\alpha H'(f)$ and $L'(f)$.

	Third, when $\alpha$ or $\sigma_I$ becomes larger, the optimal $\eta$ becomes closer to 0.5. Indeed, a larger $\alpha$ or a larger $\sigma_I$ increases $\alpha H(f)$, the LP's fee revenue from liquidity traders. Therefore, the LP's net profit from the AMM becomes more dominant to his opportunity cost. In consequence, the optimal $\eta$ would move closer to the value that maximizes the LP's net profit from the AMM, which is 0.5.

	\subsubsection{Short Sales on CEX}
	
	In the previous discussions, we did not allow short sales on CEX. Some crypto-assets have tradeable futures contracts on them, so one can effectively short-sell these assets using the futures contracts. The amount of assets that one can short-sell, however, is limited due to the margin requirement. In the following, we study whether short sales on CEX can reduce the opportunity cost of holding assets in the liquidity pool on DEX. To this end, we set the portfolio constraint as
	\begin{align}
		\Omega = \{\vecomega=(\omega^M,\omega^A,\omega^B)\mid \omega_A\ge -1,\omega^B\ge -1,\omega^M\ge 0, \omega_A+\omega_B+\omega_M\le 2\}.\label{eq:PortfolioWeakConstraints}
	\end{align}
	In other words, we allow the LP to short-sell each asset on CEX up to 100\% of her wealth. We also restrict that the total amount that the LP can invest on DEX and CEX is no more than 200\% of her wealth; i.e., the LP can borrow up to 100\% of her wealth at the risk-free rate.
	
	Fixing $\eta=1/2$, we plot in Figure \ref{fi:InvestmentVSReturnParametersShortSale} the LP's optimal investment on DEX and CEX with respect to the asset return parameters, assuming the portfolio constraint \eqref{eq:PortfolioWeakConstraints}, which allows for short sales on CEX. Compared to Figure \ref{fi:InvestmentVSReturnParameters}, we can observe that allowing for short sales increases the LP's investment on DEX in that the range of parameter values under which the LP is willing to invest on DEX becomes wider. This is because short sales on CEX can reduce the opportunity cost: The LP can take short positions in one of the assets on CEX so that her total holdings of the two assets become more efficient.

	\begin{figure}
		\centering
		\begin{minipage}{0.35\textwidth}
			\includegraphics[width=1.15\textwidth]{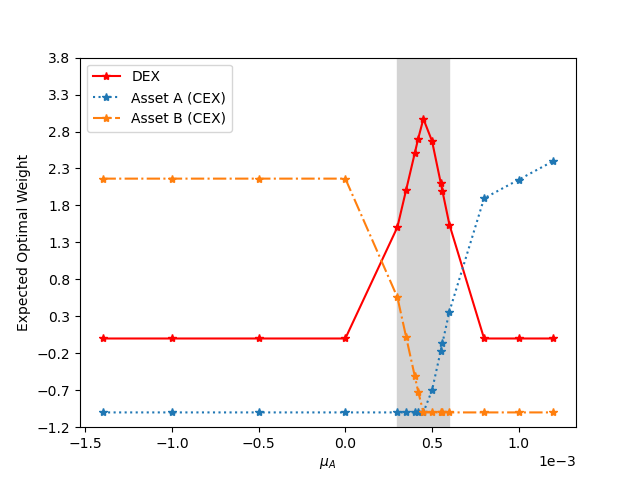}
		\end{minipage}
		\hspace{1em}
		\begin{minipage}{0.35\textwidth}
			\includegraphics[width=1.15\textwidth]{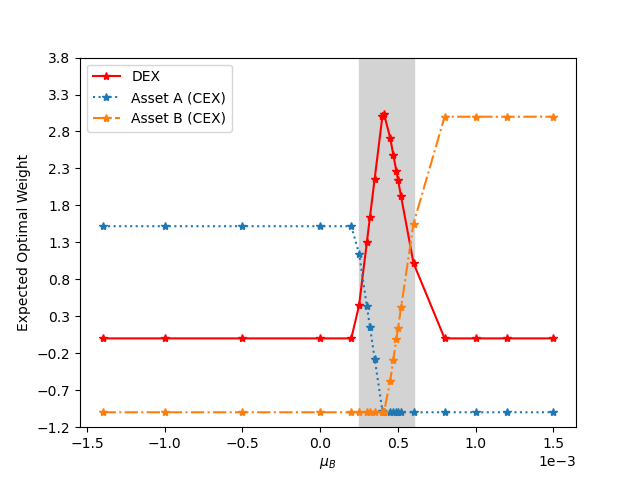}
		\end{minipage}
		\\
		\begin{minipage}{0.35\textwidth}
			\includegraphics[width=1.15\textwidth]{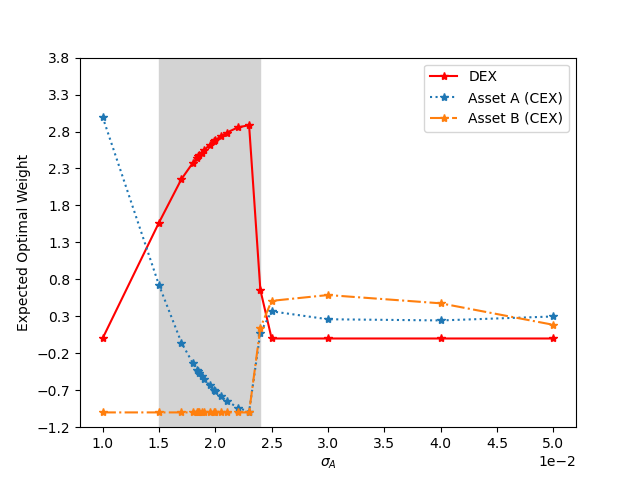}
		\end{minipage}
		\hspace{1em}
		\begin{minipage}{0.35\textwidth}
			\includegraphics[width=1.15\textwidth]{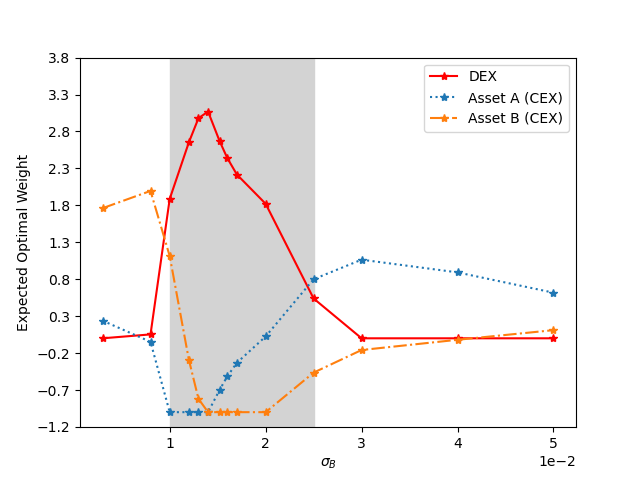}
		\end{minipage}
		\caption{LP's expected optimal percentage investment on DEX (solid line), in asset $A$ on CEX (dotted line), and in asset $B$ on CEX (dash-dotted line) with respect to $\mu_A$ (top left panel), $\mu_B$ (top right panel), $\sigma_A$ (bottom left panel), and $\sigma_B$ (bottom right panel). The pricing function is CGMMM with $\eta=1/2$. The portfolio constraint is given by \eqref{eq:PortfolioWeakConstraints}, so short sales on CEX are allowed. Other model parameter values are given in Table \ref{tbl:parameters}.}\label{fi:InvestmentVSReturnParametersShortSale}
	\end{figure}

	In Figure \ref{fi:OptimalEtaShortSale}, we plot the optimal $\eta$ with respect to the return parameters when short sales are allowed. Again, we only plot $\eta$ for those parameter values under which the LP invests on DEX, because otherwise the value of $\eta$ is irrelevant. We can observe the same pattern of the optimal $\eta$ with respect to the return parameters as in the case of no short sale (Figure \ref{fi:OptimalEta}). This shows that even if short sales are allowed, the effect of the opportunity cost still persists and the LP would choose $\eta$ to reduce the opportunity cost.

	\begin{figure}
		\centering
		\begin{minipage}{0.35\textwidth}
			\includegraphics[width=1.15\textwidth]{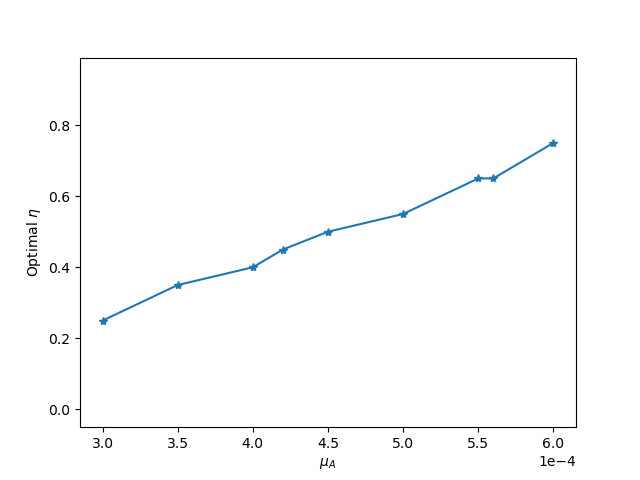}
		\end{minipage}
		\hspace{1em}
		\begin{minipage}{0.35\textwidth}
			\includegraphics[width=1.15\textwidth]{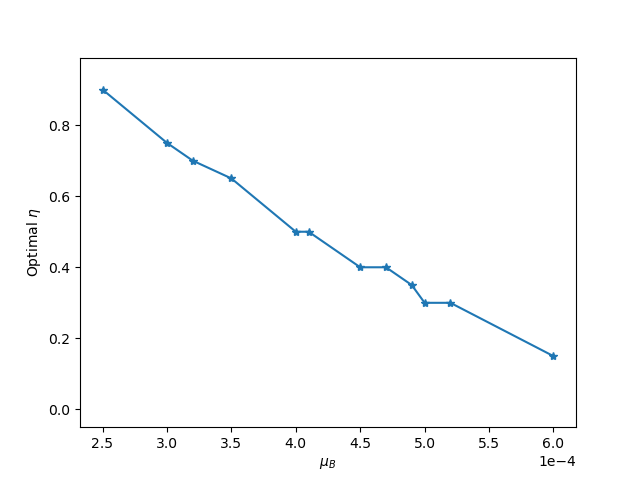}
		\end{minipage}
		\\
		\begin{minipage}{0.35\textwidth}
			\includegraphics[width=1.15\textwidth]{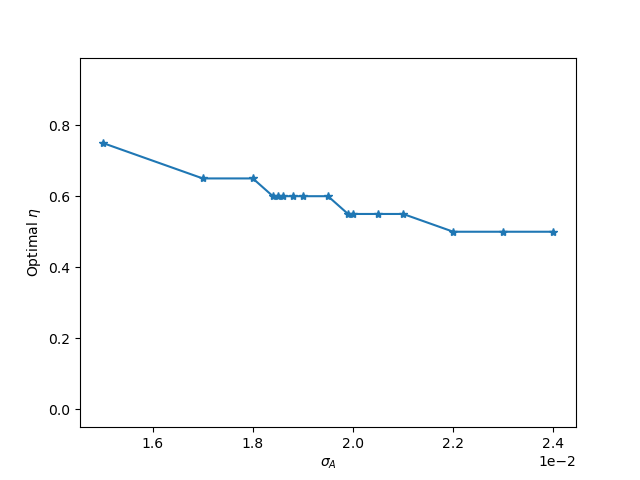}
		\end{minipage}
		\hspace{1em}
		\begin{minipage}{0.35\textwidth}
			\includegraphics[width=1.15\textwidth]{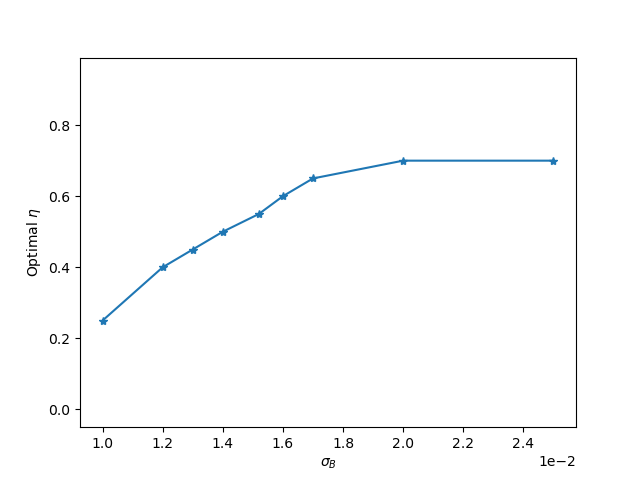}
		\end{minipage}
		\caption{Optimal pricing function parameter $\eta$ maximizing the LP's utility with respect to $\mu_A$ (top left panel), $\mu_B$ (top right panel), $\sigma_A$ (bottom left panel), and $\sigma_B$ (bottom right panel). The pricing function is CGMMM. The portfolio constraint is given by \eqref{eq:PortfolioWeakConstraints}, so short sales on CEX are allowed. Other model parameter values are given in Table \ref{tbl:parameters}.}\label{fi:OptimalEtaShortSale}
	\end{figure}
	
	Figure \ref{fi:InvestmentVSReturnParametersOptimalEtaShortSale} plots the LP's optimal investment on DEX and CEX with respect to the return parameters with $\eta$ chosen optimally and with short sales allowed. Compared to Figure \ref{fi:InvestmentVSReturnParametersShortSale}, we find that allowing the LP to choose $\eta$ increases her investment on DEX. We also observe that when the LP invests on DEX, she short-sells both assets $A$ and $B$ on CEX. The short sale is mainly for the purpose of funding the investment on DEX rather than reducing the opportunity cost. The opportunity cost is reduced by setting the value of $\eta$ optimally.
	
	\begin{figure}
		\centering
		\begin{minipage}{0.35\textwidth}
			\includegraphics[width=1.15\textwidth]{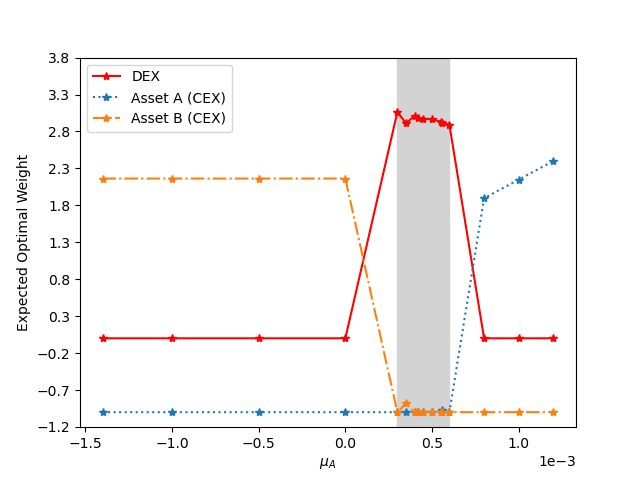}
		\end{minipage}
		\hspace{1em}
		\begin{minipage}{0.35\textwidth}
			\includegraphics[width=1.15\textwidth]{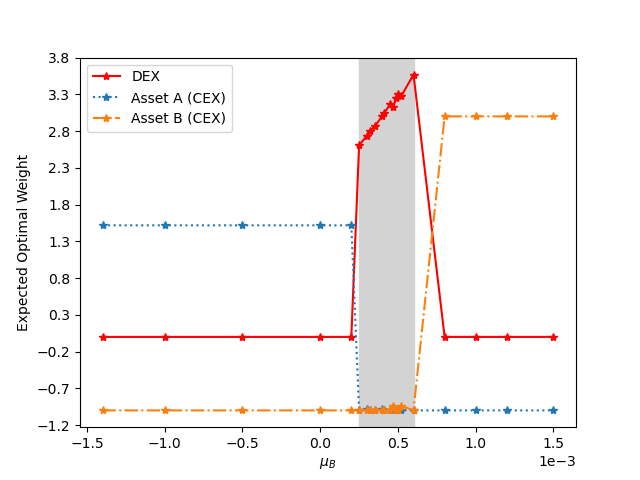}
		\end{minipage}
		\\
		\begin{minipage}{0.35\textwidth}
			\includegraphics[width=1.15\textwidth]{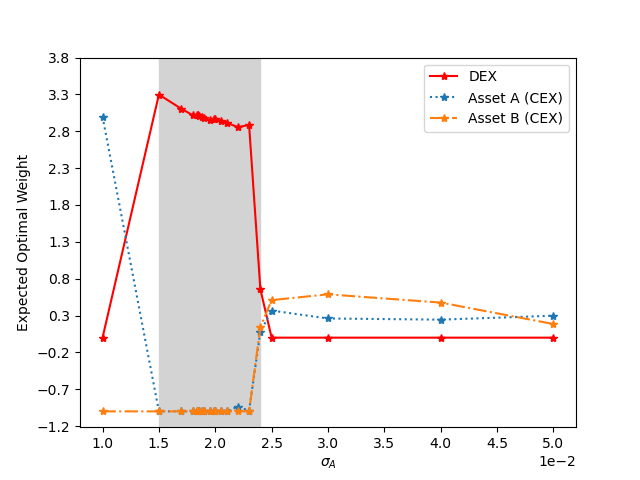}
		\end{minipage}
		\hspace{1em}
		\begin{minipage}{0.35\textwidth}
			\includegraphics[width=1.15\textwidth]{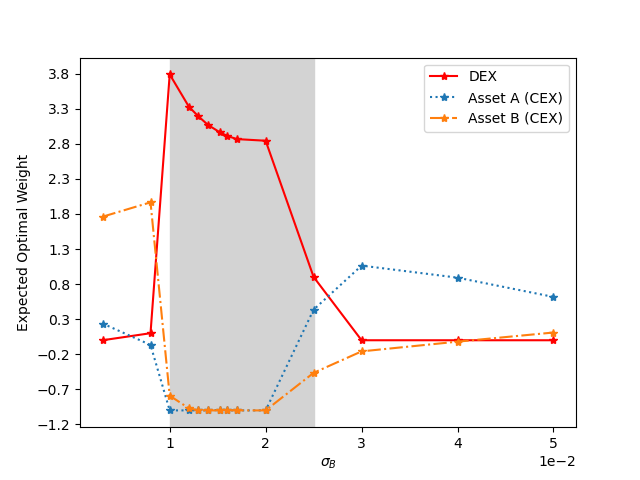}
		\end{minipage}
		\caption{LP's optimal percentage investment on DEX (solid line), in asset $A$ on CEX (dotted line), and in asset $B$ on CEX (dash-dotted line) with respect to $\mu_A$ (top left panel), $\mu_B$ (top right panel), $\sigma_A$ (bottom left panel), and $\sigma_B$ (bottom right panel) under the optimal pricing function. The pricing function is CGMMM with $\eta$ chosen optimally to maximize the LP's utility for each return parameter value. The portfolio constraint is given by \eqref{eq:PortfolioWeakConstraints}, so short sales on CEX are allowed. Other model parameter values are given in Table \ref{tbl:parameters}.}\label{fi:InvestmentVSReturnParametersOptimalEtaShortSale}
	\end{figure}

	\subsection{Empirical Test}\label{sect:implications}
	
	In this subsection, we conduct an empirical test of the following model implication derived in Section \ref{subsect:fee}: the return parameters affect the optimal choice of unit trading fee $f$ only through the fundamental exchange rate volatility $\sigma$. 
	
	We use data in 110 AMM pools on Uniswap v3 and asset price data on Binance\footnote{See \url{https://data.binance.vision/}.} from July 1, 2022 to June 30, 2023. Note that Uniswap v2 uses a fixed fee and thus does not provide comparable cross-sectional variation in fee choices necessary for our empirical analysis. By contrast, Uniswap v3 provides observable variation in fee tiers across pools. Our model uses the protocol in Uniswap v2, and Uniswap v3 differs from Uniswap v2 mainly in a feature called concentrated liquidity.\footnote{See \url{https://thegraph.com/hosted-service/subgraph/uniswap/uniswap-v3} for info about Uniswap v3 pools.} Therefore, our theoretical analysis and model implications do not automatically extend to the case of Uniswap v3. However, we believe that our main model implications remain qualitatively the same in Uniswap v3, because liquidity provision in Uniswap v3, in which the LP can choose to provide liquidity in a certain range of the exchange rate, is locally similar to liquidity provision in Uniswap v2. Therefore, the empirical analysis based on Uniswap v3 data below still sheds some light on our model implications.
	
	For each AMM pool $j$, we collect its unit trading fee $f_j$ from Uniswap v3 and estimate the volatilities of the two assets in the pool $\sigma_{A,j}$ and $\sigma_{B,j}$ and their correlation $\rho_{j}$ using the Binance asset price data. Then we estimate the volatility of the fundamental exchange rate for the pool $j$ as
		$
		\sigma_{j} = \sqrt{\sigma_{A,j}^2-2\rho_{j}\sigma_{A,j}\sigma_{B,j}+\sigma_{B,j}^2 }.
		$
	We then run a linear regression between the unit trading fee $f$ and the corresponding fundamental exchange rate volatility $\sigma$, asset correlation $\rho$, and asset $B$'s volatility $\sigma_B$ across the 110 AMM pools, and the results are reported in Table \ref{Tbl:em_result}. At the 5\% significance level, the unit trading fee is positively correlated to the fundamental exchange rate volatility $\sigma$ but not related to the asset correlation $\rho$ or asset $B$'s volatility $\sigma_B$. This observation is consistent with our model implication in Section \ref{subsect:fee}.
	
	\begin{table}
		\centering
		\footnotesize
		\begin{tabular}{@{}ccccc@{}}
			\toprule
			& Intercept & $\rho$  & $\sigma_B$ &  $\sigma$ \\
			\midrule
			coefficient & $0.3454$ & $-0.0040$  & $-3.2821$ & $14.5096$ \\
			$p$-value & $0.1602$ &	$0.9873$ &	$0.5054$	& $0.0420$\\
			\bottomrule
		\end{tabular}
		\caption{Linear regression of the unit trading fee $f$ with respect to the correlation of assets $A$ and $B$'s prices ($\rho$), asset $B$'s volatility $\sigma_B$, and the volatility of the fundamental exchange rate of asset $A$ over asset $B$ ($\sigma$) across 110 AMM pools on Uniswap v3.}\label{Tbl:em_result}
	\end{table}

	\section{Conclusion}\label{sect:conclusion}
	In this paper, we have proposed a model for the LP's liquidity provision on DEX. Our model features LP's investment allocations on both DEX and CEX and her consumption, as well as the shocks to LP's return on DEX from both the potential arrival of liquidity traders and the fundamental asset price changes. We have formulated the model as a stochastic control problem, established the dynamic programming equation for the problem, and proved the existence and uniqueness of the solution.
	
	Through numerical experiments, we have found that there exists an optimal level of unit trading fee and optimal pricing function in the class of CGMMM functions that maximize the LP's utility. We have shown that the optimal unit trading fee depends on the return distributions of the assets in the liquidity pool through the volatility of their fundamental exchange rate: The higher the volatility is, the higher the unit trading fee the LP would choose. This finding is supported by an empirical study using liquidity pools on Uniswap v3. We have also found that the liquidity provider would choose the pricing function to minimize the opportunity cost that arises from the inefficient asset allocation in the liquidity pool. 
	
	We assume constant market parameters in our model, so the comparative statics of the optimal fee and pricing function with respect to these parameters suggest how AMM parameters should respond to unexpected, persistent changes in the market conditions. In practice, however, market conditions change over time, e.g., the asset volatility is not constant. Therefore, although our finding is consistent with the recent evolution of AMM protocols to allow for variable and dynamic fees,\footnote{Uniswap v3 introduced multiple fee tiers for the same token pair, and Uniswap v4 further allows dynamic fees and customized pool logic through hooks.} a comprehensive analysis of the optimal design of dynamic fee and pricing function calls for a model that accounts for dynamic market conditions and frictions such as implementation costs, gas costs, governance or coordination frictions, and manipulation concerns. This can be an interesting topic for future research.

	\appendix
	
	\section{Proofs}\label{proof}
	
	\begin{pfof}{Proposition \ref{le:TradersProblem}}
		We can rewrite the objective function of \eqref{eq:TraderProblem} as
		\begin{align*}
			by^B\left[(a/b)\big(1+f\mathbf 1_{D^{A}/y^A<0}\big)(D^{A}/y^A)(y^A/y^B) + (1+f\mathbf 1_{D^{B}/y^B<0}) (D^{B}/y^B)\right].
		\end{align*}
		The second and third constraints of \eqref{eq:TraderProblem} can be written as $D^i/y^i< 1$, $i\in \{A,B\}$. Thanks to Assumption \ref{as:PricingFunction}-(iii), the first constraint of \eqref{eq:TraderProblem} is equivalent to the following constraint:
		\begin{align*}
			F(y^A/y^B,1) = F\big((y^A/y^B)(1-D^A/y^A),1-D^B/y^B\big).
		\end{align*}
		Therefore, \eqref{eq:TraderProblem} is equivalent to the following optimization problem:
		\begin{align}\label{eq:TraderProblemReduced}
			\begin{array}{rl}
				\underset{d^A,d^B}{\max} & \alpha \beta\big(1+f\mathbf 1_{d^A<0}\big)d^A + (1+f\mathbf 1_{d^B<0}) d^B \\				\text{subject to} & F(\beta,1) = F\big(\beta(1-d^A),1-d^B\big),\\
				& d^A < 1,~~d^B < 1,
			\end{array}
		\end{align}
		where $\alpha$ stands for $a/b$, the investor's belief of the exchange rate, and $\beta$ stands for $y^A/y^B$, the deposit ratio in the liquidity pool before the trade. When the optimal solution of \eqref{eq:TraderProblemReduced} exists, we denote it as $\big(\varphi^A(\alpha,\beta),\varphi^B(\alpha,\beta)\big)$, where we highlight the dependence of the optimal solution on $\alpha$ and $\beta$. Then, the optimal solution of \eqref{eq:TraderProblem}, denote as $(\hat D^{A},\hat D^{B})$, is $\hat D^i = \varphi^i(a/b,y^A/y^B)y^i$, $i\in \{A,B\}$.
		
		Next, we solve \eqref{eq:TraderProblemReduced}. By Assumption \ref{as:PricingFunction}-(i) and the first constraint of \eqref{eq:TraderProblemReduced}, for any feasible $(d_A,d_B)$, it is either the case in which $d_A=d_B=0$, or the case in which $d_A>0$ and $d_B<0$, or the case in which $d_A<0$ and $d_B>0$. Thus, we can consider the maximization of the objective function of \eqref{eq:TraderProblemReduced} in these three cases and identify the case that yields the largest value.
		For the first case, the objective function value is 0.
		
		For the second case, the optimization problem becomes
		\begin{align}\label{eq:TraderProblemReducedCase2}
			\begin{array}{rl}
				\underset{d^A,d^B}{\max} & \alpha \beta d^A +  (1+f)d^B \\ 
				\text{subject to} & F(\beta,1) = F\big(\beta(1-d^A),1-d^B\big),\\ 
				& 0<d^A<1,d^B < 0.
			\end{array}
		\end{align}
		Recall the discussions following Assumption \ref{as:PricingFunction}. There exists a maximal interval $[0,K)$ for some $K\in(0,1]$ on which $h(z)$ is uniquely defined by $F(\beta,1)=F(\beta(1-z),1-h(z))$, and $h$ is strictly decreasing and continuously differentiable on $[0,K)$ with $h(0)=0$. In addition, if $K<1$, then $\lim_{z\uparrow K}h(z)=-\infty$. Thus, problem \eqref{eq:TraderProblemReducedCase2} is equivalent to
		\begin{align}\label{eq:TraderProblemReducedCase2Equi1}
			\begin{array}{rl}
				\underset{d^A}{\max} & \alpha \beta d^A + (1+f) h(d^A) \\
				\text{subject to} & 0<d^A < K.
			\end{array}
		\end{align}
		Differentiating the equation defining $h$, we derive
		\begin{align*}
			h'(z) = -\frac{\beta F_x(\beta(1-z),1-h(z))}{ F_y(\beta(1-z),1-h(z))} = -\beta G\left(\frac{\beta(1-z)}{1-h(z)}\right),\quad z\in [0,K),
		\end{align*}
		where the second equality is due to Assumption \ref{as:PricingFunction}-(ii). Because $h(z)$ is strictly decreasing in $z$ and $G(u)$ is strictly decreasing in $u$, we conclude that $h'$ is strictly decreasing on $[0,K)$. Moreover, $h'(0) = -\beta G(\beta)$ and $\lim_{z\uparrow K}h'(z) = -\beta \lim_{u\downarrow 0} G(u) = -\infty$, where the first equality in the above calculation of the limit is the case because it is either the case $K=1$ or the case in which $K<1$ and $\lim_{z\uparrow K}h(z) = -\infty$, and the second equality is due to Assumption \ref{as:PricingFunction}-(ii). Now, denoting the objective function of \eqref{eq:TraderProblemReducedCase2Equi1} as $V(d^A)$, we have
			$
			V'(d^A) = \alpha \beta  + (1+f) h'(d^A).
			$
		Thus, if $G(\beta)\ge \alpha/(1+f)$, $V'(d^A)<0$ for all $d^A\in (0,K)$, showing that $V(d^A)<V(0)=0$ for all $d^A\in (0,K)$. If $G(\beta)< \alpha/(1+f)$, there exists a unique maximizer $\hat d^A\in (0,K)$, which solves the equation $V'(\hat d^A)=0$, namely,
		\begin{align}\label{eq:TraderProblemReducedCase2OptimaldA}
			\frac{1- \hat d^A}{1-h(\hat d^A)} = \frac{1}{\beta} G^{-1}\big(\alpha/(1+f)\big),
		\end{align}
		and the optimal value $V(\hat d^A)>V(0)=0$. The optimal $d^B$ for \eqref{eq:TraderProblemReducedCase2}, denoted as $\hat d^B$, is then $h(\hat d^A)$. Combining \eqref{eq:TraderProblemReducedCase2OptimaldA} and the definition of $h$, we immediately conclude that $(\hat d^A,\hat d^B)$ is actually the unique solution of \eqref{eq:EquationsOptimalTradingForAssetA}. Consequently, we can solve $\hat d^B$ from
		\begin{align}
			F\left(G^{-1}\big(\alpha/(1+f)\big) (1-\hat d^B),1-\hat d^B\right)=F(\beta ,1), \label{eq:TraderProblemReducedCase2OptimaldB}
		\end{align}
		which immediately implies that $\hat d^B$ is continuous in $(\alpha,\beta)$ with $G(\beta)< \alpha/(1+f)$, so $\hat d^A$ is also continuous. Moreover, because $F(x,y)$ is strictly increasing in $x$ and $y$ and $G$ is strictly decreasing, we immediately conclude that $\hat d^B$ is strictly decreasing in both $\alpha$ and $\beta$ with $G(\beta)< \alpha/(1+f)$ and
		\begin{align*}
			\lim_{\alpha\downarrow G(\beta)(1+f)}\hat d^B = 0,\quad  \lim_{\beta\downarrow G^{-1}(\alpha/(1+f))}\hat d^B = 0.
		\end{align*}
		On the other hand, dividing both sides of \eqref{eq:TraderProblemReducedCase2OptimaldB} by $\beta$ and recalling Assumption \ref{as:PricingFunction}-(iii) and \eqref{eq:TraderProblemReducedCase2OptimaldA}, we conclude that $\hat d^A$ solves
			$
			F\left(1-\hat d^A,(1-\hat d^A)/G^{-1}\big(\alpha/(1+f)\big)\right)=F(1 ,1/\beta),
			$
		which implies that $\hat d^A$ is strictly increasing in both $\alpha$ and $\beta$ with $G(\beta)< \alpha/(1+f)$ and
			$
			\lim_{\alpha\downarrow G(\beta)(1+f)}\hat d^A = 0
			,$ $\lim_{\beta\downarrow G^{-1}(\alpha/(1+f))}\hat d^A = 0.
			$
		
		Similarly, for the case in which $d_A<0$ and $d_B>0$, the optimization problem \eqref{eq:TraderProblemReduced} becomes
		\begin{align}\label{eq:TraderProblemReducedCase3}
			\begin{array}{rl}
				\underset{d^A,d^B}{\max} & \alpha \beta (1+f) d^A +  d^B \\ 
				\text{subject to} & F(\beta,1) = F\big(\beta(1-d^A),1-d^B\big),\\ 
				& d^A<0,0<d^B < 1.
			\end{array}
		\end{align}
		When $G(\beta)\le \alpha(1+f)$, the objective function in \eqref{eq:TraderProblemReducedCase3} is strictly less than 0 for any feasible $(d^A,d^B)$. When $G(\beta)>\alpha(1+f)$, \eqref{eq:TraderProblemReducedCase3} admits a unique optimal solution $(\hat d^A,\hat d^B)$, which solves
		\begin{align*}
			\frac{1- \hat d^A}{1-\hat d^B} = \frac{1}{\beta} G^{-1}\big(\alpha (1+f)\big),\quad  F\big(\beta(1-\hat d^A),1-\hat d^B\big)=F(\beta,1),
		\end{align*}
		and the optimal value is strictly larger than 0. Moreover, (i) both $\hat d^A$ and $\hat d^B$ are continuous in $(\alpha,\beta)$ with $G(\beta)>\alpha(1+f)$, (ii) $\hat d^A$ is strictly increasing and $\hat d^B$ is strictly decreasing in both $\alpha$ and $\beta$ with $G(\beta)>\alpha(1+f)$, and (iii)
			$
			\lim_{\alpha\uparrow G(\beta)/(1+f)}\hat d^i=0,~~ \lim_{\beta\uparrow G^{-1}(\alpha(1+f))}\hat d^i=0,~~ i\in \{A,B\}.
			$
		
		Combining the above discussions of the three cases, we complete the proof.
	\end{pfof}
	
	\begin{pfof}{Theorem \ref{th:OptimalSolution}}
		We first show that $\mathbb{T}_k$ defines a mapping from ${\cal X}$ to itself. Fix any $k\in \{0,1,\dots, N-1\}$ and any bounded, measurable function $J\in {\cal X}$. Denote
		\begin{align*}
			H_k(\vecomega;J):&=\big(R^p_{k+1}\big)^{1-\gamma}J(S_{k+1},Z_{k+1})+\frac{\delta^{N-k-1}}{1-\delta^N} U\big(R^p_{k+1}\big).
		\end{align*}
		Recall \eqref{eq:StateZDynamics}, \eqref{eq:AMMGrossReturnFunction}, \eqref{eq:StateSDynamics}, and that $L^M$ and $L^S$ are continuous, we conclude that $\mathbb{E}\left[H_k(\vecomega;J)|S_k=s,Z_k=z\right]$ is measurable in $(\vecomega,s,z)$. Moreover, note that $R^p_{k+1}\le R^M_{k+1}+R^A_{k+1}+R^B_{k+1}+R_f$ and thus
		\begin{align*}
			U(R^p_{k+1})&\le U(R^M_{k+1}+R^A_{k+1}+R^B_{k+1}+R_f) = 4^{1-\gamma}U\left(\frac{1}{4}(R^M_{k+1}+R^A_{k+1}+R^B_{k+1}+R_f)\right) - U(4)\\
			&\le 4^{1-\gamma}\left|U\left(\frac{1}{4}(R^M_{k+1}+R^A_{k+1}+R^B_{k+1}+R_f)\right)\right| +|U(4)|,
		\end{align*}
		where the first inequality is due to the monotonicity of $U$.
		By Assumption \ref{as:GrowthCondition}-(i),
		\begin{align*}
			\left|U\left(\frac{1}{4}(R^M_{k+1}+R^A_{k+1}+R^B_{k+1}+R_f)\right)\right|
		\end{align*}
		is integrable for each fixed $(S_k,Z_k)=(s,z)$, so by the dominated convergence theorem, the boundedness of $J$, and the continuity of $H_k(\vecomega;J)$ in $\vecomega$, we conclude that $\mathbb{E}\left[H_k(\vecomega;J)|S_k=s,Z_k=z\right]$ is continuous in $\vecomega$ for each fixed $(s,z)$. Because $\Omega$ is separable, we conclude that
		\begin{align*}
			(\mathbb{T}_k J)(s,z)=\sup_{\vecomega\in \Omega}    \mathbb{E}\left[H_k(\vecomega;J)|S_k=s,Z_k=z\right]
		\end{align*}
		is measurable in $(s,z)$. Moreover, by Assumption \ref{as:GrowthCondition}-(i), $(\mathbb{T}_k J)(s,z)$ is bounded in $(s,z)$.

		Next, we consider the case with $\gamma=1$. It is straightforward to see that $\mathbb{T}_k$ is a contraction mapping on ${\cal X}$ with the Lipschitz constant $\delta$ for every fixed $k\in \{0,1,\dots, N-1\}$. Moreover,
		\begin{align*}
			(\hat{\mathbb{T}}_0 J)(s,z)= (\mathbb{T}_0 J)(s,z) + \sup_{c_0\in (0,1)}\left\{U(c_0) +  \frac{\delta^N}{1-\delta^N} U(1-c_0)\right\},
		\end{align*}
		which yields that $\hat{\mathbb{T}}_0 \mathbb{T}_1\cdots \mathbb{T}_{N-1}$ is a contraction mapping on ${\cal X}$ with the Lipschitz constant $\delta^N$. Thus, $\hat{\mathbb{T}}_0 \mathbb{T}_1\cdots \mathbb{T}_{N-1}$ admits a unique fixed point $v_0^*$. Moreover, the standard verification theorem shows that for every $k\in\{0,1,\dots, N-1\}$, $v^*_k$ is the LP's optimal expected utility for her problem \eqref{eq:LPProblem}, starting from unit wealth at time $k$, where $v_k^*:=\mathbb{T}_k\cdots \mathbb{T}_{N-1}v_0^*$, $k=1,\dots, N-1$. Furthermore, the optimal percentage of wealth consumed at time 0 is
		\begin{align*}
			c_0^*&=\underset{c_0\in (0,1)}{\mathrm{argmax}}\left\{U(c_0) + (1-c_0)^{1-\gamma} (\mathbb{T}_0 v_1^*)(s,z) + \frac{\delta^N}{1-\delta^N} U(1-c_0)\right\} = \delta^{N}
		\end{align*}
		and the optimal percentage of wealth invested in AMM and CEX is
		\begin{align*}
			\vecomega_k^* & = \underset{\vecomega_k\in \Omega}{\mathrm{argmax}}\;\mathbb{E}\Big[ \big(R^p_{k+1}\big)^{1-\gamma}J(S_{k+1},Z_{k+1})+\frac{\delta^{N-k-1}}{1-\delta^N} U\big(R^p_{k+1}\big) \mid S_k=s,Z_k=z\Big]\\
			&  = \underset{\vecomega_k\in \Omega}{\mathrm{argmax}}\;\mathbb{E}\Big[ U\big(R^p_{k+1}\big) \mid S_k=s,Z_k=z\Big].
		\end{align*}
		Because $U$ is strictly concave and $R^p_{k+1}$ is linear in $\vecomega_k$ and $\Omega$ is compact, $\omega_k^*$ uniquely exists.
		
		Next, we consider the case with $\gamma \in (0,1)$. Following the same proof as for the case of $\gamma=1$ and recalling Assumption \ref{as:GrowthCondition}, we can show that for every $k=0,1,\dots, N-1$, $\mathbb{S}_k$ is a contraction mapping on ${\cal X}_+$ with the Lipschitz constant $\delta \bar{R}^{1-\gamma}$. For each $a\ge 0$, denote
		\begin{align}
			H(c,a):=c^{1-\gamma} + a (1-c)^{1-\gamma}.\label{eq:ConsumptionAggregator}
		\end{align}
		Straightforward calculation yields that
		$
		h(a):= \sup_{c\in (0,1)} H(c,a) = (1+a^{1/\gamma})^\gamma
		$
		and the unique maximizer of $H(c,a)$ in $c\in (0,1)$ is $(1+a^{1/\gamma})^{-1}$ when $a>0$. Moreover, $h$ is strictly increasing and $h'(a)=\left(\frac{a^{1/\gamma}}{1+a^{1/\gamma}}\right)^{1-\gamma}\le 1$, $a\ge 0$. As a result,
		$
		\hat{\mathbb{S}}_0 \mathbb{S}_1\cdots \mathbb{S}_{N-1} = h \mathbb{S}_0 \mathbb{S}_1\cdots \mathbb{S}_{N-1}
		$
		is a contraction mapping on ${\cal X}_+$ with Lipschitz constant $(\delta \bar{R}^{1-\gamma})^N$. Thus, $\hat{\mathbb{S}}_0 \mathbb{S}_1\dots \mathbb{S}_{N-1}$ has a unique fixed point $\tilde v_0^*$ in ${\cal X}_+$ and starting from any $\tilde v_0$, the iterative sequence $(\hat{\mathbb{S}}_0 \mathbb{S}_1\dots \mathbb{S}_{N-1})^n \tilde v_0$, $n=1,\dots,$ converges to $\tilde v^*$ exponentially. Furthermore, it is straightforward to see that $\hat{\mathbb{S}}_0 \mathbb{S}_1\cdots \mathbb{S}_{N-1}$ is monotone, so
		\begin{align*}
			\tilde v^*_0=\hat{\mathbb{S}}_0 \mathbb{S}_1\cdots \mathbb{S}_{N-1}\tilde v^*_0\ge \hat{\mathbb{S}}_0 \mathbb{S}_1\cdots \mathbb{S}_{N-1}0=1.
		\end{align*}
		In consequence, $\tilde v^*_{N-1}:=\mathbb{S}_{N-1}\tilde v_0^*\ge \mathbb{S}_{N-1}1>0$, and mathematical induction yields that for every $k=1,\dots, N-2$,
		$
		\tilde v^*_k:= \mathbb{S}_k \mathbb{S}_{k+1}\cdots \mathbb{S}_{N-1} \tilde v_0^*= \mathbb{S}_k\tilde v^*_{k+1}>0.
		$
		The standard verification theorem shows that for every $k\in\{0,1,\dots, N-1\}$, $v^*_k$ as in \eqref{eq:OptimalValueGammaNeq1} is the LP's optimal expected utility for her problem \eqref{eq:LPProblem}, starting from unit wealth at time $k$. Furthermore, the optimal percentage of wealth consumed at time 0 is the maximizer of $H(c,\mathbb{S}_0\tilde v^*_1)$ in $c\in (0,1)$. Because $\tilde v^*_1>0$ and thus $\mathbb{S}_0\tilde v^*_1>0$, the maximizer uniquely exists and is
		$
		c_0^*=(1+\left(\mathbb{S}_0 \tilde v_1^*(s,z)\right)^{1/\gamma})^{-1}.
		$
		The optimal percentage of wealth invested in AMM and CEX is
		\begin{align*}
			\vecomega_k^* & = \underset{\vecomega_k\in \Omega}{\mathrm{argmax}}\;\mathbb{E}\Big[ \big(R^p_{k+1}\big)^{1-\gamma}\tilde v^*_{k+1}(S_{k+1},Z_{k+1}) \mid S_k=s,Z_k=z\Big]\\
			& = \underset{\vecomega_k\in \Omega}{\mathrm{argmax}}\;\mathbb{E}\Big[ U(R^p_{k+1}\big)\tilde v^*_{k+1}(S_{k+1},Z_{k+1}) \mid S_k=s,Z_k=z\Big].
		\end{align*}
		
		Finally, we consider the case with $\gamma>1$. Following the proof in the case of $\gamma=1$, we can show that for every $k=0,1,\dots, N-1$, $\mathbb{S}_k$ maps ${\cal X}_+$ to ${\cal X}_+$. Moreover,
		it is straightforward to see that $\mathbb{S}_k$ is concave and monotone. It is also easy to observe that
		\begin{align*}
			h^-(a):= \inf_{c\in (0,1)}H(c,a) = (1+a^{1/\gamma})^\gamma
		\end{align*}
		is concave and monotone in $a\ge 0$. As a result,
		$
		\hat{\mathbb{S}}_0 \mathbb{S}_1\cdots \mathbb{S}_{N-1} = h^- \mathbb{S}_0 \mathbb{S}_1\cdots \mathbb{S}_{N-1}
		$
		is a concave, monotone mapping on ${\cal X}_+$. It is straightforward to see that
		$
		\hat{\mathbb{S}}_0 \mathbb{S}_1\cdots \mathbb{S}_{N-1}0= h^-(0)= 1>0.
		$
		On the other hand, for any constant $L>0$, we have
		\begin{align}
			(\mathbb{S}_k L)(s,z) &= \delta L \inf_{\vecomega_k\in \Omega} \mathbb{E}\Big[ \big(R^p_{k+1}\big)^{1-\gamma} \mid S_k=s,Z_k=z\Big]\notag\\
			& = \delta L \left(\sup_{\vecomega_k\in \Omega} \left(\mathbb{E}\Big[ \big(R^p_{k+1}\big)^{1-\gamma} \mid S_k=s,Z_k=z\Big]\right)^{1/(1-\gamma)}\right)^{1-\gamma}\notag\\
			& = \delta L \left(\sup_{\vecomega_k\in \Omega} U^{-1}\left(\mathbb{E}\Big[ U\big(R^p_{k+1}\big) \mid S_k=s,Z_k=z\Big]\right)\right)^{1-\gamma}\notag\\
			&\le  \delta L \bar R^{1-\gamma}.\label{eq:GammaLarger1OperatorUpperbound}
		\end{align}
		Consequently,
		\begin{align*}
			(\hat{\mathbb{S}}_0 \mathbb{S}_1\cdots \mathbb{S}_{N-1}L)(s,z) = h^-\left((\mathbb{S}_0 \mathbb{S}_1\cdots \mathbb{S}_{N-1})L(s,z)\right)\le h^-\big((\delta \bar R^{1-\gamma})^N L\big)<L
		\end{align*}
		for sufficiently large $L$, where the first inequality is due to \eqref{eq:GammaLarger1OperatorUpperbound} and the monotonicity of $\mathbb{S}_k$, and the second inequality is because $\delta \bar R^{1-\gamma}<1$ as imposed in Assumption \ref{as:GrowthCondition}-(ii). Then, by Theorem 3.1 in \citet{Du1990:FixedPoints}, $\hat{\mathbb{S}}_0 \mathbb{S}_1\dots \mathbb{S}_{N-1}$ has a unique fixed point $\tilde v_0^*$ on ${\cal X}_+$. Moreover, for any $\tilde v_0\in {\cal X}_+$, the iterative sequence $(\hat{\mathbb{S}}_0 \mathbb{S}_1\dots \mathbb{S}_{N-1})^n \tilde v_0$, $n=1,2,\dots$ converges to $\tilde v_0^*$ exponentially.  Following the same proof as in the case of $\gamma\in (0,1)$, we can show that $\tilde v^*_0\ge 1$ and that $v^*_k$ as in \eqref{eq:OptimalValueGammaNeq1} is the LP's optimal expected utility for her problem \eqref{eq:LPProblem}, starting from unit wealth at time $k$, $k=0,\dots, N-1$. Furthermore, the optimal percentage of wealth consumed at time 0 and the optimal percentage of wealth invested in AMM and the CEX uniquely exist and are as given in \eqref{eq:OptimalActionGammaNeq1}.
	\end{pfof}

	\bibliography{LongTitles,BibFile}
	
\end{document}